\documentclass[sigconf]{acmart}
\AtBeginDocument{%
  }

\setcopyright{acmlicensed}
\copyrightyear{2018}
\acmYear{2018}
\acmDOI{XXXXXXX.XXXXXXX}
\acmConference[Conference acronym 'XX]{Make sure to enter the correct
conference title from your rights confirmation email}{June 03--05,
2018}{Woodstock, NY}
\acmISBN{978-1-4503-XXXX-X/2018/06}

\usepackage[ruled,linesnumbered,vlined]{algorithm2e}
\usepackage{algorithmic}
\usepackage{amsmath}
\usepackage{mathtools}
\usepackage{amsthm}
\usepackage{newtxmath}
\usepackage{amsfonts}
\usepackage{subfig}
\usepackage{enumitem}
\usepackage{balance}
\usepackage{multirow}
\usepackage[normalem]{ulem}
\usepackage{tabularx}
\usepackage{pdfpages}
\usepackage{makecell}

\begin{document}

\title{GRACE: A Dynamic Coreset Selection Framework for Large Language Model Optimization}


\author{Tianhao Tang}
\affiliation{%
  \institution{CSE, HKUST}
  \city{Hong Kong SAR}
  \country{}
}
\email{ttangae@cse.ust.hk}

\author{Haoyang Li}
\affiliation{%
  \institution{Computing, PolyU}
  \city{Hong Kong SAR}
  \country{}
}
\email{haoyang-comp.li@polyu.edu.hk}

\author{Lei Chen}
\affiliation{%
  \institution{DSA, HKUST (GZ)\&HKUST}
  \city{Guangzhou, China}
  \country{}
}
\email{leichen@cse.ust.hk}

\renewcommand{\shortauthors}{Tang et al.}

\begin{abstract}
  Large Language Models (LLMs) have demonstrated remarkable capabilities in natural language understanding and generation. However, their immense number of parameters and complex transformer-based architectures result in significant resource demands and computational complexity during training, making it challenging to optimize them efficiently on large datasets. To reduce training costs while preserving performance, researchers have investigated coreset selection techniques, which aim to identify small, representative subsets of the entire training dataset to accelerate LLM training. However, existing coreset selection methods fail to adapt to the dynamic nature of LLM training and often struggle with scalability for models of this size. To address these limitations, we propose a graph-guided adaptive and dynamic coreset selection framework for LLMs, namely GRACE. GRACE dynamically constructs and updates coresets by combining representation diversity with gradient-based importance metrics, ensuring both informativeness and efficiency. To mitigate the computational cost of frequent updates, GRACE leverages a $k$-NN graph-based propagation mechanism and selectively updates scores and embeddings, adapting to evolving training dynamics. Extensive experiments on three benchmarks demonstrate that GRACE significantly improves training efficiency and downstream performance across diverse LLMs and tasks.
\end{abstract}

%  \begin{CCSXML}
%    <ccs2012>
%    <concept>
%    <concept_id>00000000.0000000.0000000</concept_id>
%    <concept_desc>Do Not Use This Code, Generate the Correct Terms for Your Paper</concept_desc>
%    <concept_significance>500</concept_significance>
%    </concept>
%    <concept>
%    <concept_id>00000000.00000000.00000000</concept_id>
%    <concept_desc>Do Not Use This Code, Generate the Correct Terms for Your Paper</concept_desc>
%    <concept_significance>300</concept_significance>
%    </concept>
%    <concept>
%    <concept_id>00000000.00000000.00000000</concept_id>
%    <concept_desc>Do Not Use This Code, Generate the Correct Terms for Your Paper</concept_desc>
%    <concept_significance>100</concept_significance>
%    </concept>
%    <concept>
%    <concept_id>00000000.00000000.00000000</concept_id>
%    <concept_desc>Do Not Use This Code, Generate the Correct Terms for Your Paper</concept_desc>
%    <concept_significance>100</concept_significance>
%    </concept>
%    </ccs2012>
%  \end{CCSXML}

%  \ccsdesc[500]{Do Not Use This Code~Generate the Correct Terms for Your Paper}
%  \ccsdesc[300]{Do Not Use This Code~Generate the Correct Terms for Your Paper}
%  \ccsdesc{Do Not Use This Code~Generate the Correct Terms for Your Paper}
%  \ccsdesc[100]{Do Not Use This Code~Generate the Correct Terms for Your Paper}

\keywords{Large Language Model, coreset, data selection, adaptive training}

%  \received{20 February 2007}
%  \received[revised]{12 March 2009}
%  \received[accepted]{5 June 2009}

\maketitle

\section{Introduction}\label{sec:intro}

Large language models (LLMs) \cite{touvron2023llama2openfoundation,dubey2024llama3herdmodels,li2024survey,openai2023gpt4}, such as GPT~\cite{openai2023gpt4}, LLaMA~\cite{touvron2023llama2openfoundation,dubey2024llama3herdmodels}, and DeepSeek~\cite{deepseekai2024deepseekv3}, have demonstrated remarkable abilities in understanding and generating human-readable language, emerging as novel approaches for interacting with database systems. Their strengths in tasks such as natural language query processing~\cite{10.14778/3681954.3682003,zhou2025cracksql}, knowledge extraction~\cite{10.14778/3626292.3626294,10.14778/3659437.3659461}, and data summarization~\cite{10.14778/3659437.3659461} make them particularly valuable for assisting in managing database systems~\cite{10.1145/3709652,zhao2024llmdbdemo,10.14778/3675034.3675043}, such as constructing and simplifying complex queries~\cite{10.14778/3696435.3696440,ren2024purple} or facilitating efficient data exploration~\cite{10.14778/3659437.3659461,10597732}. In general, LLMs are primarily built on the Transformer architecture~\cite{vaswani2017attention}, featuring billions of parameters (e.g., Llama-3 with 7$\sim$65 billion parameters) and trained on massive datasets. This design enables them to capture intricate data patterns, following an autoregressive training scheme that predicts the next token based on previous context. The training objective is to maximize the likelihood of correct token sequence.

Despite their impressive performance, LLMs present significant challenges regarding the computational demands of training.
Specifically, given $N$ training samples with an average sequence length $T$, and an LLM with $L$ layers and hidden dimension $d$, the time complexity of training this LLM on the entire dataset once (i.e., one epoch) is $O(NL(dT^2+d^2T))$. This complexity is determined by both the size of the data and the architecture of the model, specifically the sequence length, hidden size, and total number of parameters~\cite{hoffmann2022training,li2024survey}. Therefore, considering the large number of parameters in models such as LLMs, it is challenging and resource-intensive to optimize LLMs on new corpora or downstream task datasets. Recently, parameter-efficient fine-tuning (PEFT) methods~\cite{dettmers2023qlora,he2022towards,houlsby2019parameter,hu2022lora,li2021prefix} such as LoRA~\cite{hu2022lora} have been proposed. These methods introduce a small fraction of trainable parameters while freezing the rest, achieving performance comparable to full fine-tuning. However, effective task and domain adaptation with PEFT methods still requires processing millions of tokens and substantial GPU resources.

To accelerate model training and fine-tuning, recent research has proposed selecting a small and representative coreset from the entire training dataset and optimizing the model on this coreset instead of the full dataset~\cite{chaiGoodCoreDataeffectiveDataefficient2023,chaiEfficientCoresetSelection2023,10.14778/3712221.3712249,10.14778/3561261.3561267}. The coreset is designed to preserve the statistical and structural properties of the original data, ensuring that training on it achieves comparable model performance. By prioritizing the most informative and diverse samples, coreset selection reduces redundancy, improves training speed, and lowers resource consumption. The coreset selection technique has been highly successful in accelerating the training of neural networks, including graph neural networks~\cite{li20242,li2024fight,liCamelManagingData2022} and convolutional neural networks~\cite{chaiEfficientCoresetSelection2023,killamsettyGRADMATCHGradientMatching2021,mirzasoleimanCoresetsDataefficientTraining2020,zhengCoveragecentricCoresetSelection2023}. Depending on the technique, existing approaches can be categorized into three types: \textit{data-property-based}, \textit{uncertainty-based}, and \textit{gradient-based} approaches.

Specifically, data-property-based approaches \cite{chaiGoodCoreDataeffectiveDataefficient2023,chenMaybeOnly052023,liCamelManagingData2022,tirumalaD4ImprovingLLM2023} select coresets based on intrinsic properties of the data, such as feature similarity or diversity. While these methods are simple and efficient for selection, they overlook information about model-specific learning difficulties for the target models. Also, uncertainty-based methods \cite{anknerPerplexedPerplexityPerplexityBased2024,baiMultiAgentCollaborativeData2024,marionWhenLessMore2023,yangSmallToLargeS2LScalable2024} first obtain the prediction probabilities or loss values from the target model on each training sample and construct the coreset by prioritizing uncertain samples, such as those with incorrectly predicted labels. However, these selection methods require a full evaluation of samples using the target model, which is computationally intensive, especially for LLMs. Also, the uncertainty metrics do not directly correlate with the gradient optimization procedure, and thus the selected samples are not sufficiently beneficial for training. Further, gradient-based methods \cite{joaquinIn2CoreLeveragingInfluence2024,nguyenMemoryefficientTrainingLLMs2024,xiaLESSSelectingInfluential2024,yuMATESModelAwareData2024} utilize gradient information, which explicitly encodes how each sample affects model updates during optimization. These methods construct coresets by prioritizing samples with larger gradient magnitudes or by solving a gradient-matching problem to approximate the gradients of the full dataset using a smaller subset. However, due to the high dimensionality and significant computational cost of calculating gradients for LLMs with billions of parameters, these approaches often encounter efficiency and scalability challenges when applied to LLMs~\cite{hoffmann2022training,yin2024compute}. More importantly, the above methods select a static coreset before model training, failing to adapt to the evolving training dynamics of LLMs. Although some approaches dynamically update coresets at fixed intervals using intermediate metrics like sample losses or gradients \cite{yuDiversifyConquerDiversityCentric2024,yuMATESModelAwareData2024}, they cannot accurately align with real-time model dynamics, often resulting in unnecessary updates and excessive computational overhead.

In this paper, we propose an adaptive and dynamic coreset selection framework for LLMs. We construct dynamic coresets based on gradient matching during training and adaptively update them to efficiently capture important training information. However, directly constructing dynamic coresets presents several technical challenges, as outlined below:
\begin{itemize}[leftmargin=20pt]
  \item[\textbf{C1:}] While the gradient matching problem has been studied in conventional deep learning models (e.g., CNNs), its application to LLMs introduces unique challenges. The vast scale and complex transformer architecture of LLMs result in prohibitive computational costs and a lack of theoretical guarantees for coreset selection in transformer-based LLMs.

  \item[\textbf{C2:}] Dynamic coreset selection for LLMs inherently requires multiple evaluations of data samples, which results in high computational costs, even for forward passes through the model. Therefore, designing an effective metric and algorithm to select coresets for LLMs with billions of parameters is critical for accelerating LLM training.

  \item[\textbf{C3:}] Beyond simply updating the coreset based on a manually defined epoch interval, determining the optimal moment to update the coreset is nontrivial. A naive approach that repeatedly solves the matching problem leads to excessive computational overhead, which impedes efficiency.
\end{itemize}

\begin{sloppypar}
  To address these challenges, we propose GRACE, a \textbf{GR}aph-guided \textbf{A}daptive and Dynamic \textbf{C}oreset s\textbf{E}lection framework that selects coresets for transformer-based LLMs and adaptively updates them using graph-based approximations.
  First, we conduct a theoretical analysis of gradient calculations in LLMs based on transformer architectures. Building on this analysis, we formalize the coreset selection problem, incorporating both feature-coverage representation scores and gradient-related importance scores to ensure diversity and informativeness in the selected coreset. Since the problem is NP-hard, we design an efficient greedy algorithm with a provable approximation ratio. Secondly, to adaptively update the coreset during training, we propose an adaptive checking strategy to periodically check for value changes in gradient-related scores during training. Additionally, since recalculating all dynamic indicators is expensive, we propose to exploit data similarity properties by constructing a $k$-NN graph based on the data representations. When an update is required, instead of recalculating scores or representations for all samples, we leverage the historical indicators saved during training and selectively recompute indicators for a subset of the full data, propagating the update to similar samples through the graph structure. This enables an efficient way to update dynamic indicators with evolving LLMs and ensures the coreset remains effective.
\end{sloppypar}

To summarize, our contributions are as follows:
\begin{itemize}[leftmargin=12pt]
  \item We propose a dynamic and adaptive coreset selection framework, GRACE, for LLM training and fine-tuning. GRACE dynamically constructs and updates coresets based on the training dynamics, with an adaptive mechanism to determine update moments.

  \item  We provide a theoretical analysis of the gradient matching problem for transformer-based LLMs. This analysis enables the formulation of a tractable coreset selection objective that balances feature coverage and gradient-based importance, offering a foundation for efficient coreset construction in LLM training.

  \item To address the inefficiency of frequent full recomputation, GRACE employs a graph-guided update mechanism. The framework constructs a $k$-NN graph and uses training history to model data similarity, allowing selective recomputation of scores and embeddings, as well as efficient approximation of graph updates. An adaptive checking strategy identifies when updates are necessary, thereby further reducing computational overhead.

  \item Extensive experiments on the MathInstruct, BioInstruct, and DialogSum benchmarks demonstrate that GRACE significantly improves training efficiency and downstream performance across diverse tasks.
\end{itemize}

\section{Preliminary and Related Works}\label{sec:related_work}
In this section, we first introduce the preliminaries of LLMs and then discuss existing coreset selection approaches. Important notations used in this paper are summarized in Table~\ref{tab:notation}.
\subsection{Large Language Model and Training}\label{subsec:llm}
\subsubsection{Large Language Models}
Large language models (LLMs)~\cite{li2024survey}, such as GPT~\cite{openai2023gpt4},  LLaMA~\cite{touvron2023llama2openfoundation,dubey2024llama3herdmodels}, and DeepSeek~\cite{deepseekai2024deepseekv3},
are trained on extensive datasets (e.g., corpora).
These models have demonstrated exceptional abilities in text understanding and analysis, which have been applied to a wide range of tasks, such as database tuning~\cite{10.1145/3709652}, query optimization~\cite{10.14778/3681954.3681960}, and data integration~\cite{10.14778/3659437.3659461,10.14778/3574245.3574258}.
Besides the extensive training data, the success of LLMs also relies on the transformer architecture.
LLMs are composed of multiple transformer blocks~\cite{vaswani2017attention} with billions of parameters, which enable LLMs to capture long-range dependencies within input sequences and autoregressively generate output sequences.

Given any natural language input, it is first processed by a tokenizer to break the sentence into a sequence $X = (x_1, x_2, \ldots, x_T)$ of length $T$, where each token represents a word or sub-word. Each token $x_i$ is then mapped to an embedding $\mathbf{x}_i$, forming the initial hidden representation $\mathbf{X}=[\mathbf{x}_1, \mathbf{x}_2, \ldots, \mathbf{x}_T] \in \mathbb{R}^{T\times d}$, where $d$ is the hidden dimension of the LLM.
The transformer processes the embeddings through $L$ layers of transformer blocks. Each transformer block  consists of a multi-head self-attention (MHSA) mechanism module followed by a feed-forward network.
Specifically, given the input  $\mathbf{H}^{(l-1)}$  to the transformer block at layer $l$ and $\mathbf{H}^{(0)}=\mathbf{X}$,
the MHSA mechanism applies self-attention across $n_{h}$ attention heads.
The self-attention operation for each head $i \in \{1,\cdots, n_h\}$ is computed as:
\begin{align}\label{eqdef:attention}
  &\mathbf{Q}^{(l)}_i = \mathbf{H}^{(l-1)}\mathbf{W}^{(l)}_{Q_i}, \\  &\mathbf{K}^{(l)}_i = \mathbf{H}^{(l-1)}\mathbf{W}^{(l)}_{K_i}, \\  &\mathbf{V}^{(l)}_i = \mathbf{H}^{(l-1)}\mathbf{W}^{(l)}_{V_i}, \\
  &\mathbf{Z}^{(l)}_i = \operatorname{softmax}\!\Biggl(\frac{\mathbf{Q}^{(l)}_i (\mathbf{K}^{(l)}_i)^\top}{\sqrt{d_k}}\Biggr)\,\mathbf{V}^{(l)}_i
\end{align}
where $\mathbf{Q}^{(l)}_i \in \mathbb{R}^{T\times d_k}$, $\mathbf{K}^{(l)}_i \in \mathbb{R}^{T\times d_k}$,
and $\mathbf{V}^{(l)}_i \in \mathbb{R}^{T\times d_v}$ are the learned Query matrix, Key matrix, and Value matrix.
Also,
$\mathbf{W}^{(l)}_{Q_i}$, $\mathbf{W}^{(l)}_{K_i}$,  and $\mathbf{W}^{(l)}_{V_i}$ are learnable weight matrices that project the input $\mathbf{H}^{(l-1)}_i$ into query, key, and value spaces, respectively.
The outputs of all attention heads are then concatenated and projected with an output weight $\mathbf{W}^{(l)}_O$ to produce the final MHSA output:
\begin{align}\label{defeq:multihead}
  \mathbf{Z}^{(l)}=\texttt{CONCAT}\left[\mathbf{Z}^{(l)}_1, ..., \mathbf{Z}^{(l)}_{n_h}\right]\mathbf{W}^{(l)}_O
\end{align}
And this is followed by a feed-forward network:
\begin{align}\label{defeq:ffn}
  \mathcal{F}(\mathbf{Z}^{(l)}) = \mathbf{W}^{(l)}_2\,\sigma\!\Bigl(\mathbf{Z}^{(l)}\mathbf{W}^{(l)}_1 + \mathbf{b}^{(l)}_1\Bigr) + \mathbf{b}^{(l)}_2
\end{align}
where \(\mathbf{W}^{(l)}_1\) and \(\mathbf{W}^{(l)}_2\) are weight matrices, \(\mathbf{b}^{(l)}_1\) and \(\mathbf{b}^{(l)}_2\) are bias vectors, and \(\sigma(\cdot)\) denotes the activation function (e.g., \texttt{ReLU}).
Combining these components, the hidden representation of the sequence is updated at each layer $l = 1, 2, \ldots, L$ as follows:
\begin{equation*}
  \mathbf{H}^{(l+1)} = \mathbf{H}^{(l)} + \mathcal{F}(\mathbf{H}^{(l)} + \mathcal{A}(\mathbf{H}^{(l)}))
\end{equation*}
\noindent For simplicity, we omit the normalization operation and positional encoding in the transformer block.

After passing through $L$ layers, the final hidden state $\mathbf{H}^{(L)}$ is then projected to the token vocabulary space by a final projection matrix $\mathbf{W}_h$, resulting in the predicted outputs $\mathbf{\hat{O}}=\mathbf{H}^{(L)}\mathbf{W}_h\in\mathbb{R}^{T\times n_{vocab}}$.  The token-wise probabilities are then obtained by applying a softmax function row-wise.
Specifically, unlike autoregressive inference that generates tokens sequentially, the training procedure is fully parallel: the entire input sequence is available and processed at once. This allows the model to compute all token-level predictions simultaneously and enables efficient loss and gradient computation with respect to the full output logits, as follows:
\begin{equation*}
  \hat{P}_{\boldsymbol{\theta}}(X)=\text{softmax}(\mathbf{\hat{O}}) \in [0,1]^{T\times n_{vocab}}
\end{equation*}
\noindent where $n_{vocal}$ is the vocabulary size and $\boldsymbol{\theta}$ is the parameter of the large language model.

\subsubsection{Large Language Model Training}
The training loss leverages an autoregressive paradigm to maximize the correct next-token prediction. The loss function of learning a sequence is defined as:
\begin{equation}\label{eq:loss}
  \mathcal{L}(X, \boldsymbol{\theta}) = -\frac{1}{T}\sum_{i=1}^{T}\log{\hat{P}_{\boldsymbol{\theta}}(x_i|x_{<i})}
\end{equation}
\noindent where $x_i=X[i]$ is the $i$-th token of $X$, and $x_{<i}=X[1:i-1]$ denotes the tokens from the first to the $(i-1)$-th tokens for $X$.
Specifically, $\hat{P}_{\boldsymbol{\theta}}(x_i|x_{<i})$ can be seen as the $i$-th row of $\hat{P}_{\boldsymbol{\theta}}(X)$.
This loss function is equivalent to the average of cross-entropy losses that maximize the likelihood of correct next-token predictions. The training process employs the mini-batch stochastic gradient descent (SGD) algorithm~\cite{bottou2012stochastic} that uniformly samples mini-batches from the entire dataset.
Given the training data $\mathcal{D}_{train}$, the LLM $\mathcal{M}_\theta$ can be optimized as follows:
\begin{equation*}
  \boldsymbol{\theta}_{t+1} = \boldsymbol{\theta}_t - \eta \frac{1}{|\mathcal{D}_{train}|} \sum_{X_i \in \mathcal{D}_{train}} \nabla_{\boldsymbol{\theta}} \mathcal{L}(X_i, \boldsymbol{\theta}_t)
\end{equation*}
\noindent where the gradient $\nabla_\theta \mathcal{L}(X_i, \boldsymbol{\theta}_t)$ denotes the gradient of the sequence $X_i$ at training iteration $t$.

\begin{table}
  \small
  \caption{Important notations.}
  \begin{tabularx}{\linewidth}{p{0.18\linewidth}|X}
    \hline
    \textbf{Notation} & \textbf{Description} \\
    \hline
    $X, x_i, T$ & Input sequence, $i$-th token, sequence length\\
    \hline
    $\mathbf{X}, \mathbf{x}_1$ &  Initial
    hidden states, $i$-th token's embedding \\
    \hline
    $\mathbf{H}^{(l)}_i$ & Hidden states of $X_i$ after $l$-th layer \\
    \hline
    $\bar{\mathbf{H}}^{(l)}_i$ & Mean hidden states of $X_i$ along sequence \\
    \hline
    $\mathbf{Z}^{(l)}, \mathbf{\hat{O}}$ & MHSA output at $l$-th layer, predicted logits \\ \hline
    $\hat{P}_{\boldsymbol{\theta}}(X)$ & Predicted probabilities of sample $X$\\
    \hline
    $\hat{P}_{\boldsymbol{\theta}}(x_i|x_{<i})$ & Predicted probabilities of $i$-th output token \\
    \hline
    $\mathcal{D}_{train}$ &
    Training dataset \\
    \hline
    $\mathcal{D}_{test}$ & Test dataset \\
    \hline
    $\mathcal{D}_{core}$ & Selected coreset \\
    \hline
    $\mathcal{M}_{\boldsymbol{\theta}}, L$ & LLM with parameters $\boldsymbol{\theta}$, LLM layer size  \\
    \hline
    $\mathcal{L}(X_i,\boldsymbol{\theta})$ & Prediction loss of sequence $X_i$ \\
    \hline
    $\nabla_{\boldsymbol{\theta}}\mathcal{L}(X_i,\boldsymbol{\theta})$ & Gradient w.r.t model parameters\\
    \hline
    $\nabla_{\hat{\mathbf{O}}}\mathcal{L}(X_i,\boldsymbol{\theta})$ & Gradient w.r.t logits outputs \\
    \hline
    $G$ & Graph built from training data \\
    \hline
    $R(\cdot), I(\cdot)$ & Representation score, importance score \\
    \hline
    $\lambda, \delta$ & Balance coefficient, update check threshold\\
    \hline
    $b$ & Selection budget in size\\
    \hline
    $\mathcal{T}$ & Total number of training steps \\
    \hline
    $\Delta^I,\Delta^\mathbf{H}$ & Change of importance score and embedding \\
    \hline
    $I^{new}(X_i)$ & Importance score after updates \\
    \hline
    $\bar{\mathbf{H}}^{new}(X_i)$ & Mean hidden states after updates \\
    \hline
  \end{tabularx}
  \label{tab:notation}
\end{table}

\subsection{Coreset selection for LLM Training}

Due to the large parameter size of LLMs,
it is highly time-consuming and resource-intensive to optimize the LLMs~\cite{isikScalingLawsDownstream2024}, especially on large datasets.
Recently, coreset selection approaches~\cite{chaiGoodCoreDataeffectiveDataefficient2023,chaiEfficientCoresetSelection2023,10.14778/3712221.3712249,10.14778/3574245.3574261,li20242,li2024fight,liCamelManagingData2022,10.1145/3677139,10.14778/3561261.3561267}
have been proposed to select a weighted subset of the dataset that approximates the statistical properties of the full dataset.
In this way, the model optimized on the coreset achieves performance comparable to the model optimized on the full dataset.
By training models on carefully selected coreset, the training time can be significantly reduced while achieving performance comparable to models trained on the entire dataset.
We provide the formal definition of coreset selection as follows:

\begin{definition}[Coreset Selection]
  Given a model $\mathcal{M}_\theta$, a training dataset $\mathcal{D}_{train}$, and a test dataset $\mathcal{D}_{test}$,
  coreset selection aims to construct a weighted subset $\mathcal{D}_{core} \subseteq \mathcal{D}_{train}$ of size $b$ such that the performance of the model $\mathcal{M}_\theta$ trained on $\mathcal{D}_{core}$  is as close as possible to the performance of a model trained on the full dataset $\mathcal{D}_{train}$ when evaluated on the test dataset $\mathcal{D}_{test}$.
  The objective can be expressed as follows:
  \begin{align*}
    & \arg\min_{\mathcal{D}_{core}}\ \left| \textsf{Eval}(\mathcal{D}_{test}, \mathcal{M}_{\boldsymbol{\theta}_{\mathcal{D}_{train}}}) - \textsf{Eval}(\mathcal{D}_{test},\mathcal{M}_{\boldsymbol{\theta}_{\mathcal{D}_{core}}})\right| \\
    & s.t. \quad |\mathcal{D}_{core}| \leq b, \ \mathcal{D}_{core} \subseteq \mathcal{D}_{train}
  \end{align*}
  \noindent where $\textsf{Eval}(\mathcal{D}_{test},\cdot)$ is an evaluation function on the test data $\mathcal{D}_{test}$, such as accuracy.
  Also,  $\mathcal{M}_{\boldsymbol{\theta}_{\mathcal{D}_{train}}}$ and $\mathcal{M}_{\boldsymbol{\theta}_{\mathcal{D}_{core}}}$ are LLMs trained on the full training dataset $\mathcal{D}_{train}$ and the coreset $\mathcal{D}_{core}$, respectively.
\end{definition}

The optimization problem is challenging to solve explicitly for two main reasons. First, the objective function is often non-convex or high-dimensional~\cite{danilova2020recent}, making direct optimization difficult.
Second, solving the problem requires evaluating all possible subsets~\cite{feldman2020introduction},
which is computationally prohibitive due to the combinatorial nature of the selection process.

Therefore, researchers have explored alternative formulations to construct the coreset efficiently while approximating the desired performance~\cite{albalak2024a,feldman2020introduction}.
Specifically, depending on the techniques used to select the coreset
from the entire dataset, existing approaches
can be classified into three types, i.e., \textit{data-property-based}~\cite{bhattExperimentalDesignFramework2024,chaiGoodCoreDataeffectiveDataefficient2023,chaiEfficientCoresetSelection2023,chenMaybeOnly052023,10.14778/3712221.3712249,tirumalaD4ImprovingLLM2023,10.14778/3561261.3561267}, \textit{uncertainty-based}~\cite{anknerPerplexedPerplexityPerplexityBased2024,demidovskijDARELDataReduction2023,heSHEDShapleyBasedAutomated2024,10.14778/3574245.3574261,liQuantityQualityBoosting2024,DBLP:conf/icde/LiuDLLCZ24,10.1145/3677139,mekalaSmallerLanguageModels2024,thrushImprovingPretrainingData2024,yangSmallToLargeS2LScalable2024,zhangHarnessingDiversityImportant2024,zhouLoBaSSGaugingLearnability2023}, and \textit{gradient-based}~\cite{dengInfluentialLanguageData2024,killamsettyGRADMATCHGradientMatching2021,mirzasoleimanCoresetsDataefficientTraining2020,nguyenMemoryefficientTrainingLLMs2024,yangSustainableLearningCoresets2023,zhangTAGCOSTaskagnosticGradient2024}.

\subsubsection{Data-property-based  Approaches}
Data-property-based approaches rely on the intrinsic properties of the dataset, such as feature diversity or statistical similarity, to select the coreset.
For instance, methods such as FastCore~\cite{chaiEfficientCoresetSelection2023} and D4~\cite{tirumalaD4ImprovingLLM2023} employ clustering techniques (e.g., K-means~\cite{lloyd1982least}) to group similar data points based on their features extracted from a pre-trained encoder~\cite{bhattExperimentalDesignFramework2024,chenMaybeOnly052023,tirumalaD4ImprovingLLM2023}.
Subsequently, they select representative samples that maximize the coverage of dataset's distribution.
While these approaches are simple and effective at maintaining data diversity, they do not consider the task-specific information and the influence of individual data points on the specific model (e.g., LLMs)~\cite{albalak2024a}.
Consequently, the task-agnostic nature can lead to suboptimal performance~\cite{albalak2024a}.

\subsubsection{Uncertainty-based Approaches}
Uncertainty-based methods select coreset instances by leveraging uncertainty metrics derived from the training target model or proxy models, such as training loss and prediction scores.
The basic idea behind these methods is that uncertain samples, which are those the model cannot confidently classify or predict, contain valuable information and are likely to have a significant impact on improving the model's performance during training.
Firstly,
several studies, including ~\cite{anknerPerplexedPerplexityPerplexityBased2024,demidovskijDARELDataReduction2023,marionWhenLessMore2023,thrushImprovingPretrainingData2024,yangSmallToLargeS2LScalable2024}, compute prediction scores and training loss for inputs to the model. These works then empirically select samples based on uncertainty indicators such as prediction scores, training loss, or a combination of both.
Depending on the strategy, they may empirically select the most uncertain samples, moderately uncertain samples, or the easiest samples to learn.
Second, other approaches~\cite{heSHEDShapleyBasedAutomated2024,liQuantityQualityBoosting2024,mekalaSmallerLanguageModels2024,zhangHarnessingDiversityImportant2024,zhouLoBaSSGaugingLearnability2023} define uncertainty indicators based on changes in the training loss before and after training.

However, uncertainty-based approaches present three key limitations when applied to LLMs. First, they require a trained or partially trained model to calculate prediction or uncertainty scores, which is impractical for LLMs due to their immense number of parameters.
Second, estimating uncertainty over large datasets is computationally expensive, as it necessitates running the LLM on every instance to obtain scores.
Finally, empirically defined uncertainty scores are not theoretically or directly linked to model optimization, meaning that samples selected based on uncertainty may not provide the most beneficial gradients for improving the model.

\subsubsection{Gradient-based Approaches}\label{sssec:gradient-based}
Since data-property-based and uncertainty-based methods cannot theoretically guarantee optimal model performance, gradient-based methods are proposed.
They directly use gradient information from the training process to select data samples that are most influential for model optimization.
The most common way is to build a gradient matching problem~\cite{dengInfluentialLanguageData2024,nguyenMemoryefficientTrainingLLMs2024,zhangTAGCOSTaskagnosticGradient2024}, which is generally defined in other deep learning models~\cite{killamsettyGRADMATCHGradientMatching2021,mirzasoleimanCoresetsDataefficientTraining2020,yangSustainableLearningCoresets2023}. The gradient matching coreset aims to closely approximate the gradient of all data samples at a specific training step with parameter $\boldsymbol{\theta}$, seeking to ensure that training on the coreset results in similar parameter updates as training on the full dataset. Formally, following~\cite{li20242}, the problem is defined as:

\begin{definition}[Gradient Matching Coreset]
  Given a model $\mathcal{M}_\theta$ and a full training dataset $\mathcal{D}_{train}$ of size $n$, the goal is to construct a weighted subset $\mathcal{D}_{{core}}\subseteq\mathcal{D}_{train}$ of size less than $b$,
  such that the total gradient computed on the coreset $\mathcal{D}_{{core}}$ closely approximates the total gradient of the full dataset $\mathcal{D}_{{train}}$ as follows:
  \begin{align}
    & \arg\min_{\mathcal{D}_{core}}\| \sum_{i=1}^{|\mathcal{D}_{train}|}\nabla_{\boldsymbol{\theta}} \mathcal{L}(X_i, \boldsymbol{\theta}) - \sum_{j=1}^{|\mathcal{D}_{core}|} w_j\nabla_{\boldsymbol{\theta}} \mathcal{L}(X_j, \boldsymbol{\theta}) \|     \label{eq:gradient_matching} \\
    & {s.t.} \   |\mathcal{D}_{core}|\leq b, w_j\geq 0, \  \sum_{j=1}^{|\mathcal{D}_{core}|}w_j=n, \nonumber
  \end{align}
\end{definition}
\noindent where $\nabla_{\boldsymbol{\theta}} \mathcal{L}(X_j, \boldsymbol{\theta})$ is the gradient of sample $X_j$ with respect to model parameters $\boldsymbol{\theta}$, and $w_j$ is the non-negative weight of each selected sample $X_j \in \mathcal{D}_{core}$.

Due to the non-convexity of the objective function and the exponential combinatorial nature of subset selection in Equation~\eqref{eq:gradient_matching}, it is infeasible to obtain the optimal coreset through exhaustive search~\cite{mirzasoleimanCoresetsDataefficientTraining2020}.
The computational complexity grows exponentially with the dataset size, making direct optimization impractical for large-scale problems~\cite{mirzasoleimanCoresetsDataefficientTraining2020}.
To address this, several approaches~\cite{li2024fight}
compute the gradient norms of the loss function for individual instances and select samples by maximizing both gradient magnitudes and diversity.
While they are effective for smaller models, these approaches are unsuitable for LLMs due to the prohibitive cost of gradient computations, which are significantly higher for billions of parameters compared to smaller models like GNNs~\cite{li2024fight} or MLPs~\cite{mirzasoleimanCoresetsDataefficientTraining2020}.
To improve efficiency, several researchers replace gradient-based metrics with feature-based distances between samples~\cite{li20242,liCamelManagingData2022,mirzasoleimanCoresetsDataefficientTraining2020}.
This relaxation avoids gradient computations but fails to capture the complex dependencies and rich contextual information in transformer-based LLMs, leading to a loss of critical information and suboptimal performance.

More importantly, the above approaches select the coreset statically before training, which ignores the evolving dynamics of LLMs during training.
This static selection often results in suboptimal coresets that fail to adapt to changes in the model's learning process.
Thus, to address these challenges, we propose a dynamic coreset selection framework for LLMs.

\begin{figure*}[htbp]
  \centering
  \vspace{-1em}
  \includegraphics[width=\linewidth]{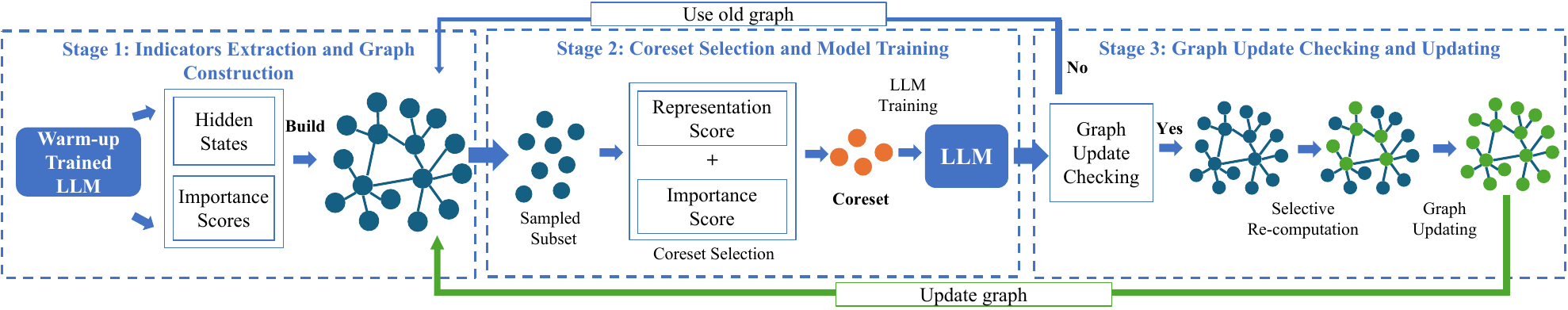}
  \vspace{-2em}
  \caption{Overview of GRACE. Stage 1: Hidden states and importance scores of all training samples are extracted from a warm-up trained model. A $k$-NN graph is constructed from these indicators. Stage 2: Sample a subset of samples from the graph and select the coreset, then train the model with the coreset. Stage 3: At checking steps, if significant changes in importance scores are detected, the graph will be updated. A subset is selected for recomputation and updated scores and embeddings are propagated through the graph structure. The $k$-NN graph is also updated and then reused in the next selection. }
  \label{fig:overall}
  \vspace{-1em}
\end{figure*}

\section{Framework Overview}
We introduce the procedure for GRACE, which is composed of three stages, summarized in Algorithm~\ref{alg:overall}.

\noindent \textbf{Stage 1}: \textbf{Indicator Extraction and Graph Construction.} Before coreset selection and model training, we extract the required features of training samples.
Based on this feature extraction, we construct a mutual $k$-NN graph to capture the similarity structure among samples, which serves as the backbone for coreset construction, enhancing diversity and supporting efficient local updates later in the process.
As shown in Algorithm~\ref{alg:overall} lines 1-3, given the warmup-trained target LLM $\mathcal{M}_{\boldsymbol{\theta}_{warm}}$, we extract importance scores $I(X_i)$ and the last-layer hidden states $\mathbf{H}^L_i$ for any $X_i\in \mathcal{D}_{train}$. We then construct a mutual $k$-NN graph $G$ by treating data samples as nodes and computing Euclidean distances from extracted hidden states. During construction, we retrieve the top-$k$ nearest neighbors in embedding space for each sample $X_i$ and keep an undirected edge $(i,j)$ only if the relation is mutual: $X_i$ is in the top-$k$ list of $X_j$, and $X_j$ is in the top-$k$ list of $X_i$.

\noindent \textbf{Stage 2}: \textbf{Coreset Selection and Model Training.}
After indicator extraction and graph construction, we train the LLM with coreset selection, as shown in Algorithm~\ref{alg:overall} lines 4-9. We divide the training steps within one epoch into several intervals, each comprising $t_c$ training steps, and define $b$ as the coreset budget for each interval. As in lines 5-6 of Algorithm~\ref{alg:overall}, at the start of each interval, we sample a subset $\mathcal{D}_s$ of size $n_S$ without replacement from the entire training dataset and construct a coreset $\mathcal{D}_{core}$ of size $b$ with the coreset selection objective following Equation~\eqref{eq:target} in Section~\ref{subsec:our_objective}. Then, the model is trained on the coreset $\mathcal{D}_{core}$ for $t_c$ steps, as shown in lines 7-9. This process is also shown as Stage 2 in Figure~\ref{fig:overall}.

\noindent \textbf{Stage 3}: \textbf{Graph Update Checking and Updating.} After $t_c$ training steps, following Algorithm~\ref{alg:overall} lines 10-11, we perform the checking and update procedure as described in Algorithm~\ref{algo:dynamic}, which may update importance scores and embeddings. We first randomly sample a small validation subset from the most recently selected coreset, recompute their current importance scores with the latest model parameters, and measure the score change via Equation~\eqref{eq:check}.
If the change exceeds the threshold, we select samples via Equation~\eqref{eq:update_select}, recompute their scores as well as embeddings, and update them through substitution and propagation over the graph structure via Equation~\eqref{eq:update}.
Finally, we incrementally repair only the affected portions of the $k$-NN structure. If no update is needed, we will keep the graph unchanged. After the checking and update step, we redo the sampling and coreset selection and repeat Stages 2 and 3 until the training finishes, as shown in Figure~\ref{fig:overall}.

\normalem
\begin{algorithm}[ht]
  \caption{Overall Framework}
  \label{alg:overall}
  \KwIn{Warmed-up LLM $\mathcal{M}_{\boldsymbol{\theta}_{warm}}$, initial LLM $\mathcal{M}_{\boldsymbol{\theta}_0}$, training dataset $\mathcal{D}_{train}$, training steps $\mathcal{T}$, check interval $t_c$, sample size $n_S$, coreset budget $b$}
  \KwOut{Trained LLM with parameters $\boldsymbol{\theta}_t$}
  $\{\bar{\mathbf{H}}_i, I(X_i)\}_{i=1}^n \gets $ \textsf{ExtractFeature}($\mathcal{M}_{\boldsymbol{\theta}_{warm}}$) \\
  $G \gets $ \textsf{BuildGraph}($\{\bar{\mathbf{H}}_i, I(X_i)\}_{i=1}^n$) \\
  $t=0$ \\
  \While {$t < \mathcal{T}$}{
    $\mathcal{D}_{S} \gets$\textsf{RandomSample}($\mathcal{D}_{train}, n_S$) \\
    $\mathcal{D}_{core} \gets$ \textsf{SelectCoreset}($\mathcal{D}_{S}, G, b$) \DontPrintSemicolon \tcp*{Algorithm~\ref{algo:greedy}}
    \For {$i=1\dots t_c$} {
      $\boldsymbol{\theta}_t \gets$
      $\boldsymbol{\theta}_{t-1} - \eta \frac{1}{|\mathcal{D}_{core}|} \sum_{X_i \in \mathcal{D}_{core}} \nabla_{\boldsymbol{\theta}} \mathcal{L}(X_i, \boldsymbol{\theta}_{t-1})$ \\
      $t\gets t+1$ \\
    }
    \If{$t\bmod t_c = 0$} {
      $G \gets$ \textsf{CheckandUpdate}($\mathcal{M}_{\boldsymbol{\theta}_t}, G, \mathcal{D}_{core}$) \DontPrintSemicolon \tcp*{Algorithm~\ref{algo:dynamic}}
    }
  }
  \Return Trained LLM $\mathcal{M}_{\boldsymbol{\theta}}$
\end{algorithm}
\ULforem

\section{Methods}\label{sec:ss}

In this section, we first provide a theoretical analysis of how to select a coreset for LLMs to preserve their performance,
taking into account the specific characteristics of transformer-based architectures and autoregressive training loss in Section~\ref{subsec:theoretical}. Building on this analysis, in Section~\ref{subsec:our_objective}, we define the objectives for coreset selection that incorporate representation scores of hidden states and the gradient-based importance scores extracted from the LLM. Then we introduce the $k$-NN graph and propose a greedy algorithm for coreset selection. In Section~\ref{subsec:graph}, we propose an adaptive graph update checking and updating approach over the $k$-NN graph to fit the changes in training steps and reduce computational overheads.

\subsection{Theoretical Analysis}\label{subsec:theoretical}
As introduced in Section~\ref{sssec:gradient-based}, directly utilizing Equation~\eqref{eq:gradient_matching} as the objective for coreset selection on LLMs is nontrivial. First, the relaxation analysis of transformers becomes more complex due to the non-linearity introduced by self-attention mechanisms, requiring further theoretical investigation. Second, utilizing gradients from other layers adds significant computational overhead during the backward pass. Moreover, the process remains highly time-consuming due to the size and complexity of LLMs, and the sparsity of gradients persists, given the billions of parameters in these models. Therefore, a theoretical analysis of coreset selection for LLMs remains lacking, particularly in the context of transformer-based architectures and billion scale LLM parameters. In the following, we present a theoretical framework to adapt the coreset selection definition in Equation~\eqref{eq:gradient_matching}
to LLMs.

Since it is infeasible to solve the coreset selection problem in Equation~\eqref{eq:gradient_matching} directly, following ~\cite{li20242,mirzasoleimanCoresetsDataefficientTraining2020,yangSustainableLearningCoresets2023}, we first relax the weight $w_j$ to a binary selection indicator $w_j \in \{0, 1\}$
and define a mapping function $\gamma: \mathcal{D}_{train} \rightarrow \mathcal{D}_{core}$ to associate each training sample $X_i \in  \mathcal{D}_{train}$ with a representative sample $ X_j \in \mathcal{D}_{core}$ in the coreset, i.e., $\gamma(X_i) = X_j$.
Therefore, we have the following upper bound through the triangle inequality:
\begin{align*}
  &\| \sum_{X_i\in\mathcal{D}_{train}}\nabla_{\boldsymbol{\theta}} \mathcal{L}(X_i, \boldsymbol{\theta}) - \sum_{X_j\in \mathcal{D}_{core}} w_j\nabla_{\boldsymbol{\theta}} \mathcal{L}(X_j, \boldsymbol{\theta}) \| \nonumber \\
  & \leq \|\sum_{X_i\in\mathcal{D}_{train}}(\nabla_{\boldsymbol{\theta}} \mathcal{L}(X_i, \boldsymbol{\theta}) - \nabla_{\boldsymbol{\theta}} \mathcal{L}(\gamma(X_i), \boldsymbol{\theta}))\| \nonumber\\
  & \leq \sum_{X_i\in\mathcal{D}_{train}}\|\nabla_{\boldsymbol{\theta}} \mathcal{L}(X_i, \boldsymbol{\theta}) - \nabla_{\boldsymbol{\theta}} \mathcal{L}(\gamma(X_i), \boldsymbol{\theta})\|
\end{align*}
Since minimizing this upper bound corresponds to assigning each training point to its nearest neighbor in the gradient space of the coreset instance, we have:
\begin{align}\label{eq:upper2}
  & \sum_{X_i\in\mathcal{D}_{train}}\|\nabla_{\boldsymbol{\theta}} \mathcal{L}(X_i, \boldsymbol{\theta}) - \nabla_{\boldsymbol{\theta}} \mathcal{L}(\gamma(X_i), \boldsymbol{\theta})\| \nonumber\\
  & \leq \sum_{X_i\in\mathcal{D}_{train}}\min_{X_j\in\mathcal{D}_{core}}\|\nabla_{\boldsymbol{\theta}} \mathcal{L}(X_i, \boldsymbol{\theta}) - \nabla_{\boldsymbol{\theta}} \mathcal{L}(X_j, \boldsymbol{\theta})\|
\end{align}

Equation~\eqref{eq:upper2} is still computationally intensive, as it requires extracting and comparing gradients for every training sample. Therefore, we further upper-bound the gradient representations through intermediate features that are easier to obtain during forward passes. As introduced in Section~\ref{subsec:llm}, LLMs are composed of layers of transformer blocks, which are non-linear and consist of a complex self-attention mechanism. Therefore, we first analyze the coreset selection based on a single-layer transformer block and discuss how to extend it to multi-layer transformer-based LLMs.

\begin{theorem}\label{thm:select}
  Given a one-layer transformer-based LLM $\mathcal{M}_\theta$ and a training dataset $\mathcal{D}_{{train}} = \{X_i\}_{i=1}^n$, the LLM $\mathcal{M}_\theta$ processes each training instance $X_i$ as input and progressively predicts the logits $\hat{\mathbf{O}}_i$ for $X_i$. Consequently, given the  $\mathcal{D}_{{train}}$ and the coreset  $\mathcal{D}_{core}$,
  the gradient difference in Equation~\eqref{eq:upper2} can be upper-bounded by the following objective:
  \begin{align}
    & \sum_{X_i\in\mathcal{D}_{train}}\min_{X_j\in\mathcal{D}_{core}}\|\nabla_{\boldsymbol{\theta}} \mathcal{L}(X_i, \boldsymbol{\theta}) - \nabla_{\boldsymbol{\theta}} \mathcal{L}(X_j, \boldsymbol{\theta})\|  \nonumber\\
    \leq & \sum_{X_i\in \mathcal{D}_{train}} \min_{X_j\in \mathcal{D}_{core}} c_1\left\|\mathbf{X}_i - \mathbf{X}_j \right\|  + c_2\| \nabla_{\mathbf{\hat{O}}_j}\mathcal{L}(X_j,\boldsymbol{\theta})\| + c_3 \label{eq:thm1}
  \end{align}
  \noindent where $\mathbf{X}_i$ denotes the input embedding representations of sample $X_i$ and $\nabla_{\mathbf{\hat{O}}_i}\mathcal{L}(X_i,\boldsymbol{\theta})$ is the gradient of loss with respect to the output logits $\hat{\mathbf{O}}_i$ for $X_i$. Also, $c_1$ and $c_2$ are Lipschitz constants associated with the model parameters $\boldsymbol{\theta}$ and $c_3=\| \sum_{X_i \in \mathcal{D}_{train} }{ \nabla_{\mathbf{\hat{O}}_i} \mathcal{L}(X_i,\boldsymbol{\theta})}\| $.
\end{theorem}

\begin{proof}[Proof Sketch]\label{prf1}
  We analyze over $\mathbf{W}_1$, $\mathbf{W}_2$, $\mathbf{W}_Q$, $\mathbf{W}_K$, $\mathbf{W}_V$ and $\mathbf{W}_O$ defined in~\eqref{eqdef:attention},~\eqref{defeq:multihead} and~\eqref{defeq:ffn}. Defining $c_1$ and $c_2$ as constants correlated with $\mathbf{W}_1$, $\mathbf{W}_2$, $\mathbf{W}_Q$, $\mathbf{W}_K$, $\mathbf{W}_V$ and $\mathbf{W}_O$, we have:
  \begin{align*}
    &\|\nabla_{\boldsymbol{\theta}} \mathcal{L}(X_i, \boldsymbol{\theta}) - \nabla_{\boldsymbol{\theta}} \mathcal{L}(X_j, \boldsymbol{\theta})\| \\
    \leq & c_1\|\mathbf{X}_i - \mathbf{X}_j\| + c_2\|\nabla_{\hat{\mathbf{O}_i}}\mathcal{L}(X_i,\boldsymbol{\theta})-\nabla_{\hat{\mathbf{O}_j}}\mathcal{L}(X_j,\boldsymbol{\theta})\| \\
    \leq & c_1(\|\mathbf{X}_i - \mathbf{X}_j\|) + c_2\|\nabla_{\hat{\mathbf{O}_i}}\mathcal{L}(X_i,\boldsymbol{\theta})\|+c_2\|\nabla_{\hat{\mathbf{O}_j}}\mathcal{L}(X_j,\boldsymbol{\theta})\| \
  \end{align*}
  Then we have:
  \begin{align*}
    &  \sum_{X_i\in\mathcal{D}_{train}}\min_{X_j\in\mathcal{D}_{core}}\|\nabla_{\boldsymbol{\theta}} \mathcal{L}(X_i, \boldsymbol{\theta}) - \nabla_{\boldsymbol{\theta}} \mathcal{L}(X_j, \boldsymbol{\theta})\|  \nonumber\\
    \leq & \sum_{X_i\in \mathcal{D}_{train}} \min_{X_j\in \mathcal{D}_{core}} c_1\left\|\mathbf{X}_i - \mathbf{X}_j \right\| +  c_2\| \nabla_{\mathbf{\hat{O}}_j}\mathcal{L}(X_j,\boldsymbol{\theta})\| +c_3
  \end{align*} Thus, the theorem is proved. Due to the space limit, we put the whole proof in the Appendix.
\end{proof}

The objective in Equation~\eqref{eq:thm1} is to select a coreset $\mathcal{D}_{{core}}$ with size $b$ from the training data $\mathcal{D}_{{train}}$,
such that the LLM optimized on  $\mathcal{D}_{\text{core}}$ achieves similar performance to that optimized on $\mathcal{D}_{{train}}$, balancing data diversity and data quality.
Specifically,
\begin{itemize}[leftmargin=12pt]
  \item The first term $\|\mathbf{X}_i - \mathbf{X}_j\|$ is a \textit{representation score} that measures the distance between a training data point $X_i \in \mathcal{D}_{{train}}$ and a coreset point $X_j\in \mathcal{D}_{{core}}$. The goal is to minimize this distance, ensuring that the coreset $\mathcal{D}_{{core}}$ adequately represents the entire training dataset $\mathcal{D}_{{train}}$.
  \item The second term $ \|\nabla_{\mathbf{\hat{O}}_j} \mathcal{L}(X_j, \boldsymbol{\theta})\|$ is an \textit{importance score} that focuses on the gradient magnitude of the coreset point  $X_j \in \mathcal{D}_{{core}}$.  By minimizing this term, the selected coreset points are encouraged to have gradients that are significant enough to effectively guide the model's parameter updates, thereby preserving the model's training quality.
\end{itemize}
The inclusion of both terms ensures that the coreset points are not only representative of the input diversity in the training data but also make substantial contributions to gradient updates during optimization. This dual focus enables the trained LLMs to achieve high performance when using the selected coreset.

\subsection{Coreset Selection}\label{subsec:our_objective}
In this subsection, we propose the coreset selection objective and present a greedy algorithm for coreset selection.

\subsubsection{Representation Score}\label{define:repre}

Intuitively, the point $X_j$ in $\mathcal{D}_{core}$ should have a minimal distance to as many data points in $\mathcal{D}_{train}$ as possible for sufficient data diversity. To minimize the embedding distance $\|\mathbf{X}_i - \mathbf{X}_j\|$, we utilize the cosine similarity to convert it into a maximization objective. We use hidden state representations from the final layer of the model for each instance, as the transformer blocks there provide more semantic and contextual information.

Specifically, given a training dataset $\mathcal{D}_{{train}} = \{X_i\}_{i=1}^n$ where each $X_i$ is composed of $T$ tokens, suppose the LLM $\mathcal{M}_\theta$ processes each training instance $X_i$ as input and produces hidden state representations $\mathbf{H}_i^{(L)} \in \mathbb{R}^{T\times d}$ for $X_i$ at the last layer $L$ of transformer blocks, where $d$ is the hidden state dimension.
Formally, the representation score of the coreset $\mathcal{D}_{core}$ is:
\begin{equation}\label{eq:repre}
  R(\mathcal{D}_{core})=\sum_{X_i\in \mathcal{D}_{train}} \max_{X_j\in \mathcal{D}_{core}}\cos(\bar{\mathbf{H}}_i^{(L)} ,\bar{\mathbf{H}}_j^{(L)})
\end{equation}
\noindent To obtain a fixed-size representation, we compute the mean in the token length dimension of $\mathbf{H}_i^{(L)}$, denoted as  $\bar{\mathbf{H}}_i^{(L)}\in\mathbb{R}^d$.

\subsubsection{Importance score}

In Theorem~\ref{thm:select}, the selected coreset is encouraged to have gradients with respect to the output logits that can effectively guide the model's parameter updates, i.e., $ \|\nabla_{\mathbf{\hat{O}}_j} \mathcal{L}(X_j, \boldsymbol{\theta})\|$. It can be computed by:
\begin{equation} \label{eq:el2n}
  I(X)=\|\nabla_{\mathbf{\hat{O}}}\mathcal{L}(X,\boldsymbol{\theta})\|=\left\|\hat{P}(\hat{\mathbf{O}})-\mathbf{O}\right\|
\end{equation}
\noindent where $\mathbf{O}$ is the one-hot ground truth matrix of sample $X$. We define $I(X)$ as the importance score.

The importance score measures per-sample gradient sensitivity, analogous to per-sample gradient-norm difficulty. However, directly selecting samples with large importance score values can be sub-optimal. High importance scores may reflect noisy and mislabeled inputs. In addition, under constrained training budgets, prioritizing hard samples does not always contribute to model training due to limited steps for gradients explode. In practice, with constrained training budgets, prioritizing low- or mid-difficulty mass yields steadier improvement. Such a strategy is consistent with the observations in~\cite{sorscherNeuralScalingLaws2022,zhengCoveragecentricCoresetSelection2023} and with our evaluation setup.

Thus, to adaptively favor moderately difficult samples~\cite{acharyaBalancingFeatureSimilarity2024,choLightweightDatasetPruning2025}, instead of selecting based on the largest or smallest importance scores, we first rescale the raw importance scores to $[0,1]$ via min-max normalization $\tilde{I}(X)=\frac{I(X)-\min_{X'\in \mathcal{D}_{train}}I(X)}{\max_{X'\in \mathcal{D}_{train}}I(X)-\min_{X'\in \mathcal{D}_{train}}I(X)}$ to normalize scale differences across datasets, sources, or checkpoints. Then, we define a warped importance score for selection as:

\begin{align}\label{eq:importance}
  \hat{I}(X)&=(\text{Beta}(\tilde{I}(X),\alpha,\beta))^\gamma\nonumber \\
  &=(\frac{\Gamma(\alpha+\beta)}{\Gamma(\alpha)\Gamma(\beta)}\tilde{I}(X)^{\alpha-1}(1-\tilde{I}(X))^{\beta-1})^\gamma
\end{align}
\noindent $\alpha$ and $\beta$ are the parameters associated with the Beta distribution that control the mean and variance. $\gamma \geq 0$ is a temperature parameter that controls the sharpness of the warped importance scores without changing the induced ranking.
We define:
\begin{equation}
  \alpha=1+C\cdot (\text{Mean}(\tilde{I}(X)))^q \cdot \eta^r,
  \beta=C-\alpha
\end{equation}
\noindent where $\eta$ is the coreset selection budget in percentage, $q$ and $r$ are hyper-parameters controlling the contribution of the mean score and budget, and $C$ is a constant.
In this way, we suppress outliers with extremely high or low $I(X_i)$. In addition, we can shift the mode of the Beta curve toward harder samples when the budget is large or toward easier samples when the budget is tight, thus adapting to different regions according to the score distribution and training budget. The temperature $\gamma$ allows us to smoothly align the score scale with the representation score computed by embedding similarity.
Formally, given a selected subset from the training dataset $\mathcal{D}_{{core}} \subseteq \mathcal{D}_{{train}}$, the warped importance score of the coreset $\mathcal{D}_{core}$ is:
\begin{align}\label{eq:impor}
  \hat{I}(\mathcal{D}_{core})&=\sum_{X_i\in \mathcal{D}_{core}}(\text{Beta}(\tilde{I}(X_i),\alpha,\beta))^\gamma \nonumber\\
  &=\sum_{X_i\in \mathcal{D}_{core}}(\frac{\Gamma(\alpha+\beta)}{\Gamma(\alpha)\Gamma(\beta)}\tilde{I}(X_i)^{\alpha-1}(1-\tilde{I}(X_i))^{\beta-1})^\gamma
\end{align}

With these scores, we then formally define our coreset selection problem as follows:
\begin{definition}[Coreset Selection Problem]\label{define:coreset} Given the training dataset $\mathcal{D}_{train}$ and the selection budget $b$, our target is to select a coreset $\mathcal{D}_{core}\subseteq \mathcal{D}_{train}$ that maximizes the following objective $RS(\mathcal{D}_{core})$:
  \begin{equation}
    \begin{split}
      & RS(\mathcal{D}_{core})=\lambda R(\mathcal{D}_{core}) + (1-\lambda)\hat{I}(\mathcal{D}_{core}) \\
      & s.t. |\mathcal{D}_{core}|\leq b \\
    \end{split}
    \label{eq:target}
  \end{equation}
  \noindent where $R(\mathcal{D}_{core})$ is the representation score of $\mathcal{D}_{core}$ defined in Equation~\eqref{eq:repre}, $\hat{I}(\mathcal{D}_{core})$ is the importance score of $\mathcal{D}_{core}$ defined in Equation~\eqref{eq:importance}, and $\lambda$ is a hyper-parameter balancing contributions of each term.
\end{definition}

\begin{theorem}
  \label{thm:nphard}
  The Coreset Selection Problem is NP-hard.
\end{theorem}

\begin{proof}[Proof Sketch]
  This problem can be reduced from the Maximum Coverage Problem (MCP)~\cite{hochbaum1997approximating}, which is NP-hard.
  Due to space limit, we put the full proof in our technique report.
\end{proof}

\begin{figure*}[t]
  \centering
  \vspace{-1em}
  \includegraphics[width=\linewidth]{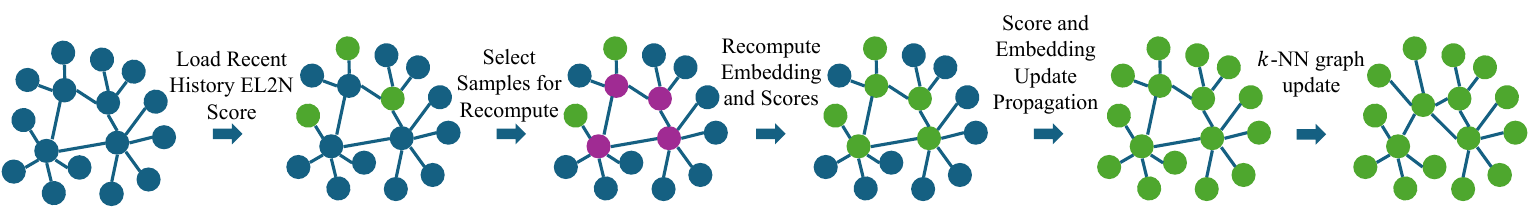}
  \caption{Overview of the graph update process. From left to right: we begin with a $k$-NN graph constructed from initial representations and EL2N scores (blue nodes). When adaptive update checking is triggered, a subset of recently trained samples (green) is loaded and we use their history scores as new scores. Another small number of samples (purple) are then selected for recomputation. Combining the history scores and recomputed scores as updated scores, they are then propagated to their neighbors through the graph structure (blue to green), allowing approximate updates without full recomputation.}
  \label{fig:graph-update}
  \vspace{-1em}
\end{figure*}

\subsubsection{Coreset Selection Algorithm}

Since the coreset selection problem is NP-hard, we cannot obtain an optimal solution in polynomial time. To address this, we propose a greedy algorithm. The basic idea is to iteratively select the data point that provides the maximum gain of $RS$ score until reaching the budget $b$ in a greedy manner. First of all, we define the marginal gain of the $RS$ score as follows:
\begin{equation} \label{eq:marginalrs}
  \Delta RS(X_i|\mathcal{D}_{core})=RS(\mathcal{D}_{core}\cup \{X_i\})-RS(\mathcal{D}_{core})
\end{equation}

\normalem

\begin{algorithm}[t]
  \caption{SelectCoreset}
  \label{algo:greedy}
  \KwIn{Candidate samples $\mathcal{D}_{S}$, selection budget $b$, Graph $G$ built on training data}
  \KwOut{Coreset $\mathcal{D}_{core}$}
  $\mathcal{D}_{core} \gets \emptyset$ \\
  \While{$|\mathcal{D}_{core}|<b$} {
    \For{$X_i\in \mathcal{D}_{S}\backslash \mathcal{D}_{core}$} {
      $RS(\mathcal{D}_{core}\cup \{X_i\}) \gets \text{Equation \eqref{eq:target}}$ \\
      $\Delta RS(X_i|\mathcal{D}_{core})=RS(\mathcal{D}_{core}\cup \{X_i\})-RS(\mathcal{D}_{core})$ \\
    }
    $X^*={\arg\max}_{X_i\in \mathcal{D}_{S}\setminus \mathcal{D}_{core}}\Delta RS(X_i|\mathcal{D}_{core})$\\
    $\mathcal{D}_{core}=\mathcal{D}_{core}\cup \{X^*\}$
  }

  \Return $\mathcal{D}_{core}$
\end{algorithm}

\ULforem

To simplify the process of the coreset selection algorithm, we build a mutual $k$-NN graph $G$ to save the hidden state representations and importance scores. The detailed definition of $G$ is as follows:
\begin{definition}[$k$-NN graph] \label{define:knngraph}
  We construct a mutual $k$-NN graph $G=(\mathcal{D}_{train},E)$ where each node $X_i$ is a data point in $\mathcal{D}_{train}$ filled with its importance score and hidden state representation. There is an edge $e_{ij} \in E$ connecting $X_i$ and $X_j$ if and only if $X_j$ is one of the top-$k$ nearest neighbors of $X_i$, and $X_i$ is one of the top-$k$ nearest neighbors of $X_j$ too. The edge weight of $e_{ij}$ is defined as $w_{ij}=\exp(-{\|\bar{\mathbf{H}}_i^{(L)} - \bar{\mathbf{H}}_j^{(L)}\|^2}/{100})$.
\end{definition}

Since hidden states and importance scores evolve during training, $G$ must be updated accordingly, where the details are discussed in Section~\ref{subsec:graph}. Then, as shown in Algorithm~\ref{algo:greedy}, we initialize $\mathcal{D}_{core}$ as an empty set. Then we compute the marginal gain of $RS$ score as Equation~\eqref{eq:marginalrs} for each $x_i$ in $\mathcal{D}_S$, where $RS(\cdot)$ is the score calculated from Equation~\eqref{eq:target}. This process is repeated until the number of selected samples reaches the budget $b$.

\begin{theorem}
  \label{thm:greedy}
  Algorithm~\ref{algo:greedy} achieves a $(1-\frac{1}{e})$ approximation ratio to the optimal solution.
\end{theorem}
\begin{proof}[Proof Sketch]
  Suppose $\mathcal{D}_{core}^{k}$ is the selected coreset using Algorithm~\ref{algo:greedy} of size $k$, and $\mathcal{D}_{core}^{opt}$ is the unknown subset that maximizes the score $RS(\cdot)$ in Equation~\eqref{eq:marginalrs}. We can first prove that $\Delta RS(X_i|\mathcal{D}_{core}^{k})$ is monotone increasing and submodular. Then, we prove that it satisfies $RS(\mathcal{D}_{core}^{k}) \geq (1-\frac{1}{e})RS(\mathcal{D}_{core}^{opt})$.
  Due to space limit, we put the full proof in our Appendix.
\end{proof}

\noindent\underline{\textbf{Time Complexity Analysis.}}
\label{subsec:complex-coreset} Given the selection budget $b$, the candidate sample size $|\mathcal{D}_S|$, and the embedding dimension $d$, the time complexity is $O(d)$ for each computation of Equation~\eqref{eq:repre}, and $O(1)$ for Equation~\eqref{eq:impor}. This requires a loop of $O(|D_S|b)$ for each selection, with $k$ selections in total. Thus, the total time complexity for Algorithm~\ref{algo:greedy} is $O(|\mathcal{D}_S|(b^2d+b))$.

\subsection{Graph Update Checking and Updating}\label{subsec:graph}
According to Equation~\eqref{eq:gradient_matching} and Equation~\eqref{eq:target}, the selection quality is constrained by representations and scores under parameter $\boldsymbol{\theta}$. While $\boldsymbol{\theta}$ changes during training, the extracted representations and scores evolve. Therefore, if we use the static graph constructed at the beginning to select the coreset for subsequent training steps, the new coreset may not accurately reflect the current training dynamics or capture the most representative samples under the updated model parameters.

One simple way is to fully update the $k$-NN graph $G$ by recomputing all sample embeddings and scores. However, its time complexity is $\mathcal{O}(N^2 d+NL(dT^2+d^2T))$, which is high and impractical for LLM training. In addition, we observe that model drift is not uniform: most samples maintain roughly the same difficulty ranking across short training intervals, and only a small subset changes substantially. To efficiently capture training dynamics, we propose approximating score and representation updates via the graph $G$. The dynamic update procedure in Algorithm~\ref{algo:dynamic} consists of three main steps:

\begin{itemize}[leftmargin=*]
  \item \textbf{Step 1: Adaptive Update Checking: }
    In lines 1–3, we randomly sample a set $\mathcal{S}_t$ from the current $D_{{core}}$ and compute their importance score variation using Equation~\eqref{eq:check}. If the variation exceeds the threshold $\delta$, an approximate graph update is triggered.

  \item \textbf{Step 2: Selective Recalculation: }
    In lines 4–10, for each $X_i$ in $D_{train}$, we calculate $s_{update}$ according to Equation~\eqref{eq:update_select}. We then iteratively expand a set $D_{{recal}}$ to include the top-$k_{{recal}}$ disjoint data points with the highest $s_{{update}}$ scores. The importance score of data in  $D_{{recal}}$ will be accurately updated.

  \item \textbf{Step 3: Score and Embedding Propagation: }
    In lines 11–16, we recompute the importance scores for samples in $D_{{recal}}$ using the current model, while approximately updating the remaining importance scores in $D_{{train}} \setminus D_{{recal}}$ via Equation~\eqref{eq:update}. Similarly, we further check and update embeddings for selected samples and update the $k$-NN graph accordingly.
\end{itemize}

\normalem
\begin{algorithm}[t]
  \caption{CheckandUpdate}
  \label{algo:dynamic}
  \KwIn{Graph $G$, model $\mathcal{M}_{\boldsymbol{\theta}_t}$, graph score update check threshold $\delta$, embedding update check threshold $\delta_h$, coreset trained for last steps $\mathcal{D}_{core}$, update momentum $\alpha$, recompute budget $k_{recal}$}
  \KwOut{Updated Graph $G$}
  $\mathcal{S}_t \gets \text{RandomSample}(\mathcal{D}_{core}, n_s)$ \\
  $\Delta^I \gets$  Equation~\eqref{eq:check} \DontPrintSemicolon \tcp*{Update Checking}
  \If{$\Delta^I>\delta$} {
    $\{s_{update}(X_i)\}_{X_i\in \mathcal{D}_{train}} \gets \text{Equation~\eqref{eq:update_select}}$ \\
    $\mathcal{D}_{recal} \gets \emptyset$ \\
    \While{$|\mathcal{D}_{recal}| < k_{recal}$} {
      $X^*=\arg\max_{X_i \in \mathcal{D}_{train}} s_{update}(X_i)$ \\
      \If{
        $X^*\notin \mathcal{D}_{recal}, X^*\notin \mathcal{N}(X_j) \ \forall X_j\in \mathcal{D}_{recal}$
      }{
        $\mathcal{D}_{recal} \gets \mathcal{D}_{recal}\cup\{X^*\}$,
        $s_{update}(X^*)\gets 0$ \\
      }
    }

    $\{I^{new}(X_i),\mathbf{H}_i^{new}\} \gets$ \textsf{ExtractFeature}($\mathcal{M}_{\boldsymbol{\theta}_t}, \mathcal{D}_{recal}$) \\
    $\{I^{new}(X_i)\}_{X_i\in\mathcal{D}_{train}\backslash\mathcal{D}_{recal}} \gets$ Equation~\eqref{eq:update} \\
    $G \gets \{I^{new}(X_i)\}_{X_i\in\mathcal{D}_{train}}$ \\
    $\Delta^\mathbf{H} \gets$  Equation~\eqref{eq:check_rep} \DontPrintSemicolon \tcp*{Embedding Update Checking}
    \If{$\Delta^\mathbf{H}>\delta_h$} {
      $\{\mathbf{H}_i^{new}\}_{X_i\in\mathcal{D}_{train}\backslash\mathcal{D}_{recal}} \gets$ Equation~\eqref{eq:update_rep} \\
      $G \gets \text{Update $k$-NN structure}$
    }
  }
  \Return Updated Graph $G$
\end{algorithm}
\ULforem

\subsubsection{Adaptive Update Checking}\label{subsubsec:adaptive-check}
To avoid unnecessary recomputation, we adopt an adaptive strategy that triggers graph updates only if significant changes in importance scores are detected. We periodically evaluate the need for recomputation based on the discrepancy between historical and current importance scores. When the training steps reach the check interval $t_c$, we first uniformly sample a subset $\mathcal{S}_t$ from the trained data in the previous steps and evaluate their current importance scores for comparison. For samples $X_i \in \mathcal{S}_t$ with historical scores at previous $c$ steps, we define the average discrepancy score of $\mathcal{S}_t$ that reflects the variation in importance scores as:
\begin{equation}\label{eq:check}
  \Delta^I=\frac{1}{|\mathcal{S}_{t}|}\sum_{X_i\in\mathcal{S}_{t}}\frac{\sum_{t_k=0}^{t_c} \lambda_c^{t_k - t_c}|I(X_i)_{t_k}-I(X_i)_{t_c}|}{\sum_{t_k=0}^{t_c}\lambda_c^{t_k - t_c}I(X_i)_{t_c}}
\end{equation}
\noindent where $\lambda_c=0.99$ is an exponential decay factor that emphasizes recent changes. $I(X_i)$ is the importance score in Equation~\eqref{eq:el2n}. We then compare it to a predefined threshold $\delta$.
Once the discrepancy exceeds the threshold, it means a significant shift in training dynamics, and we start the graph updates. During training, we compute and cache importance scores at each training step.

\subsubsection{Selective Recomputation}\label{subsubsec:selective-recom}

We first adopt a selective recomputation strategy guided by two key principles: (1) reusing historical importance scores for samples whose scores are likely to remain valid during training, and (2) prioritizing representative samples that are located at the boundaries of homogeneous regions in the graph. First, we initially filter out candidates that have recent historical scores, as they are less likely to require updating. To quantify this, we define the staleness score as $\text{sta}(X_i)=t-t_{history}$, which indicates the number of steps since the last update for $X_i$. $t_{history}$ is the training step of the latest update.

Secondly, we also want to de-prioritize samples that are highly similar to their neighbors and prioritize samples that lie at the edge of homogeneous regions. We define the uniqueness score as $\text{uni}(X_i)=|I(X_i)-\frac{\sum_{j\in\mathcal{N}(X_i)}w_{ij}I(X_j)}{\sum_{j\in\mathcal{N}(X_i)}w_{ij}}|$, where $w_{ij}$ is the edge weight in Definition~\ref{define:knngraph}. $\text{uni}(X_i)$ captures the deviation of $X_i$'s current importance score from the weighted average of its neighborhood. Larger $uni(X_i)$ means the node sits in a decision boundary region where a correction would affect many other points.
Combining two scores, we have the selective update determination score as:
\begin{equation}\label{eq:update_select}
  s_{update}(X_i)=\text{uni}(X_i)\cdot\text{sta}(X_i)
\end{equation}
The larger the $s_{update}$, the more strongly the sample is prioritized to recompute the importance score. To ensure diversity and avoid over-concentration in dense regions, we apply a simple greedy algorithm to iteratively select samples with the largest $s_{update}(X_i)$, while simultaneously avoiding the selection of neighbors of already selected samples. This selection process is repeated until a predefined budget $k_{recal}$ is reached, forming a subset of samples $\mathcal{D}_{recal}$ for score and embedding recomputation.

\subsubsection{Approximate Graph Updating}\label{subsubsec:update}

After selective recomputation, we have new scores and representations of the selected samples. We first update the scores via approximate updating and further update the embeddings and the $k$-NN graph.

We propagate updated importance scores to the remaining outdated samples using their local neighborhood information. Given the updated scores $I^\text{new}(X_{i})$ for each node $X_i$ in $G$, we identify the nodes whose neighbors have updated scores and are sufficiently similar to their updated neighbors. Concretely, For every non-anchor node $X_i \notin \mathcal{D}_{\text{recal}}$, we consider only those graph neighbors $X_j$ that are in $D_{\text{recal}}$ and discard neighbors that are far in representation space or that have historically disagreed in difficulty. We define the affinity between $X_i$ and its neighbor $X_j$ as:
\begin{equation}\label{eq:aff}
  \text{aff}(i,j)=w_{ij}\cdot \exp(-0.1|I(X_i)-I(X_j)|)
\end{equation}
\noindent where $w_{ij}$ is the edge weight in Definition~\ref{define:knngraph}. If any of the neighbors has $\text{aff}(i,j)>\beta$ for a threshold $\beta$, we include it for propagation.

Let $\mathcal{N}_i^\star=\{j\in\mathcal{N}(X_i), \text{aff}(i,j)>\beta\}$ be the set of neighbors of $X_i$
that passed the threshold.
We update $X_i$'s score by solving a simple quadratic objective. This objective balances its old score with those of its most similar neighbors as follows:

\begin{align*}
  E(X_i) = \frac{1}{2} \big(I^{new}(X_i) - I(X_i)\big)^2 \nonumber + \frac{1}{2}
  \sum_{j \in \mathcal{N}_i^\star}
  w_{ij}\big(I^{new}(X_i) - I^{new}(X_j)\big)^2
\end{align*}

where $I^{\text{new}}(X_j)$ is the recomputed score of neighbor $X_j $, and $w_{ij}$ is the edge weight.  Define $\alpha_{ij}=\frac{w_{ij}}{\sum_{j\in\mathcal{N}_i^\star}w_{ij}}, \sum_{j\in\mathcal{N}_i^\star}\alpha_{ij}=1$, the closed-form solution of this strictly convex objective is exactly the following weighted average:

\begin{align}\label{eq:update}
  I^{new}(X_i)= & \frac{1}{2} I(X_i) + \frac{1}{2}\alpha_{ij}I^{\text{new}}(X_j)
\end{align}
Then, we apply this update to the neighbors of each $X_i$ that satisfy the affinity condition.
After the update is complete, we reconstruct the coreset based on the updated scores and resume training.

Instead of using the same procedure for updating the representation, we have an additional check for the embedding update. Given the old embeddings as well as the recomputed embeddings in Section~\ref{subsubsec:selective-recom}, we check the average embedding shift for these embeddings by:
\begin{equation}\label{eq:check_rep}
  \Delta^\mathbf{H}=\frac{1}{|\mathcal{D}_{recal}|}\sum_{X_i\in\mathcal{D}_{recal}}\|\bar{\mathbf{H}}_i-\bar{\mathbf{H}}_i^{new}\|
\end{equation}
where $\bar{\mathbf{H}}_i\in\mathbb{R}^d$ denotes the mean hidden state representation obtained by $\bar{\mathbf{H}}_i=\frac{1}{T}\sum_{t=1}^T\mathbf{H}_i[t]$, where $\mathbf{H}_i \in\mathbb{R}^{T\times d}$ represents the hidden states from the model.
The score is then compared with a given threshold $\delta_h$. If the shift is insignificant, we skip the embedding update and retain the graph unchanged. Otherwise, we perform the representation update using the same strategy as the scores. For each embedding recomputed, we directly replace the embeddings with the new embeddings. We adopt the same similarity check as Equation~\eqref{eq:check} and propagate the embedding via:
\begin{equation}\label{eq:update_rep}
  \bar{\mathbf{H}}_i^{new}=\frac{1}{2} \bar{\mathbf{H}}_i + \frac{1}{2}\frac{\sum_{j\in\mathcal{N}(X_i), \text{aff}(i,j)>\beta}w_{ij}\bar{\mathbf{H}}_j}{\sum_{j\in\mathcal{N}(X_i),\text{aff}(i,j)>\beta}w_{ij}}
\end{equation}

Updating the representations affects the previously constructed $k$-NN graph. However, rebuilding the entire graph after each update would incur $O(n^2d)$ cost and becomes unnecessary when only a subset of nodes changes. Therefore, GRACE performs a local graph repair strategy that updates only the changed nodes and their affected neighborhoods. Specifically, we use an LSH-based approximate neighbor index to retrieve a small candidate set for updated embeddings, and then refine the final neighbors by exact distance computation. In this way, the graph can be maintained efficiently without recomputing the full $k$-NN structure.

\begin{theorem}[Bounded Error of the Selective Update, Short version]\label{thm:error-bound-short}
  Consider the update round as defined in Equation~\eqref{eq:update}.
  Then for every $X_i$, we have \textbf{Change Stability}: The per-round change of any node is bounded, and \textbf{Local Approximation Error bound}: The local approximation methods error is bounded.
\end{theorem}

\begin{proof}[Proof Sketch]
  Due to space limit, we put the full proof in our Appendix.
\end{proof}

Theorem~\ref{thm:error-bound-short} supports our approximation strategy by showing that selective updates through similar neighbors keep the update error controlled. Therefore, our graph-based selective update can maintain locally accurate importance scores and embeddings while recomputing only a subset of samples.

\noindent\underline{\textbf{Time Complexity Analysis.}}
Suppose the LLM has $L$ layers with a hidden dimension of $d$, the average data length is $T$, the data size is $N$, the update checking size is $|\mathcal{S}_t|$, and the recomputation size is $|\mathcal{D}_{recal}|$. The time complexity is $O(|\mathcal{S}_t|L(d^2T+dT^2))$ for adaptive checking, $O(N + |\mathcal{D}_{recal}|L(d^2T+dT^2))$ for selective recomputation, $O(|\mathcal{D}_{recal}|nd)$ for score and embedding propagation, and $O(|\mathcal{D}_{recal}|(\log N + nd))$ for $k$-NN updating. In summary, the time complexity for Algorithm~\ref{algo:dynamic} is $O(|\mathcal{S}_t||\mathcal{D}_{recal}|L(d^2T+dT^2)+N+|\mathcal{D}_{recal}|(\log N + 2nd))$.

\section{Experiments}\label{sec:experiments}

In this section, we present a comprehensive empirical evaluation of GRACE. We first introduce the experimental settings in Section~\ref{subsec:exp-setting}, including datasets, tasks, model backbones, baseline methods, evaluation metrics, and implementation details. We then report the main results on MathInstruct in Section~\ref{subsec:main-exp}, analyzing both effectiveness and end-to-end efficiency. Next, we conduct ablation studies to isolate the contributions of key components in GRACE in Section~\ref{subsec:ablation} and perform a systematic sensitivity analysis under multiple selection budgets, balance coefficients, and update schedules in Section~\ref{subsec:sensitivity}. Finally, in Section~\ref{subsec:addition-exp}, we evaluate GRACE on two additional instruction-tuning benchmarks to assess its generalization across domains and tasks.

\subsection{Experimental Settings}\label{subsec:exp-setting}

\subsubsection{Datasets and Evaluation Procedure}
We use three benchmarks from different tasks and fields: MathInstruct~\cite{yue2024mammoth}, BioInstruct~\cite{10.1093/jamia/ocae122}, and DialogSum~\cite{chen-etal-2021-dialogsum} for evaluation.

\noindent\underline{\textit{\textbf{MathInstruct}}}. MathInstruct~\cite{yue2024mammoth} contains 262K high-quality instances for fine-tuning on math problems, which are constructed from 14 different math-related data sources to provide wide coverage of various math fields and difficulty levels. It contains several tasks related to mathematics, including code generation, reasoning, and question answering.

Following the settings of~\cite{yue2024mammoth}, we evaluate the performance on both in-domain and out-of-domain open-ended generation tasks from the MathInstruct~\cite{yue2024mammoth} benchmark. Here, the in-domain tasks consist of test sets from three math datasets whose training portions are included in the MathInstruct training mixture, while the out-of-domain tasks consist of three standalone math reasoning benchmarks that are never used for training and are only used for evaluation.

In-domain tasks primarily evaluate the training quality on training the distribution, while out-of-domain tasks assess the model’s ability to generalize beyond the training distribution. Together, these benchmarks offer a comprehensive assessment across a range of mathematical problems, such as arithmetic, algebra, and commonsense reasoning.

For evaluation, following~\cite{yangSmallToLargeS2LScalable2024,yue2024mammoth}, each evaluation task is assessed using \textit{\textbf{exact match accuracy}}, which measures whether the final output exactly matches the ground-truth solution.

We use the same procedure as the original paper~\cite{yue2024mammoth}, where all evaluations are conducted under the zero-shot setting and adopt the Program-of-Thought prompting approach~\cite{chen2021evaluating} as the primary strategy. If the generated code is not executable, we fall back to Chain-of-Thought prompting~\cite{wei2022chain}. This encourages the model to give numerical answers directly after reasoning steps, imitating a human's step-by-step reasoning process.

\noindent\underline{\textit{\textbf{BioInstruct and DialogSum}}}. In Section~\ref{subsec:addition-exp}, we also conduct experiments on two additional datasets. BioInstruct~\cite{10.1093/jamia/ocae122} contains biomedical question answering problems. With an instruction-input-output triplet format, the dataset evaluates a wide range of scenarios in medical and biomedical question answering tasks, such as diagnostic analysis and clinical decision making. DialogSum~\cite{chen-etal-2021-dialogsum} contains comprehensive dialog samples of daily-life scenarios extracted from open dialog repositories, each paired with a corresponding summarization reference. We use the split settings from~\cite{wang2025datawhispererefficientdata} that perform a $9:1$ train–test split on datasets.

For BioInstruct and DialogSum datasets, we use the same generation prompt at evaluation as in training with a zero-shot setting. We evaluate the generated output using the ROUGE-L metric~\cite{lin-2004-rouge}, which measures the comprehensiveness of answers compared to human references. Typically, we measure the F1 score of the ROUGE-L metric. The higher the F1 score, The better the QA answers or summarized text match the reference.

\subsubsection{Models and Baselines}

We evaluate GRACE on supervised fine-tuning tasks with LoRA~\cite{hu2022lora} using Phi-2~\cite{li2023textbooks}, Llama-2-7b~\cite{touvron2023llama2openfoundation} and Qwen2.5-7b~\cite{qwen2.5} models. We compare GRACE with several baselines that employ diverse selection methods and strategies. We set up the following selection types: static and dynamic.

\noindent\underline{\textit{\textbf{Static selection:}}}
Static selection methods select one fixed coreset for the entire training procedure. We compare with the following static selection baselines:
\begin{itemize}[leftmargin=*]
  \item \textbf{Random}. It randomly samples a subset from the training dataset.
  \item \textbf{Middle Perplexity (MP)}~\cite{anknerPerplexedPerplexityPerplexityBased2024,marionWhenLessMore2023}. It selects samples with loss values that are closest to the median of all loss values.
  \item \textbf{Facility Location (FL)}~\cite{bhattExperimentalDesignFramework2024}. It selects samples to maximize similarity to all samples in the embedding space by solving a facility location problem.
  \item \textbf{DivIR}~\cite{zhouLoBaSSGaugingLearnability2023}. It selects samples that have the largest reduction in loss value after warm-up training.
  \item \textbf{Long4Align}~\cite{zhaoLongMoreAlignment2024}. It selects samples with the longest sequential completion lengths.
  \item \textbf{GradN}~\cite{paulDeepLearningData2021}. It selects samples with the largest gradient values from the last layer.
  \item \textbf{Least Confidence (LC)} \cite{bhattExperimentalDesignFramework2024}. It selects samples that have the smallest output probabilities of tokens multiplied over the entire generated sequences.
  \item \textbf{TAGCOS}~\cite{zhangTAGCOSTaskagnosticGradient2024}. It uses Orthogonal Matching Pursuit to solve the gradient matching problem~\eqref{eq:gradient_matching} directly to select a coreset.
\end{itemize}

\begin{table*}[ht]
  \centering
  \caption{MathInstruct results with 10\% training budget. We report in-domain (GSM8K, MATH, NumGlue) and out-of-domain (SVAMP, DeepMind, SimuLeq) average accuracies (\%). }
  \label{tab:mathinstruct-10-all}
  \resizebox{\linewidth}{!}{%
    \begin{tabular}{c|c|cc|cc|cc}
      \hline
      \multirow{2}{*}{Type} & \multirow{2}{*}{\textbf{Methods}}
      & \multicolumn{2}{c|}{\textbf{Phi-2}}
      & \multicolumn{2}{c|}{\textbf{Llama2-7B}}
      & \multicolumn{2}{c}{\textbf{Qwen2.5-7B}} \\
      \cline{3-8}
      & & In-domain & Out-domain & In-domain & Out-domain & In-domain & Out-domain \\
      \hline
      - & Base (Pretrained Model)
      & 19.66 & 22.96 & 3.39 & 2.56 & 47.73 & 58.39 \\
      \hline
      \multirow{8}{*}{Static}
      & Random                         & 48.35 & 50.79 & 22.66 & 17.57 & 60.66 & 63.39 \\
      & Middle Perplexity              & 47.46 & 47.55 & 22.90 & 18.33 & 61.02 & 62.07 \\
      & Facility Location              & 46.16 & 48.89 & 21.01 & 15.54 & 60.66 & 65.26 \\
      & DivIR                          & 47.58 & 52.16 & 23.33 & 21.90 & 61.34 & 64.23 \\
      & Long4Align                     & 42.54 & 36.11 & 15.16 & 10.61 & 59.91 & 64.96 \\
      & GradN                          & 43.00 & 39.89 & 19.30 & 15.18 & 60.54 & 60.64 \\
      & Least Confidence               & 46.93 & 45.49 & 22.34 & 18.13 & 61.47 & 59.57 \\
      & TAGCOS                         & 43.59 & 42.94 & 19.44 & 16.12 & 56.24 & 54.64 \\
      \hline
      \multirow{5}{*}{Dynamic}
      & Fix Interval Random            & 47.00 & 50.07 & 20.65 & 18.80 & 60.06 & 61.14 \\
      & Fix Interval Middle Perplexity & 46.92 & 47.31 & 22.24 & 17.05 & 60.17 & 61.41 \\
      & Fix Interval Facility Location & 47.38 & 48.88 & 22.13 & 17.44 & 61.23 & 61.57 \\
      & kMQ                            & 45.90 & 48.26 & 22.47 & 16.21 & 60.85 & 59.73 \\
      & GRACE                          & \textbf{50.54} & \textbf{55.27}
      & \textbf{23.89} & \textbf{24.05}
      & \textbf{62.19} & \textbf{67.12} \\
      \hline
    \end{tabular}
  }
\end{table*}

\noindent\underline{\textit{\textbf{Dynamic selection:}}}
Dynamic selection methods reconstruct a new coreset before the start of every epoch. We establish the following baselines for comparison with dynamic selection methods:
\begin{itemize}[leftmargin=*]
  \item \textbf{Fix Interval Random}. It randomly samples a subset from the training dataset before starting each epoch.
  \item \textbf{Fix Interval Middle Perplexity (DynamicMP)}. Before starting each epoch, it recalculates the loss values of all samples and selects samples with loss values that are closest to the median of all values.
  \item \textbf{Fix Interval Facility Location (DynamicFL)}. Before starting each epoch, it re-extracts the representations for all data samples and selects samples to maximize representation coverage over the embeddings by solving a facility location problem.
  \item \textbf{kMQ}~\cite{yuDiversifyConquerDiversityCentric2024}. It first extracts embeddings for the data and clusters the embeddings into several clusters. Before each epoch, it randomly samples from clusters with weights and uses the scores to adjust the weights of each cluster before the next selection.
\end{itemize}

\subsubsection{Implementation Details}

For a fair comparison, all methods are compared under the same constrained budget of training steps, such as 10\% of the full training steps. Static selection methods first warm up the model for 10\% of the full training steps, then use it to extract scores, embeddings, or gradient vectors, and select a single coreset of size $b$. Dynamic selection methods first extract scores or embeddings using the current model before each epoch, and then build a coreset for the upcoming epoch. The selection size at each update is bounded by $b$.

For model training, we apply LoRA~\cite{hu2022lora} with rank 128, scale factor 512, and dropout rate 0.05. Models are trained for 3 epochs with a batch size of 128. We use learning rates of 2e-5 for Phi-2 and 4e-5 for Llama-2-7b and Qwen2.5-7b, with a linear learning rate scheduler with a 3\% warm-up stage. In GRACE, we use FAISS~\cite{johnson2019billion} to construct a $k$-NN graph and build the LSH index with $k=10$. The sensitivity of the selection budget $\eta$, the balance control hyperparameter $\lambda$, the check threshold $\delta$, and the check interval $t_c$ will be discussed in Section~\ref{subsec:sensitivity}.

\subsection{Main Results}\label{subsec:main-exp}

\subsubsection{Effectiveness Evaluation}

Table~\ref{tab:mathinstruct-10-all} summarizes the in-domain and out-of-domain average accuracy on MathInstruct under a 10\% training budget for all selection strategies. GRACE is the only method that consistently achieves the best in-domain and out-of-domain averages across all three backbones, showing that our coreset construction transfers well across different model architectures and task distributions.

Compared to static coreset baselines, uncertainty-based methods such as Middle Perplexity, GradN, and Least Confidence, as well as gradient-based TAGCOS, exhibit substantial performance variation across models and between in-domain and out-of-domain tasks. This suggests that their signals do not generalize robustly. This also shows that relying on a single indicator is insufficient to provide stable gains in the multitask setting. Dynamic methods that naively recalculate static criteria at fixed intervals do not uniformly improve upon their corresponding static versions and, in some cases, may even degrade performance, indicating that simple updates can introduce noise into the training dynamics. Overall, GRACE outperforms both static and dynamic alternatives on average across all evaluation regimes.

\subsubsection{Efficiency Evaluation}
We measure and report the end-to-end running time for all methods under a consistent 10\% training budget on the MathInstruct dataset, using the same hardware and optimizer settings. The reported time combines three stages: (1) the warm-up phase used by baseline methods that require an initial model to extract indicators, (2) the indicator extraction step for all scores and embeddings needed by the baselines and GRACE, and (3) the subsequent model-training phase, which includes on-the-fly indicator recalculation and coreset updates for dynamic methods.

As Figure~\ref{fig:time} shows, the uncertainty-based and the geometric-based static methods are the fastest overall because they perform feature or score extraction only once at the beginning and then train on a fixed coreset without further updates. However, the gradient-based method TAGCOS uses noticeably more time than all other methods because it needs to compute and store full gradient vectors over the training set, which is substantially more expensive than extracting forward-pass indicators. For dynamic baselines, the methods that require full feature or score extraction incur higher runtime, as they need to refresh their coresets at fixed intervals. Despite adding extra overhead through warm-up and adaptive updates over static settings, GRACE is still more efficient than entirely dynamic strategies, which are costly due to repeated recomputation and coreset re-selection. This highlights that GRACE achieves a well-balanced trade-off between dynamic settings and computational efficiency.

\begin{figure}[t]
  \centering
  \includegraphics[width=\linewidth]{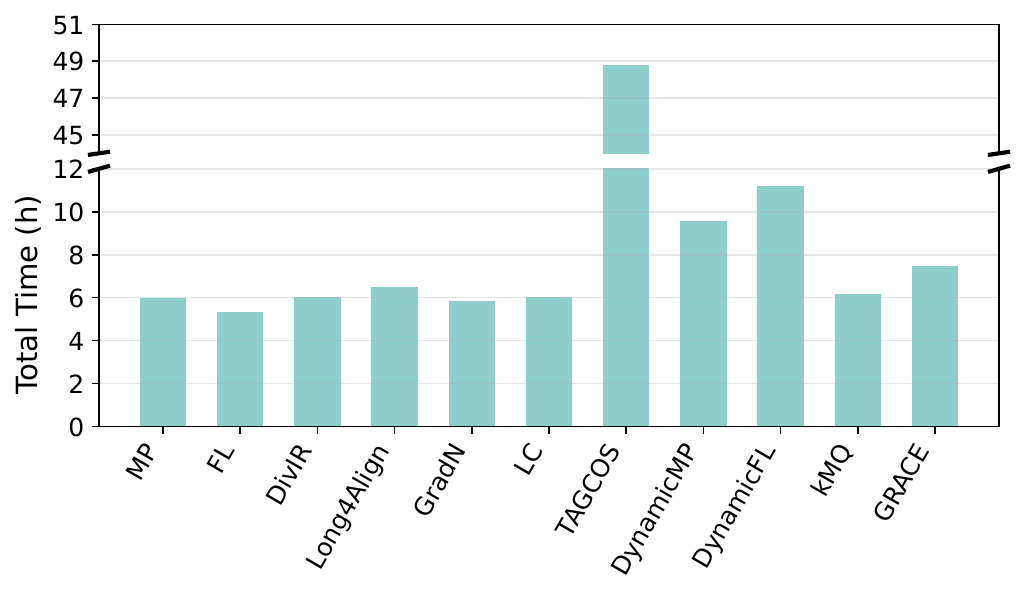}
  \setlength{\abovecaptionskip}{-10pt}
  \caption{Time Comparison}
  \label{fig:time}
  \vspace{-1em}
\end{figure}

\begin{figure*}

  \centering
  \includegraphics[width=\linewidth]{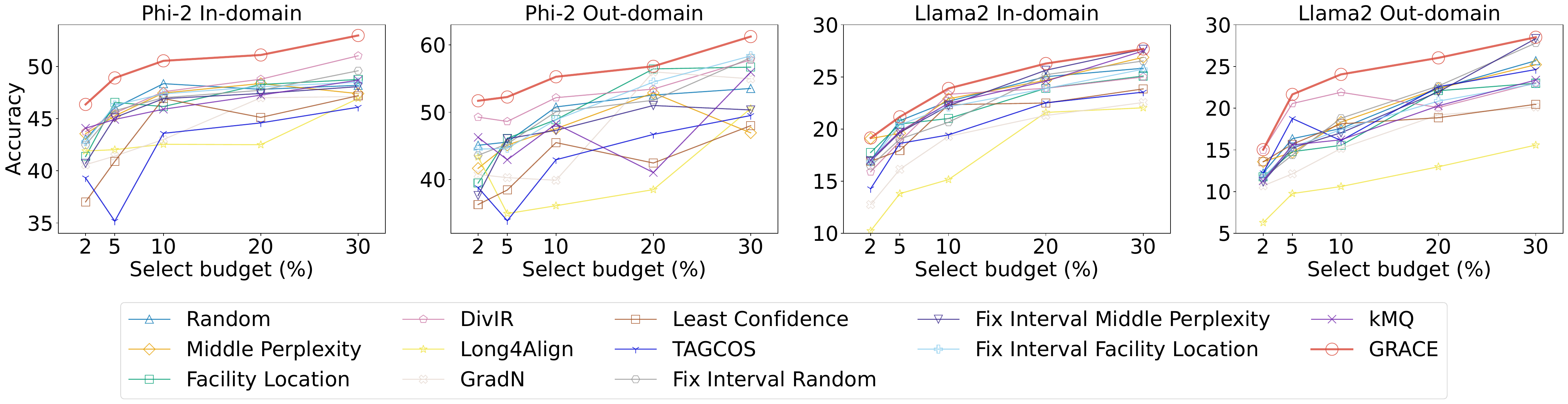}
  \vspace{-2em}
  \caption{Experiment results under different budgets on MathInstruct}
  \label{fig:math}
  \vspace{-1em}
\end{figure*}

\begin{figure}[tbp]
  \centering
  \begin{minipage}[t]{0.45\linewidth}
    \centering
    \includegraphics[width=\linewidth]{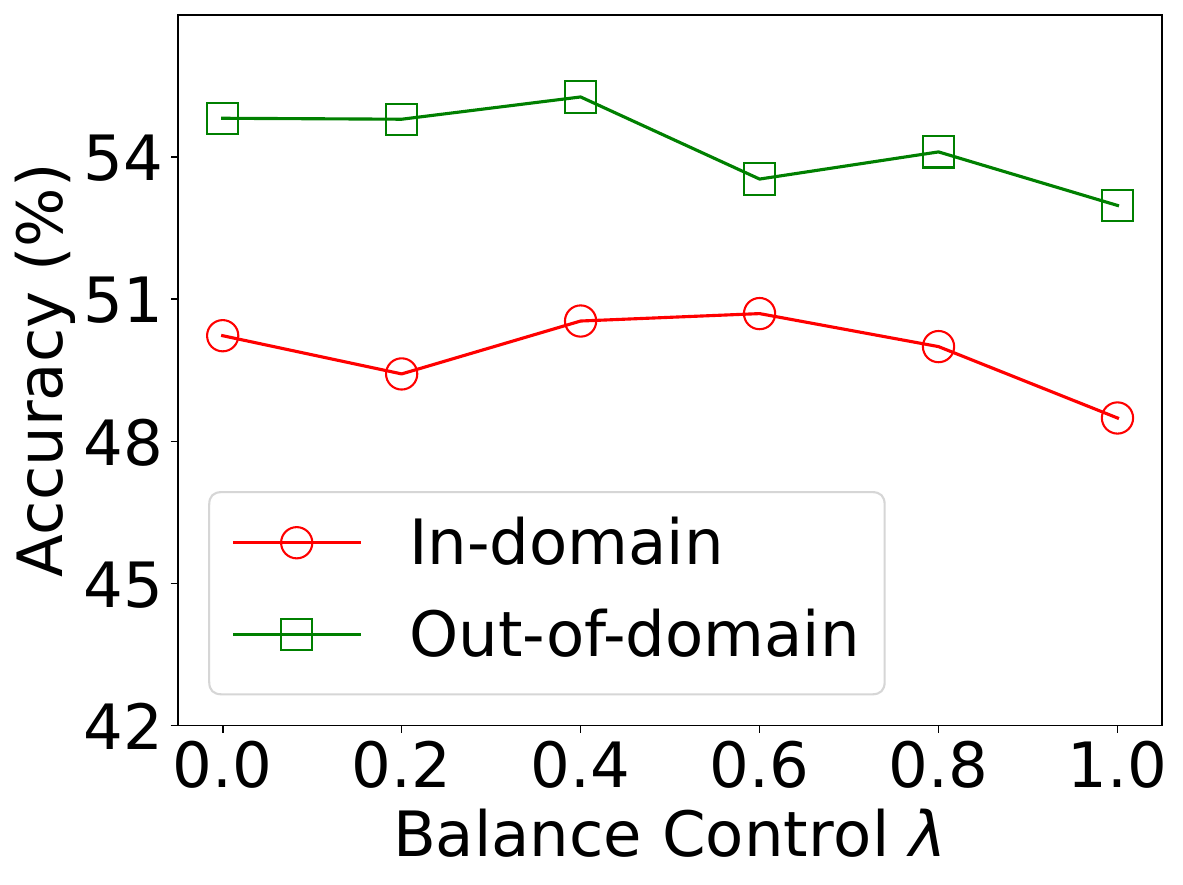}
    \vspace{-2em}
    \caption{ $\lambda$ for Phi-2}
    \label{fig:lambda-phi}
  \end{minipage}
  \hfill
  \begin{minipage}[t]{0.45\linewidth}
    \centering
    \includegraphics[width=\linewidth]{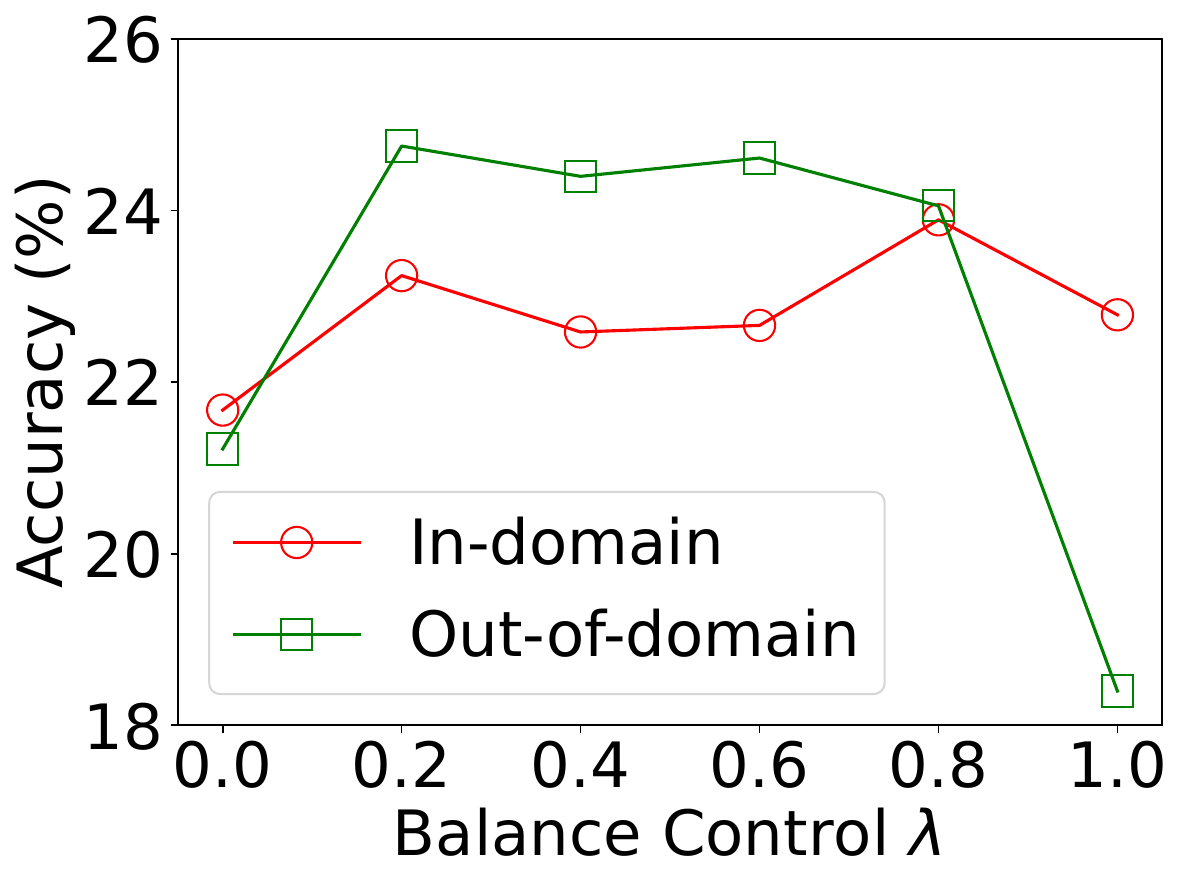}
    \vspace{-2em}
    \caption{ $\lambda$ for Llama2-7b}
    \label{fig:lambda-llama}
  \end{minipage}
  \label{fig:lambda}
  \vspace{-1em}
\end{figure}

\subsection{Ablation Study}\label{subsec:ablation}

In this section, we apply several ablation variants of GRACE to test the components and designs:
\begin{itemize}[leftmargin=*]
  \item \textbf{GRACE\textbackslash R}: GRACE excludes the representation score $R(\cdot)$ and only uses the importance score $\hat{I}(\cdot)$.
  \item \textbf{GRACE\textbackslash I}: GRACE excludes the importance score $\hat{I}(\cdot)$ and only uses the representation score $R(\cdot)$.
  \item \textbf{GRACE\textbackslash Update}: GRACE without any coreset update strategy.
  \item \textbf{GRACE\textbackslash AdaUpdate}: GRACE without the adaptive coreset update strategy. Coreset is updated at every interval.
\end{itemize}

As shown in Table~\ref{tab:ablation}, \textbf{GRACE\textbackslash R} suffers from a performance drop in generalization according to the performance decrease on out-of-domain tasks. \textbf{GRACE\textbackslash I} results in an even sharper performance drop, showing that informativeness is a critical signal for coreset selection. Using both components jointly achieves the highest overall accuracy, confirming their complementarity.

In addition, incorporating the update strategy is essential for coreset effectiveness, as \textbf{GRACE\textbackslash Update} shows significant performance degradation, particularly on out-of-domain tasks. On the other hand, \textbf{GRACE\textbackslash AdaUpdate} yields only marginal improvements or even worse results compared to the adaptive strategy. GRACE adopts an adaptive strategy triggered only when training dynamics shift significantly and can achieve the best performance, confirming that well-timed, signal-driven updates are more effective than frequent but unconditioned ones.

\begin{table}[ht]
  \caption{Ablation Experiments to test the Effect of Scores. \textit{In} means in-domain tasks, and \textit{Out} means out-of-domain tasks.}
  \label{tab:ablation}
  \begin{tabular}{l|c|c|c|c}
    \hline
    \multirow{2}{*}{Model} & \multicolumn{2}{c|}{Phi-2}      & \multicolumn{2}{c}{Llama-2}                  \\
    \cline{2-5}
    & In     & Out     & In                & Out      \\
    \hline
    \textbf{G$\cdot$\textbackslash I}            & 48.49          & 52.98          & 22.78                     & 18.40           \\
    \textbf{G$\cdot$\textbackslash R}            & 50.23          & 54.82          & 21.67                     & 21.22           \\
    \textbf{G$\cdot$\textbackslash Update}        & 48.79          & 51.54          & 23.42          & 23.62           \\
    \textbf{G$\cdot$\textbackslash AdaUpdate}     & 50.08          & 53.03          & 23.19          & 22.57           \\
    \hline
    \textbf{GRACE}                  & \textbf{50.54} & \textbf{55.27} & \textbf{23.89}            & \textbf{24.05}  \\
    \hline
  \end{tabular}
  \vspace{-1em}
\end{table}
\subsection{Parameter Sensitivity}\label{subsec:sensitivity}
In this section, we systematically probe key hyper-parameters of GRACE, including the training budget $\eta$, representation-importance balance $\lambda$, adaptive checking threshold $\delta$, and checking interval $t_c$, to reveal how each one influences performance. We conduct sensitivity experiments on the Phi-2 and Llama2-7b models.

\subsubsection{Selection budget $\eta$}
Figure~\ref{fig:math} reports the accuracy of GRACE and all baseline methods on Phi-2 and Llama2-7b under selection budgets $\eta\in\{2\%,5\%,10\%,20\%,30\%\}$, for both in-domain and out-of-domain tasks. Across models, domains, and budgets, GRACE is consistently the best performing methods, demonstrating its robustness to budget constraints. In the extremely low-budget regime, where many static and dynamic baselines suffer sharp drops in performance, GRACE degrades much more gracefully and still secures substantial gains, showing its effectiveness under severe training constraints. Compared to static methods, GRACE achieves a more stable improvement as the budget increases, suggesting that dynamic settings help to better align the selected coreset with the evolving model needs. While dynamic baselines do benefit from periodic updates, they still lag behind GRACE in most settings, and sometimes they show fluctuations in performance when the budget increases. This indicates that while periodic updates do help selection strategies improve performance, they may suffer from problems due to simple update schedules or less effective selection criteria. Overall, the figure illustrates that GRACE combines strong performance across different budgets with stable, monotonic gains as more data is allowed, achieving robust coreset quality.

\subsubsection{Balance Control $\lambda$}
We vary the balancing coefficient $\lambda$ from Equation~\eqref{eq:target} as $\lambda\in \{0, 0.2, 0.4, 0.6, 0.8, 1\}$ and plot the relationship between $\lambda$ and performance on both Phi-2 and Llama2-7b models, referring to Figure~\ref{fig:lambda-phi} and Figure~\ref{fig:lambda-llama}. Both models exhibit a non-monotonic trend, where the performance initially improves when $\lambda$ increases from $0$ and peaks around intermediate values before degrading as $\lambda$ approaches 1. This pattern confirms that neither component alone is sufficient. Typically, Llama2-7b shows sharper performance variance according to the change in $\lambda$, which might be because larger models may be more sensitive to the balance between diversity and informativeness due to their higher capacity to exploit subtle distributional coverage.

\begin{figure}[tbp]
  \centering
  \begin{minipage}[t]{0.45\linewidth}
    \centering
    \includegraphics[width=\linewidth]{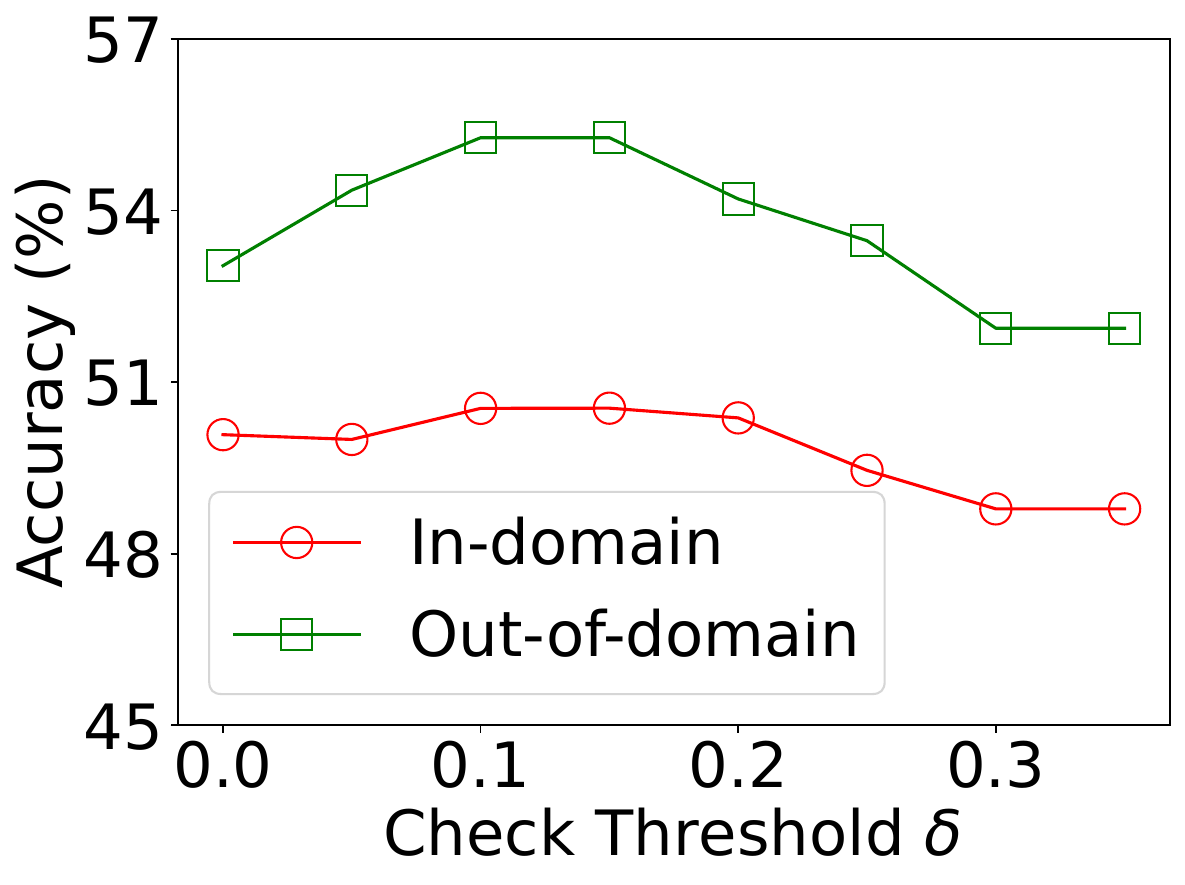}
    \vspace{-2em}
    \caption{$\delta$ for Phi-2}
    \label{fig:delta-phi}
  \end{minipage}
  \hfill
  \begin{minipage}[t]{0.45\linewidth}
    \centering
    \includegraphics[width=\linewidth]{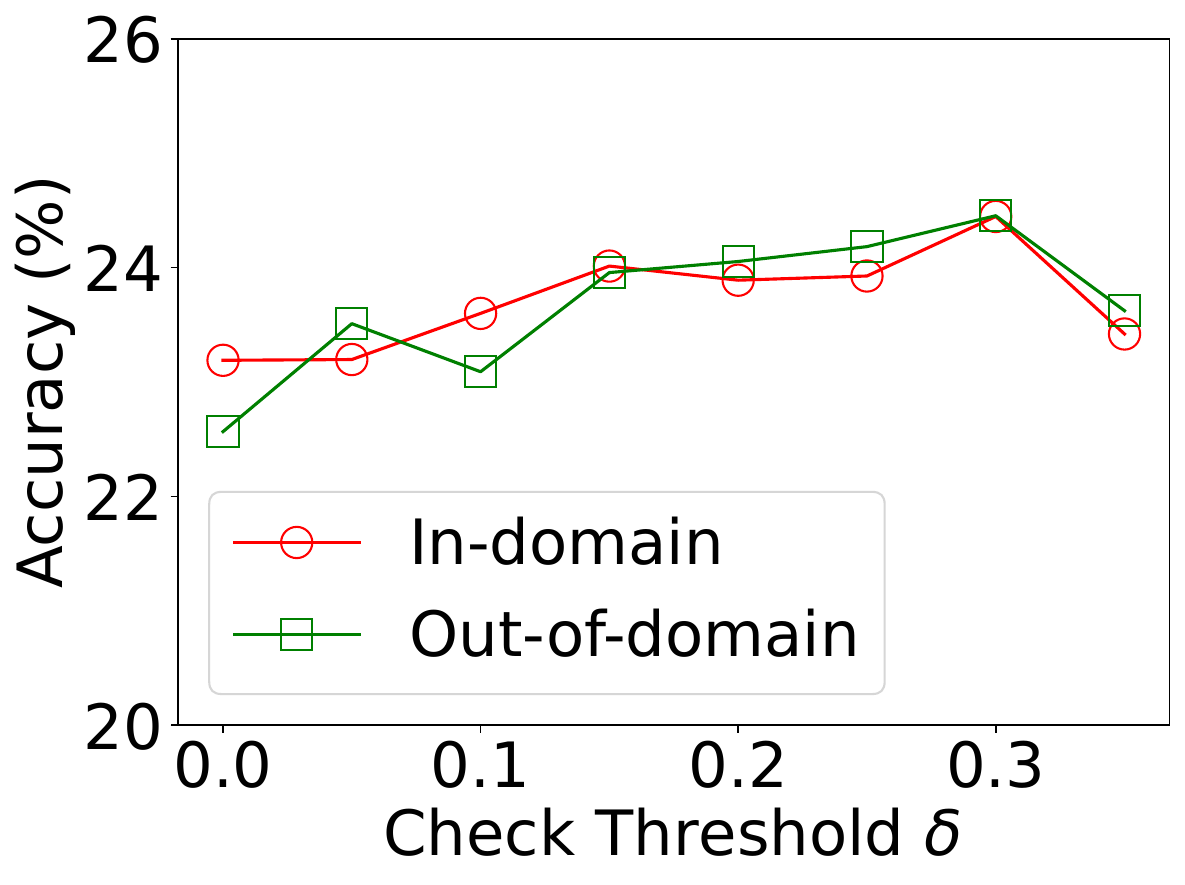}
    \vspace{-2em}
    \caption{$\delta$ for Llama2-7b}
    \label{fig:delta-llama}
  \end{minipage}
  \label{fig:delta}
  \vspace{-1em}
\end{figure}

\begin{figure}[tbp]
  \centering
  \begin{minipage}[t]{0.45\linewidth}
    \centering
    \includegraphics[width=\linewidth]{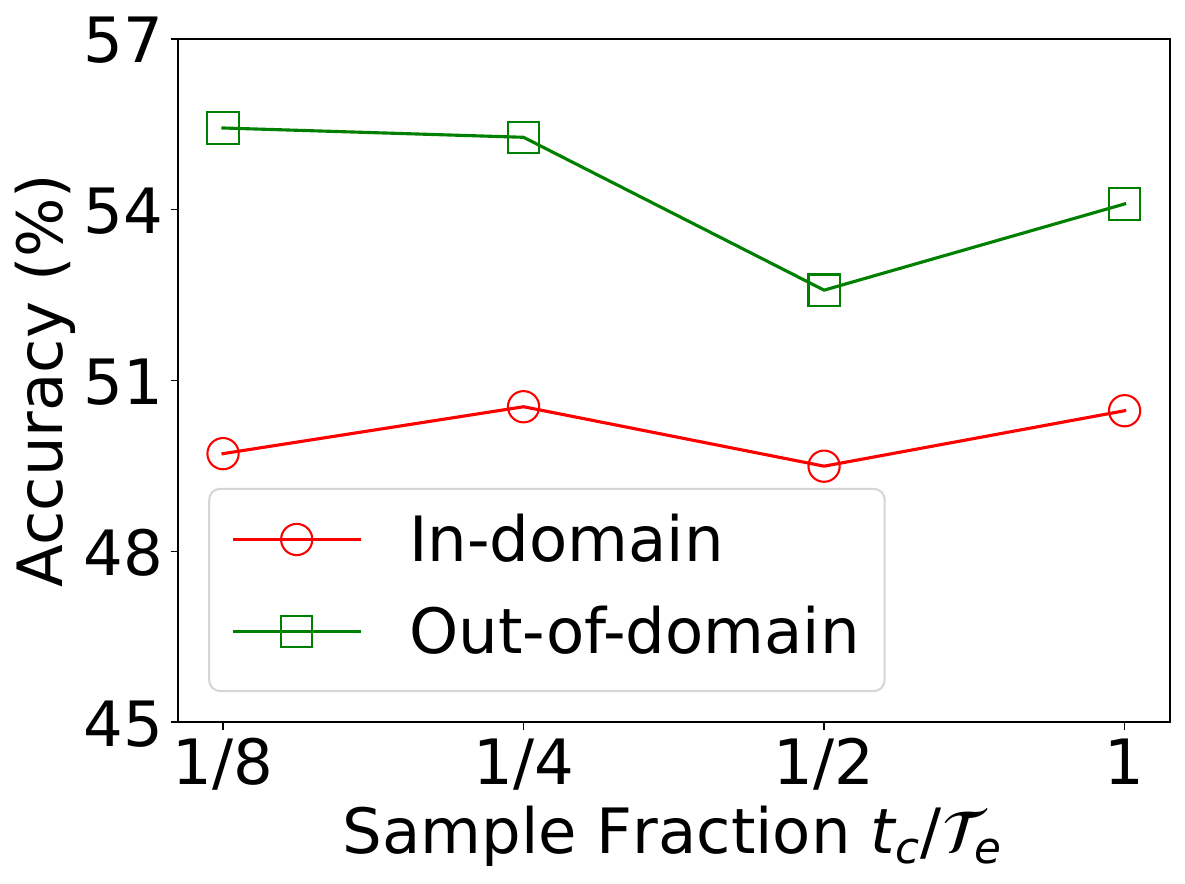}
    \vspace{-2em}
    \caption{$t_c/\mathcal{T}_e$ for Phi-2}
    \label{fig:tc-phi}
  \end{minipage}
  \hfill
  \begin{minipage}[t]{0.45\linewidth}
    \centering
    \includegraphics[width=\linewidth]{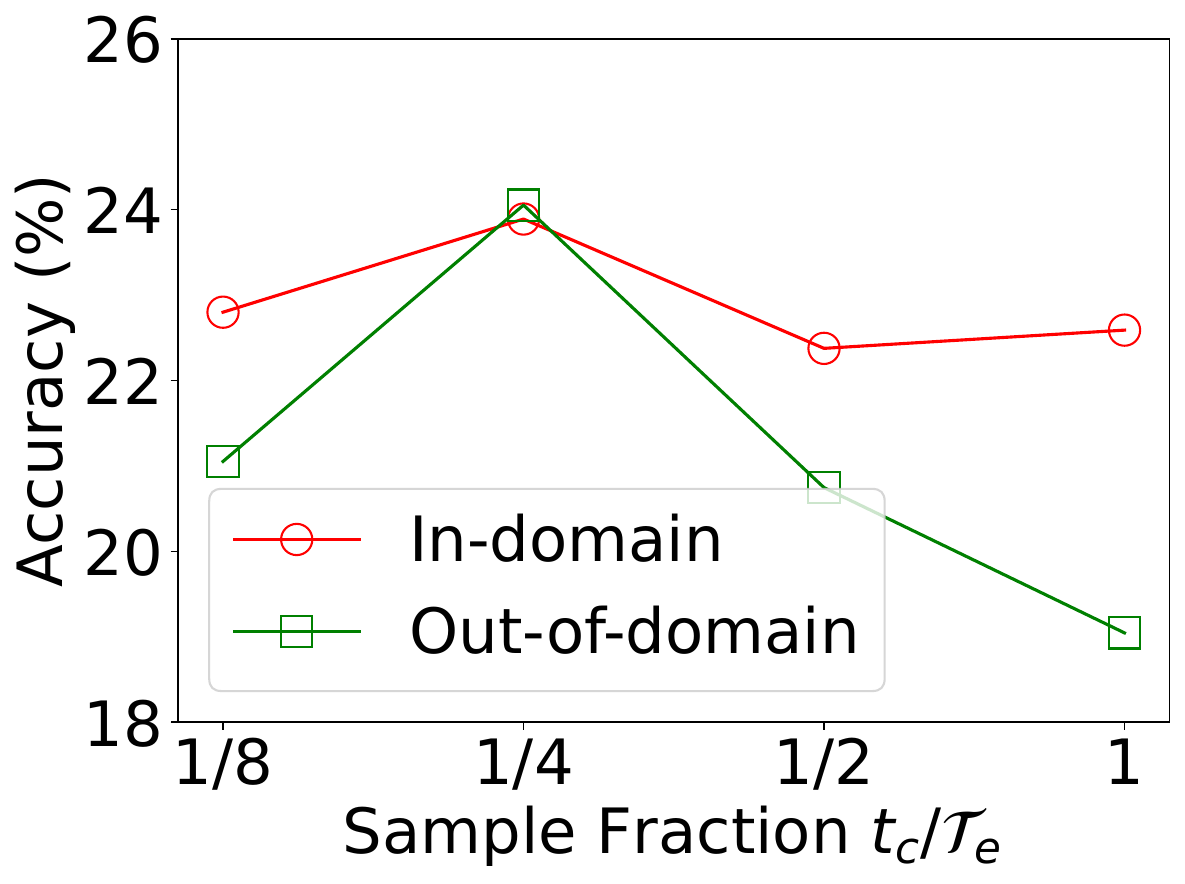}
    \vspace{-2em}
    \caption{$t_c/\mathcal{T}_e$ for Llama2}
    \label{fig:tc-llama}
  \end{minipage}
  \label{fig:tc}
  \vspace{-1em}
\end{figure}

\begin{table*}
  \centering
  \caption{ROUGE-L results on BioInstruct}
  \label{tab:bioinstruct}
  \resizebox{\linewidth}{!}{%
    \begin{tabular}{c|c|c|c|l|c|c|l|c|c|l}
      \hline
      \multirow{2}{*}{Type}    & \textbf{Models}                & \multicolumn{3}{c|}{\textbf{Llama2-7b}}       & \multicolumn{3}{c|}{\textbf{Phi-2}}           & \multicolumn{3}{c}{\textbf{Qwen2.5-7b}}        \\
      \cline{2-11}
      & \textbf{Selection Budget}      & \textbf{10\%} & \textbf{20\%} & \textbf{30\%} & \textbf{10\%} & \textbf{20\%} & \textbf{30\%} & \textbf{10\%} & \textbf{20\%} & \textbf{30\%}  \\
      \hline
      -                        & Base (Pretrained Model)        & \multicolumn{3}{c|}{19.70}                    & \multicolumn{3}{c|}{34.77}                    & \multicolumn{3}{c}{23.62}                     \\
      \hline
      \multirow{7}{*}{Static}  & Random                         & 32.57         & 33.38         & 33.47         & 33.63         & 23.13         & 20.72         & 38.10         & 37.03         & 36.85          \\
      & Middle Perplexity              & 32.62         & 32.81         & 33.10         & 31.24         & 23.36         & 16.77         & 38.98         & 38.43         & 38.07          \\
      & Facility Location              & 32.97         & 33.02         & 32.85         & 35.57         & 37.85         & 35.16         & \uline{39.41}         & 37.19         & 37.62          \\
      & DivIR                          & 31.83         & 33.36         & 33.34         & 36.91         & 33.93         & 27.43         & 38.47         & 38.47         & \uline{38.75}          \\
      & Long4Align                     & 28.29         & 29.07         & 29.80         & 33.54         & 33.66         & 34.61         & 33.29         & 33.23         & 33.67          \\
      & GradN                          & 30.77         & 33.44         & 33.60         & 32.21         & 29.29         & 22.11         & 38.91         & 38.18         & 37.37          \\
      & Least Confidence               & 31.52         & 31.54         & 32.01         & 34.52         & 30.01         & 34.64         & 34.82         & 34.40         & 34.58          \\
      & TAGCOS                         & 31.33         & \uline{33.78}         & \uline{34.12}         & 33.67         & 36.59         & 34.58         & 38.78         & \uline{38.69}         & 38.50          \\
      \hline
      \multirow{5}{*}{Dynamic} & Fix Interval Random            & 33.06         & 33.18         & 33.05         & 32.72         & 24.46         & 17.95         & 38.09         & 37.13         & 37.30          \\
      & Fix Interval Middle Perplexity & \uline{33.25}         & 32.79         & 33.02         & 31.48         & 34.96         & \uline{37.51}         & 38.94         & 38.05         & 38.30          \\
      & Fix Interval Facility Location & 32.81         & 32.75         & 33.04         & \uline{37.77}         & \uline{38.49}         & 31.62         & 37.22         & 35.99         & 36.03          \\
      & kMQ                            & 33.08         & 32.89         & 33.58         & 32.72         & 30.09         & 33.01         & 39.02         & 37.74         & 37.17          \\
      \cline{2-11}
      & GRACE           &   \textbf{33.58}   &    \textbf{33.97}           &       \textbf{34.14}        &      \textbf{39.61}        &      \textbf{40.00}         &       \textbf{39.80}        &       \textbf{39.71}        &      \textbf{39.15}         &      \textbf{38.96}          \\
      \hline
    \end{tabular}
  }

\end{table*}

\begin{table*}
  \centering
  \caption{ROUGE-L results on DialogSum}
  \label{tab:dialogsum}
  \resizebox{\linewidth}{!}{%
    \begin{tabular}{c|c|c|c|l|c|c|l|c|c|l}
      \hline
      \multirow{2}{*}{Type}    & \textbf{Models}                & \multicolumn{3}{c|}{\textbf{Llama2-7b}}       & \multicolumn{3}{c|}{\textbf{Phi-2}}           & \multicolumn{3}{c}{\textbf{Qwen2.5-7b}}        \\
      \cline{2-11}
      & \textbf{Selection Budget}      & \textbf{10\%} & \textbf{20\%} & \textbf{30\%} & \textbf{10\%} & \textbf{20\%} & \textbf{30\%} & \textbf{10\%} & \textbf{20\%} & \textbf{30\%}  \\
      \hline
      -                        & Base (Pretrained Model)        & \multicolumn{3}{c|}{15.92}                    & \multicolumn{3}{c|}{22.24}                    & \multicolumn{3}{c}{16.97}                     \\
      \hline
      \multirow{7}{*}{Static}  & Random                         & 37.38         & 37.98        & 38.86         & 36.06         & 37.47         & 37.67         & 37.93         & 38.99         & 39.28          \\
      & Middle Perplexity              & 37.12         & 38.35         & 38.55         & 36.35         & 37.23         & 37.68         & 37.63         & 39.07         & 39.19          \\
      & Facility Location              & 37.46         & 38.52         & 39.21         & 36.22         & 36.76         & 38.44         & 38.12         & 39.28         & 38.99          \\
      & DivIR                          & 37.52         & 38.54         & 39.11         & \uline{37.22}         & \uline{38.78}         & \uline{39.13}         & \uline{39.61}         & \uline{40.14}         & \uline{40.05}          \\
      & Long4Align                     & 32.80         & 34.48         & 35.66         & 31.16         & 33.92         & 34.90         & 32.55         & 34.75         & 35.87          \\
      & GradN                          & 36.90         & 37.91         & \uline{39.28}         & 35.02         & 35.80         & 37.10         & 37.38         & 38.45         & 38.66          \\
      & Least Confidence               & 36.80         & 38.02         & 38.19         & 36.41         & 36.99         & 37.72         & 37.83         & 38.64         & 39.12          \\
      & TAGCOS                         & 37.48         & 37.38         & 38.62         & 35.51         & 36.84         & 38.01         & 38.23         & 38.97         & 39.50          \\
      \hline
      \multirow{5}{*}{Dynamic} & Fix Interval Random            & 36.23         & 38.02         & 39.14         & 36.01         & 36.77         & 37.97         & 38.84         & 38.81         & 39.35          \\
      & Fix Interval Middle Perplexity & \uline{37.55}         & \uline{38.48}         & 39.00         & 35.85         & 36.93         & 37.49         & 38.06         & 39.20         & 39.41          \\
      & Fix Interval Facility Location & 35.91         & 37.47         & 37.59         & 35.56         & 36.96         & 38.09         & 37.81         & 38.92         & 39.18          \\
      & kMQ                            & 36.46         & 38.38         & 38.91         & 35.33         & 37.27         & 37.93         & 37.59         & 39.24         & 39.34          \\
      \cline{2-11}
      & GRACE    &  \textbf{37.63}    &      \textbf{38.70}         &        \textbf{39.38}       &         \textbf{37.61}      &      \textbf{38.91}        &     \textbf{39.39}          &       \textbf{39.86}       &      \textbf{40.31}        &        \textbf{40.29}        \\
      \hline
    \end{tabular}
  }
\end{table*}

\subsubsection{Checking Threshold $\delta$}
In Figure~\ref{fig:delta-phi} and~\ref{fig:delta-llama}, we present how GRACE performs under different settings of the checking threshold $\delta$ in Equation~\eqref{eq:check}. With varying $\delta\in[0,0.35]$ at a step of 0.05, we can see that extremely low thresholds ($\delta \approx 0$) result in frequent updates, similar to a strategy that is not adaptive. In contrast, high thresholds (i.e., $\delta \geq 0.3$) suppress updates and reduce the method to a static strategy. We can see that the best performance is observed under moderate thresholds, indicating that infrequent but strategically timed updates are adequate for a dynamic coreset strategy. This confirms that excessive updates may cause instability in the coreset, whereas insufficient updates result in outdated importance scores and degrade performance. Thus, a balanced choice matters.

\subsubsection{Checking interval $t_c$ and Sample Fraction $t_c/\mathcal{T}_e$}
To enable more adaptive coreset updates, we partition the training process into multiple subsets, allowing update checks and reselection to occur more flexibly. Suppose training steps in one epoch $\mathcal{T}_e$, we define the Sample Fraction as the ratio of the checking interval $t_c$ to $\mathcal{T}_e$, computed by $t_c/\mathcal{T}_e$. Here, we vary $t_c/\mathcal{T}_e\in\{0.125,0.25,0.5,1\}$. We then analyze how different $t_c/\mathcal{T}_e$ affects overall performance, as shown in Figure~\ref{fig:tc-phi} and~\ref{fig:tc-llama}. We observe that in-domain performance is relatively stable across different interval settings. However, using a moderate check interval tends to outperform both frequent and coarse settings. In addition, out-of-domain accuracy appears to be more sensitive to the chunk size. These results show that a balanced update interval is preferable for maximizing overall performance.

\subsection{Evaluation on Additional Datasets}\label{subsec:addition-exp}
We further evaluate GRACE on two additional benchmarks to assess its generalization beyond mathematical problems. Specifically, we consider a domain-specific QA dataset, BioInstruct, and a dialog summarization dataset, DialogSum, following the same training protocol as in the main experiments, and evaluate on three budgets.

\subsubsection{Effectiveness Evaluation}
Table~\ref{tab:bioinstruct} and Table~\ref{tab:dialogsum} show the performance comparison on two additional datasets for QA and summarization tasks. GRACE outperforms all the static and dynamic baselines at comparable budgets. On BioInstruct, the behavior across budgets is quite heterogeneous. The performance may drop significantly even if the selection and training budget is increased. This suggests that on this relatively small, domain-specific QA dataset, simply adding more data does not always help if the selection criterion only focuses on very easy or very hard samples. GRACE runs remain aligned with best baselines and often improve upon them, demonstrating the powerful adaptive ability under different conditions. On DialogSum, the overall task appears easier for all methods, and the gap between strong baselines is narrower. In this regime, GRACE yields consistent but moderate improvements. For all three models, GRACE outperforms the best method at all budgets. Overall, these two datasets show that GRACE generalizes beyond mathematical problems to other tasks, such as natural language question answering and summarization.

\section{Conclusion}\label{sec:conclusion}

In this paper, we propose GRACE, a dynamic coreset selection framework for efficient large language model training. Combining the representation diversity and gradient-based importance metrics with a $k$-NN graph-based update mechanism, GRACE adapts coresets to evolving training dynamics while reducing computational costs. Experiments on Llama2-7b, Phi-2 and Qwen2.5-7b models across three benchmarks demonstrate that GRACE outperforms baselines in both efficiency and performance, making it a scalable and effective solution for resource-constrained LLM training.

\balance
\bibliographystyle{ACM-Reference-Format}
\bibliography{grace}


\begin{thebibliography}{85}

%%% the Web address.

\ifx \showCODEN    \undefined \def \showCODEN     #1{\unskip}     \fi
\ifx \showISBNx    \undefined \def \showISBNx     #1{\unskip}     \fi
\ifx \showISBNxiii \undefined \def \showISBNxiii  #1{\unskip}     \fi
\ifx \showISSN     \undefined \def \showISSN      #1{\unskip}     \fi
\ifx \showLCCN     \undefined \def \showLCCN      #1{\unskip}     \fi
\ifx \shownote     \undefined \def \shownote      #1{#1}          \fi
\ifx \showarticletitle \undefined \def \showarticletitle #1{#1}   \fi
\ifx \showURL      \undefined \def \showURL       {\relax}        \fi
% invisible to TeX
\providecommand\bibfield[2]{#2}
\providecommand\bibinfo[2]{#2}
\providecommand\natexlab[1]{#1}
\providecommand\showeprint[2][]{arXiv:#2}

\bibitem[Acharya et~al\mbox{.}(2024)]%
        {acharyaBalancingFeatureSimilarity2024}
\bibfield{author}{\bibinfo{person}{Abhinab Acharya}, \bibinfo{person}{Dayou Yu}, \bibinfo{person}{Qi Yu}, {and} \bibinfo{person}{Xumin Liu}.} \bibinfo{year}{2024}\natexlab{}.
\newblock \showarticletitle{Balancing {{Feature Similarity}} and {{Label Variability}} for {{Optimal Size-Aware One-shot Subset Selection}}}. In \bibinfo{booktitle}{\emph{Forty-First {{International Conference}} on {{Machine Learning}}}}.
\newblock


\bibitem[AI(2023)]%
        {touvron2023llama2openfoundation}
\bibfield{author}{\bibinfo{person}{Meta AI}.} \bibinfo{year}{2023}\natexlab{}.
\newblock \bibinfo{title}{Llama 2: Open Foundation and Fine-Tuned Chat Models}.
\newblock
\showeprint[arxiv]{2307.09288}~[cs.CL]
\urldef\tempurl%
\url{https://arxiv.org/abs/2307.09288}
\showURL{%
\tempurl}


\bibitem[AI(2024)]%
        {dubey2024llama3herdmodels}
\bibfield{author}{\bibinfo{person}{Meta AI}.} \bibinfo{year}{2024}\natexlab{}.
\newblock \bibinfo{title}{The Llama 3 Herd of Models}.
\newblock
\showeprint[arxiv]{2407.21783}~[cs.AI]
\urldef\tempurl%
\url{https://arxiv.org/abs/2407.21783}
\showURL{%
\tempurl}


\bibitem[Albalak et~al\mbox{.}(2024)]%
        {albalak2024a}
\bibfield{author}{\bibinfo{person}{Alon Albalak}, \bibinfo{person}{Yanai Elazar}, \bibinfo{person}{Sang~Michael Xie}, \bibinfo{person}{Shayne Longpre}, \bibinfo{person}{Nathan Lambert}, \bibinfo{person}{Xinyi Wang}, \bibinfo{person}{Niklas Muennighoff}, \bibinfo{person}{Bairu Hou}, \bibinfo{person}{Liangming Pan}, \bibinfo{person}{Haewon Jeong}, \bibinfo{person}{Colin Raffel}, \bibinfo{person}{Shiyu Chang}, \bibinfo{person}{Tatsunori Hashimoto}, {and} \bibinfo{person}{William~Yang Wang}.} \bibinfo{year}{2024}\natexlab{}.
\newblock \showarticletitle{A Survey on Data Selection for Language Models}.
\newblock \bibinfo{journal}{\emph{Transactions on Machine Learning Research}} (\bibinfo{year}{2024}).
\newblock
\showISSN{2835-8856}
\urldef\tempurl%
\url{https://openreview.net/forum?id=XfHWcNTSHp}
\showURL{%
\tempurl}
\newblock
\shownote{Survey Certification}.


\bibitem[Ankner et~al\mbox{.}(2024)]%
        {anknerPerplexedPerplexityPerplexityBased2024}
\bibfield{author}{\bibinfo{person}{Zachary Ankner}, \bibinfo{person}{Cody Blakeney}, \bibinfo{person}{Kartik Sreenivasan}, \bibinfo{person}{Max Marion}, \bibinfo{person}{Matthew~L. Leavitt}, {and} \bibinfo{person}{Mansheej Paul}.} \bibinfo{year}{2024}\natexlab{}.
\newblock \bibinfo{title}{Perplexed by {{Perplexity}}: {{Perplexity-Based Data Pruning With Small Reference Models}}}.
\newblock
\href{https://doi.org/10.48550/ARXIV.2405.20541}{doi:\nolinkurl{10.48550/ARXIV.2405.20541}}


\bibitem[Arora et~al\mbox{.}(2023)]%
        {10.14778/3626292.3626294}
\bibfield{author}{\bibinfo{person}{Simran Arora}, \bibinfo{person}{Brandon Yang}, \bibinfo{person}{Sabri Eyuboglu}, \bibinfo{person}{Avanika Narayan}, \bibinfo{person}{Andrew Hojel}, \bibinfo{person}{Immanuel Trummer}, {and} \bibinfo{person}{Christopher R\'{e}}.} \bibinfo{year}{2023}\natexlab{}.
\newblock \showarticletitle{Language Models Enable Simple Systems for Generating Structured Views of Heterogeneous Data Lakes}.
\newblock \bibinfo{journal}{\emph{Proc. VLDB Endow.}} \bibinfo{volume}{17}, \bibinfo{number}{2} (\bibinfo{date}{Oct.} \bibinfo{year}{2023}), \bibinfo{pages}{92–105}.
\newblock
\showISSN{2150-8097}
\href{https://doi.org/10.14778/3626292.3626294}{doi:\nolinkurl{10.14778/3626292.3626294}}


\bibitem[Bai et~al\mbox{.}(2024)]%
        {baiMultiAgentCollaborativeData2024}
\bibfield{author}{\bibinfo{person}{Tianyi Bai}, \bibinfo{person}{Ling Yang}, \bibinfo{person}{Zhen~Hao Wong}, \bibinfo{person}{Jiahui Peng}, \bibinfo{person}{Xinlin Zhuang}, \bibinfo{person}{Chi Zhang}, \bibinfo{person}{Lijun Wu}, \bibinfo{person}{Jiantao Qiu}, \bibinfo{person}{Wentao Zhang}, \bibinfo{person}{Binhang Yuan}, {and} \bibinfo{person}{Conghui He}.} \bibinfo{year}{2024}\natexlab{}.
\newblock \bibinfo{title}{Multi-{{Agent Collaborative Data Selection}} for {{Efficient LLM Pretraining}}}.
\newblock
\showeprint[arxiv]{2410.08102}~[cs]
\href{https://doi.org/10.48550/arXiv.2410.08102}{doi:\nolinkurl{10.48550/arXiv.2410.08102}}


\bibitem[Bhatt et~al\mbox{.}(2024)]%
        {bhattExperimentalDesignFramework2024}
\bibfield{author}{\bibinfo{person}{Gantavya Bhatt}, \bibinfo{person}{Yifang Chen}, \bibinfo{person}{Arnav~M. Das}, \bibinfo{person}{Jifan Zhang}, \bibinfo{person}{Sang~T. Truong}, \bibinfo{person}{Stephen Mussmann}, \bibinfo{person}{Yinglun Zhu}, \bibinfo{person}{Jeffrey Bilmes}, \bibinfo{person}{Simon~S. Du}, \bibinfo{person}{Kevin Jamieson}, \bibinfo{person}{Jordan~T. Ash}, {and} \bibinfo{person}{Robert~D. Nowak}.} \bibinfo{year}{2024}\natexlab{}.
\newblock \bibinfo{title}{An {{Experimental Design Framework}} for {{Label-Efficient Supervised Finetuning}} of {{Large Language Models}}}.
\newblock
\showeprint[arxiv]{2401.06692}~[cs]
\href{https://doi.org/10.48550/arXiv.2401.06692}{doi:\nolinkurl{10.48550/arXiv.2401.06692}}


\bibitem[Bottou(2012)]%
        {bottou2012stochastic}
\bibfield{author}{\bibinfo{person}{L{\'e}on Bottou}.} \bibinfo{year}{2012}\natexlab{}.
\newblock \showarticletitle{Stochastic gradient descent tricks}.
\newblock In \bibinfo{booktitle}{\emph{Neural networks: tricks of the trade: second edition}}. \bibinfo{publisher}{Springer}, \bibinfo{pages}{421--436}.
\newblock


\bibitem[Castin et~al\mbox{.}(2024)]%
        {castin2024how}
\bibfield{author}{\bibinfo{person}{Val{\'e}rie Castin}, \bibinfo{person}{Pierre Ablin}, {and} \bibinfo{person}{Gabriel Peyr{\'e}}.} \bibinfo{year}{2024}\natexlab{}.
\newblock \showarticletitle{How Smooth Is Attention?}. In \bibinfo{booktitle}{\emph{Forty-first International Conference on Machine Learning}}.
\newblock
\urldef\tempurl%
\url{https://openreview.net/forum?id=aP0H8A1ywk}
\showURL{%
\tempurl}


\bibitem[Chai et~al\mbox{.}(2023a)]%
        {chaiGoodCoreDataeffectiveDataefficient2023}
\bibfield{author}{\bibinfo{person}{Chengliang Chai}, \bibinfo{person}{Jiabin Liu}, \bibinfo{person}{Nan Tang}, \bibinfo{person}{Ju Fan}, \bibinfo{person}{Dongjing Miao}, \bibinfo{person}{Jiayi Wang}, \bibinfo{person}{Yuyu Luo}, {and} \bibinfo{person}{Guoliang Li}.} \bibinfo{year}{2023}\natexlab{a}.
\newblock \showarticletitle{{{GoodCore}}: {{Data-effective}} and {{Data-efficient Machine Learning}} through {{Coreset Selection}} over {{Incomplete Data}}}.
\newblock \bibinfo{journal}{\emph{Proc. ACM Manag. Data}} \bibinfo{volume}{1}, \bibinfo{number}{2} (\bibinfo{date}{June} \bibinfo{year}{2023}), \bibinfo{pages}{157:1--157:27}.
\newblock
\href{https://doi.org/10.1145/3589302}{doi:\nolinkurl{10.1145/3589302}}


\bibitem[Chai et~al\mbox{.}(2023b)]%
        {chaiEfficientCoresetSelection2023}
\bibfield{author}{\bibinfo{person}{Chengliang Chai}, \bibinfo{person}{Jiayi Wang}, \bibinfo{person}{Nan Tang}, \bibinfo{person}{Ye Yuan}, \bibinfo{person}{Jiabin Liu}, \bibinfo{person}{Yuhao Deng}, {and} \bibinfo{person}{Guoren Wang}.} \bibinfo{year}{2023}\natexlab{b}.
\newblock \showarticletitle{Efficient {{Coreset Selection}} with {{Cluster-based Methods}}}. In \bibinfo{booktitle}{\emph{Proceedings of the 29th {{ACM SIGKDD Conference}} on {{Knowledge Discovery}} and {{Data Mining}}}} \emph{(\bibinfo{series}{{{KDD}} '23})}. \bibinfo{publisher}{Association for Computing Machinery}, \bibinfo{address}{New York, NY, USA}, \bibinfo{pages}{167--178}.
\newblock
\showISBNx{9798400701030}
\href{https://doi.org/10.1145/3580305.3599326}{doi:\nolinkurl{10.1145/3580305.3599326}}


\bibitem[Chen et~al\mbox{.}(2023)]%
        {chenMaybeOnly052023}
\bibfield{author}{\bibinfo{person}{Hao Chen}, \bibinfo{person}{Yiming Zhang}, \bibinfo{person}{Qi Zhang}, \bibinfo{person}{Hantao Yang}, \bibinfo{person}{Xiaomeng Hu}, \bibinfo{person}{Xuetao Ma}, \bibinfo{person}{Yifan Yanggong}, {and} \bibinfo{person}{Junbo Zhao}.} \bibinfo{year}{2023}\natexlab{}.
\newblock \bibinfo{title}{Maybe {{Only}} 0.5\% {{Data}} Is {{Needed}}: {{A Preliminary Exploration}} of {{Low Training Data Instruction Tuning}}}.
\newblock
\showeprint[arxiv]{2305.09246}~[cs]
\href{https://doi.org/10.48550/arXiv.2305.09246}{doi:\nolinkurl{10.48550/arXiv.2305.09246}}


\bibitem[Chen et~al\mbox{.}(2021b)]%
        {chen2021evaluating}
\bibfield{author}{\bibinfo{person}{Mark Chen}, \bibinfo{person}{Jerry Tworek}, \bibinfo{person}{Heewoo Jun}, \bibinfo{person}{Qiming Yuan}, \bibinfo{person}{Henrique Ponde De~Oliveira Pinto}, \bibinfo{person}{Jared Kaplan}, \bibinfo{person}{Harri Edwards}, \bibinfo{person}{Yuri Burda}, \bibinfo{person}{Nicholas Joseph}, \bibinfo{person}{Greg Brockman}, {et~al\mbox{.}}} \bibinfo{year}{2021}\natexlab{b}.
\newblock \showarticletitle{Evaluating large language models trained on code}.
\newblock \bibinfo{journal}{\emph{arXiv preprint arXiv:2107.03374}} (\bibinfo{year}{2021}).
\newblock


\bibitem[Chen et~al\mbox{.}(2021a)]%
        {chen-etal-2021-dialogsum}
\bibfield{author}{\bibinfo{person}{Yulong Chen}, \bibinfo{person}{Yang Liu}, \bibinfo{person}{Liang Chen}, {and} \bibinfo{person}{Yue Zhang}.} \bibinfo{year}{2021}\natexlab{a}.
\newblock \showarticletitle{{D}ialog{S}um: {A} Real-Life Scenario Dialogue Summarization Dataset}. In \bibinfo{booktitle}{\emph{Findings of the Association for Computational Linguistics: ACL-IJCNLP 2021}}, \bibfield{editor}{\bibinfo{person}{Chengqing Zong}, \bibinfo{person}{Fei Xia}, \bibinfo{person}{Wenjie Li}, {and} \bibinfo{person}{Roberto Navigli}} (Eds.). \bibinfo{publisher}{Association for Computational Linguistics}, \bibinfo{address}{Online}, \bibinfo{pages}{5062--5074}.
\newblock
\href{https://doi.org/10.18653/v1/2021.findings-acl.449}{doi:\nolinkurl{10.18653/v1/2021.findings-acl.449}}


\bibitem[Cho et~al\mbox{.}(2025)]%
        {choLightweightDatasetPruning2025}
\bibfield{author}{\bibinfo{person}{Yeseul Cho}, \bibinfo{person}{Baekrok Shin}, \bibinfo{person}{Changmin Kang}, {and} \bibinfo{person}{Chulhee Yun}.} \bibinfo{year}{2025}\natexlab{}.
\newblock \bibinfo{title}{Lightweight {{Dataset Pruning}} without {{Full Training}} via {{Example Difficulty}} and {{Prediction Uncertainty}}}.
\newblock
\showeprint[arxiv]{2502.06905}~[cs]
\href{https://doi.org/10.48550/arXiv.2502.06905}{doi:\nolinkurl{10.48550/arXiv.2502.06905}}


\bibitem[Danilova et~al\mbox{.}(2020)]%
        {danilova2020recent}
\bibfield{author}{\bibinfo{person}{Marina Danilova}, \bibinfo{person}{Pavel Dvurechensky}, \bibinfo{person}{Alexander Gasnikov}, \bibinfo{person}{Eduard Gorbunov}, \bibinfo{person}{Sergey Guminov}, \bibinfo{person}{Dmitry Kamzolov}, {and} \bibinfo{person}{Innokentiy Shibaev}.} \bibinfo{year}{2020}\natexlab{}.
\newblock \bibinfo{title}{Recent Theoretical Advances in Non-Convex Optimization}.
\newblock
\showeprint[arxiv]{2012.06188}~[math.OC]


\bibitem[DeepSeek-AI et~al\mbox{.}(2024)]%
        {deepseekai2024deepseekv3}
\bibfield{author}{\bibinfo{person}{DeepSeek-AI}, \bibinfo{person}{Aixin Liu}, \bibinfo{person}{Bei Feng}, {et~al\mbox{.}}} \bibinfo{year}{2024}\natexlab{}.
\newblock \bibinfo{title}{DeepSeek-V3 Technical Report}.
\newblock
\showeprint[arxiv]{2412.19437}~[cs.CL]


\bibitem[Demidovskij et~al\mbox{.}(2023)]%
        {demidovskijDARELDataReduction2023}
\bibfield{author}{\bibinfo{person}{Alexander~Vladimirovich Demidovskij}, \bibinfo{person}{Aleksei Trutnev}, \bibinfo{person}{Artem Tugarev}, \bibinfo{person}{Igor Salnikov}, {and} \bibinfo{person}{Stanislav Pavlov}.} \bibinfo{year}{2023}\natexlab{}.
\newblock \showarticletitle{{{DAREL}}: {{Data Reduction}} with {{Losses}} for {{Training Acceleration}} of {{Real}} and {{Hypercomplex Neural Networks}}}. In \bibinfo{booktitle}{\emph{Workshop on {{Advancing Neural Network Training}}: {{Computational Efficiency}}, {{Scalability}}, and {{Resource Optimization}} ({{WANT}}@{{NeurIPS}} 2023)}}.
\newblock


\bibitem[Deng et~al\mbox{.}(2024)]%
        {dengInfluentialLanguageData2024}
\bibfield{author}{\bibinfo{person}{Zhiwei Deng}, \bibinfo{person}{Tao Li}, {and} \bibinfo{person}{Yang Li}.} \bibinfo{year}{2024}\natexlab{}.
\newblock \bibinfo{title}{Influential {{Language Data Selection}} via {{Gradient Trajectory Pursuit}}}.
\newblock
\showeprint[arxiv]{2410.16710}
\href{https://doi.org/10.48550/arXiv.2410.16710}{doi:\nolinkurl{10.48550/arXiv.2410.16710}}


\bibitem[Dettmers et~al\mbox{.}(2023)]%
        {dettmers2023qlora}
\bibfield{author}{\bibinfo{person}{Tim Dettmers}, \bibinfo{person}{Artidoro Pagnoni}, \bibinfo{person}{Ari Holtzman}, {and} \bibinfo{person}{Luke Zettlemoyer}.} \bibinfo{year}{2023}\natexlab{}.
\newblock \showarticletitle{Qlora: Efficient finetuning of quantized llms}.
\newblock \bibinfo{journal}{\emph{Advances in neural information processing systems}}  \bibinfo{volume}{36} (\bibinfo{year}{2023}), \bibinfo{pages}{10088--10115}.
\newblock


\bibitem[Fan et~al\mbox{.}(2024)]%
        {10.14778/3681954.3681960}
\bibfield{author}{\bibinfo{person}{Ju Fan}, \bibinfo{person}{Zihui Gu}, \bibinfo{person}{Songyue Zhang}, \bibinfo{person}{Yuxin Zhang}, \bibinfo{person}{Zui Chen}, \bibinfo{person}{Lei Cao}, \bibinfo{person}{Guoliang Li}, \bibinfo{person}{Samuel Madden}, \bibinfo{person}{Xiaoyong Du}, {and} \bibinfo{person}{Nan Tang}.} \bibinfo{year}{2024}\natexlab{}.
\newblock \showarticletitle{Combining Small Language Models and Large Language Models for Zero-Shot NL2SQL}.
\newblock \bibinfo{journal}{\emph{Proc. VLDB Endow.}} \bibinfo{volume}{17}, \bibinfo{number}{11} (\bibinfo{date}{July} \bibinfo{year}{2024}), \bibinfo{pages}{2750–2763}.
\newblock
\showISSN{2150-8097}
\href{https://doi.org/10.14778/3681954.3681960}{doi:\nolinkurl{10.14778/3681954.3681960}}


\bibitem[Feldman(2020)]%
        {feldman2020introduction}
\bibfield{author}{\bibinfo{person}{Dan Feldman}.} \bibinfo{year}{2020}\natexlab{}.
\newblock \bibinfo{title}{Introduction to Core-sets: an Updated Survey}.
\newblock
\showeprint[arxiv]{2011.09384}~[cs.LG]


\bibitem[Giannakouris and Trummer(2025)]%
        {10.1145/3709652}
\bibfield{author}{\bibinfo{person}{Victor Giannakouris} {and} \bibinfo{person}{Immanuel Trummer}.} \bibinfo{year}{2025}\natexlab{}.
\newblock \showarticletitle{$\lambda$-Tune: Harnessing Large Language Models for Automated Database System Tuning}.
\newblock \bibinfo{journal}{\emph{Proc. ACM Manag. Data}} \bibinfo{volume}{3}, \bibinfo{number}{1}, Article \bibinfo{articleno}{2} (\bibinfo{date}{Feb.} \bibinfo{year}{2025}), \bibinfo{numpages}{26}~pages.
\newblock
\href{https://doi.org/10.1145/3709652}{doi:\nolinkurl{10.1145/3709652}}


\bibitem[Hadar et~al\mbox{.}(2024)]%
        {10.14778/3712221.3712249}
\bibfield{author}{\bibinfo{person}{Aviv Hadar}, \bibinfo{person}{Tova Milo}, {and} \bibinfo{person}{Kathy Razmadze}.} \bibinfo{year}{2024}\natexlab{}.
\newblock \showarticletitle{Datamap-Driven Tabular Coreset Selection for Classifier Training}.
\newblock \bibinfo{journal}{\emph{Proc. VLDB Endow.}} \bibinfo{volume}{18}, \bibinfo{number}{3} (\bibinfo{date}{Nov.} \bibinfo{year}{2024}), \bibinfo{pages}{876–888}.
\newblock
\showISSN{2150-8097}
\href{https://doi.org/10.14778/3712221.3712249}{doi:\nolinkurl{10.14778/3712221.3712249}}


\bibitem[He et~al\mbox{.}(2022)]%
        {he2022towards}
\bibfield{author}{\bibinfo{person}{Junxian He}, \bibinfo{person}{Chunting Zhou}, \bibinfo{person}{Xuezhe Ma}, \bibinfo{person}{Taylor Berg-Kirkpatrick}, {and} \bibinfo{person}{Graham Neubig}.} \bibinfo{year}{2022}\natexlab{}.
\newblock \showarticletitle{Towards a Unified View of Parameter-Efficient Transfer Learning}. In \bibinfo{booktitle}{\emph{International Conference on Learning Representations}}.
\newblock
\urldef\tempurl%
\url{https://openreview.net/forum?id=0RDcd5Axok}
\showURL{%
\tempurl}


\bibitem[He et~al\mbox{.}(2024)]%
        {heSHEDShapleyBasedAutomated2024}
\bibfield{author}{\bibinfo{person}{Yexiao He}, \bibinfo{person}{Ziyao Wang}, \bibinfo{person}{Zheyu Shen}, \bibinfo{person}{Guoheng Sun}, \bibinfo{person}{Yucong Dai}, \bibinfo{person}{Yongkai Wu}, \bibinfo{person}{Hongyi Wang}, {and} \bibinfo{person}{Ang Li}.} \bibinfo{year}{2024}\natexlab{}.
\newblock \bibinfo{title}{{{SHED}}: {{Shapley-Based Automated Dataset Refinement}} for {{Instruction Fine-Tuning}}}.
\newblock
\showeprint[arxiv]{2405.00705}~[cs]
\href{https://doi.org/10.48550/arXiv.2405.00705}{doi:\nolinkurl{10.48550/arXiv.2405.00705}}


\bibitem[Hochbaum(1997)]%
        {hochbaum1997approximating}
\bibfield{author}{\bibinfo{person}{Dorit~S Hochbaum}.} \bibinfo{year}{1997}\natexlab{}.
\newblock \showarticletitle{Approximating covering and packing problems: set cover, vertex cover, independent set, and related problems}.
\newblock \bibinfo{journal}{\emph{Approximation algorithms for NP-hard problems}} (\bibinfo{year}{1997}), \bibinfo{pages}{94--143}.
\newblock


\bibitem[Hoffmann et~al\mbox{.}(2022)]%
        {hoffmann2022training}
\bibfield{author}{\bibinfo{person}{Jordan Hoffmann}, \bibinfo{person}{Sebastian Borgeaud}, \bibinfo{person}{Arthur Mensch}, \bibinfo{person}{Elena Buchatskaya}, \bibinfo{person}{Trevor Cai}, \bibinfo{person}{Eliza Rutherford}, \bibinfo{person}{Diego de Las~Casas}, \bibinfo{person}{Lisa~Anne Hendricks}, \bibinfo{person}{Johannes Welbl}, \bibinfo{person}{Aidan Clark}, {et~al\mbox{.}}} \bibinfo{year}{2022}\natexlab{}.
\newblock \showarticletitle{Training compute-optimal large language models}. In \bibinfo{booktitle}{\emph{Proceedings of the 36th International Conference on Neural Information Processing Systems}}. \bibinfo{pages}{30016--30030}.
\newblock


\bibitem[Houlsby et~al\mbox{.}(2019)]%
        {houlsby2019parameter}
\bibfield{author}{\bibinfo{person}{Neil Houlsby}, \bibinfo{person}{Andrei Giurgiu}, \bibinfo{person}{Stanislaw Jastrzebski}, \bibinfo{person}{Bruna Morrone}, \bibinfo{person}{Quentin De~Laroussilhe}, \bibinfo{person}{Andrea Gesmundo}, \bibinfo{person}{Mona Attariyan}, {and} \bibinfo{person}{Sylvain Gelly}.} \bibinfo{year}{2019}\natexlab{}.
\newblock \showarticletitle{Parameter-efficient transfer learning for NLP}. In \bibinfo{booktitle}{\emph{International conference on machine learning}}. PMLR, \bibinfo{pages}{2790--2799}.
\newblock


\bibitem[Hu et~al\mbox{.}(2022)]%
        {hu2022lora}
\bibfield{author}{\bibinfo{person}{Edward~J Hu}, \bibinfo{person}{Yelong Shen}, \bibinfo{person}{Phillip Wallis}, \bibinfo{person}{Zeyuan Allen-Zhu}, \bibinfo{person}{Yuanzhi Li}, \bibinfo{person}{Shean Wang}, \bibinfo{person}{Lu Wang}, \bibinfo{person}{Weizhu Chen}, {et~al\mbox{.}}} \bibinfo{year}{2022}\natexlab{}.
\newblock \showarticletitle{Lora: Low-rank adaptation of large language models.}
\newblock \bibinfo{journal}{\emph{ICLR}} \bibinfo{volume}{1}, \bibinfo{number}{2} (\bibinfo{year}{2022}), \bibinfo{pages}{3}.
\newblock


\bibitem[Indyk and Motwani(1998)]%
        {10.1145/276698.276876}
\bibfield{author}{\bibinfo{person}{Piotr Indyk} {and} \bibinfo{person}{Rajeev Motwani}.} \bibinfo{year}{1998}\natexlab{}.
\newblock \showarticletitle{Approximate nearest neighbors: towards removing the curse of dimensionality}. In \bibinfo{booktitle}{\emph{Proceedings of the Thirtieth Annual ACM Symposium on Theory of Computing}} (Dallas, Texas, USA) \emph{(\bibinfo{series}{STOC '98})}. \bibinfo{publisher}{Association for Computing Machinery}, \bibinfo{address}{New York, NY, USA}, \bibinfo{pages}{604–613}.
\newblock
\showISBNx{0897919629}
\href{https://doi.org/10.1145/276698.276876}{doi:\nolinkurl{10.1145/276698.276876}}


\bibitem[Isik et~al\mbox{.}(2024)]%
        {isikScalingLawsDownstream2024}
\bibfield{author}{\bibinfo{person}{Berivan Isik}, \bibinfo{person}{Natalia Ponomareva}, \bibinfo{person}{Hussein Hazimeh}, \bibinfo{person}{Dimitris Paparas}, \bibinfo{person}{Sergei Vassilvitskii}, {and} \bibinfo{person}{Sanmi Koyejo}.} \bibinfo{year}{2024}\natexlab{}.
\newblock \bibinfo{title}{Scaling {{Laws}} for {{Downstream Task Performance}} of {{Large Language Models}}}.
\newblock
\href{https://doi.org/10.48550/ARXIV.2402.04177}{doi:\nolinkurl{10.48550/ARXIV.2402.04177}}


\bibitem[Joaquin et~al\mbox{.}(2024)]%
        {joaquinIn2CoreLeveragingInfluence2024}
\bibfield{author}{\bibinfo{person}{Ayrton~San Joaquin}, \bibinfo{person}{Bin Wang}, \bibinfo{person}{Zhengyuan Liu}, \bibinfo{person}{Nicholas Asher}, \bibinfo{person}{Brian Lim}, \bibinfo{person}{Philippe Muller}, {and} \bibinfo{person}{Nancy Chen}.} \bibinfo{year}{2024}\natexlab{}.
\newblock \bibinfo{title}{{{In2Core}}: {{Leveraging Influence Functions}} for {{Coreset Selection}} in {{Instruction Finetuning}} of {{Large Language Models}}}.
\newblock
\showeprint[arxiv]{2408.03560}~[cs, stat]
\href{https://doi.org/10.48550/arXiv.2408.03560}{doi:\nolinkurl{10.48550/arXiv.2408.03560}}


\bibitem[Johnson et~al\mbox{.}(2019)]%
        {johnson2019billion}
\bibfield{author}{\bibinfo{person}{Jeff Johnson}, \bibinfo{person}{Matthijs Douze}, {and} \bibinfo{person}{Herv{\'e} J{\'e}gou}.} \bibinfo{year}{2019}\natexlab{}.
\newblock \showarticletitle{Billion-scale similarity search with {GPUs}}.
\newblock \bibinfo{journal}{\emph{IEEE Transactions on Big Data}} \bibinfo{volume}{7}, \bibinfo{number}{3} (\bibinfo{year}{2019}), \bibinfo{pages}{535--547}.
\newblock


\bibitem[Kayali et~al\mbox{.}(2024)]%
        {10.14778/3659437.3659461}
\bibfield{author}{\bibinfo{person}{Moe Kayali}, \bibinfo{person}{Anton Lykov}, \bibinfo{person}{Ilias Fountalis}, \bibinfo{person}{Nikolaos Vasiloglou}, \bibinfo{person}{Dan Olteanu}, {and} \bibinfo{person}{Dan Suciu}.} \bibinfo{year}{2024}\natexlab{}.
\newblock \showarticletitle{Chorus: Foundation Models for Unified Data Discovery and Exploration}.
\newblock \bibinfo{journal}{\emph{Proc. VLDB Endow.}} \bibinfo{volume}{17}, \bibinfo{number}{8} (\bibinfo{date}{April} \bibinfo{year}{2024}), \bibinfo{pages}{2104–2114}.
\newblock
\showISSN{2150-8097}
\href{https://doi.org/10.14778/3659437.3659461}{doi:\nolinkurl{10.14778/3659437.3659461}}


\bibitem[Killamsetty et~al\mbox{.}(2021)]%
        {killamsettyGRADMATCHGradientMatching2021}
\bibfield{author}{\bibinfo{person}{Krishnateja Killamsetty}, \bibinfo{person}{Durga S}, \bibinfo{person}{Ganesh Ramakrishnan}, \bibinfo{person}{Abir De}, {and} \bibinfo{person}{Rishabh Iyer}.} \bibinfo{year}{2021}\natexlab{}.
\newblock \showarticletitle{{{GRAD-MATCH}}: {{Gradient Matching}} Based {{Data Subset Selection}} for {{Efficient Deep Model Training}}}. In \bibinfo{booktitle}{\emph{Proceedings of the 38th {{International Conference}} on {{Machine Learning}}}}. \bibinfo{publisher}{PMLR}, \bibinfo{pages}{5464--5474}.
\newblock
\showISSN{2640-3498}


\bibitem[Lazebnik et~al\mbox{.}(2022)]%
        {10.14778/3574245.3574261}
\bibfield{author}{\bibinfo{person}{Teddy Lazebnik}, \bibinfo{person}{Amit Somech}, {and} \bibinfo{person}{Abraham~Itzhak Weinberg}.} \bibinfo{year}{2022}\natexlab{}.
\newblock \showarticletitle{SubStrat: A Subset-Based Optimization Strategy for Faster AutoML}.
\newblock \bibinfo{journal}{\emph{Proc. VLDB Endow.}} \bibinfo{volume}{16}, \bibinfo{number}{4} (\bibinfo{date}{Dec.} \bibinfo{year}{2022}), \bibinfo{pages}{772–780}.
\newblock
\showISSN{2150-8097}
\href{https://doi.org/10.14778/3574245.3574261}{doi:\nolinkurl{10.14778/3574245.3574261}}


\bibitem[Li et~al\mbox{.}(2024d)]%
        {10.14778/3681954.3682003}
\bibfield{author}{\bibinfo{person}{Boyan Li}, \bibinfo{person}{Yuyu Luo}, \bibinfo{person}{Chengliang Chai}, \bibinfo{person}{Guoliang Li}, {and} \bibinfo{person}{Nan Tang}.} \bibinfo{year}{2024}\natexlab{d}.
\newblock \showarticletitle{The Dawn of Natural Language to SQL: Are We Fully Ready?}
\newblock \bibinfo{journal}{\emph{Proc. VLDB Endow.}} \bibinfo{volume}{17}, \bibinfo{number}{11} (\bibinfo{date}{July} \bibinfo{year}{2024}), \bibinfo{pages}{3318–3331}.
\newblock
\showISSN{2150-8097}
\href{https://doi.org/10.14778/3681954.3682003}{doi:\nolinkurl{10.14778/3681954.3682003}}


\bibitem[Li et~al\mbox{.}(2024a)]%
        {li20242}
\bibfield{author}{\bibinfo{person}{Haoyang Li}, \bibinfo{person}{Shimin Di}, \bibinfo{person}{Lei Chen}, {and} \bibinfo{person}{Xiaofang Zhou}.} \bibinfo{year}{2024}\natexlab{a}.
\newblock \showarticletitle{E2GCL: Efficient and Expressive Contrastive Learning on Graph Neural Networks}. In \bibinfo{booktitle}{\emph{2024 IEEE 40th International Conference on Data Engineering (ICDE)}}. IEEE, \bibinfo{pages}{859--873}.
\newblock


\bibitem[Li et~al\mbox{.}(2024b)]%
        {li2024fight}
\bibfield{author}{\bibinfo{person}{Haoyang Li}, \bibinfo{person}{Shimin Di}, \bibinfo{person}{Calvin Hong~Yi Li}, \bibinfo{person}{Lei Chen}, {and} \bibinfo{person}{Xiaofang Zhou}.} \bibinfo{year}{2024}\natexlab{b}.
\newblock \showarticletitle{Fight Fire with Fire: Towards Robust Graph Neural Networks on Dynamic Graphs via Actively Defense}.
\newblock \bibinfo{journal}{\emph{Proceedings of the VLDB Endowment}} \bibinfo{volume}{17}, \bibinfo{number}{8} (\bibinfo{year}{2024}), \bibinfo{pages}{2050--2063}.
\newblock


\bibitem[Li et~al\mbox{.}(2024c)]%
        {li2024survey}
\bibfield{author}{\bibinfo{person}{Haoyang Li}, \bibinfo{person}{Yiming Li}, \bibinfo{person}{Anxin Tian}, \bibinfo{person}{Tianhao Tang}, \bibinfo{person}{Zhanchao Xu}, \bibinfo{person}{Xuejia Chen}, \bibinfo{person}{Nicole Hu}, \bibinfo{person}{Wei Dong}, \bibinfo{person}{Qing Li}, {and} \bibinfo{person}{Lei Chen}.} \bibinfo{year}{2024}\natexlab{c}.
\newblock \showarticletitle{A survey on large language model acceleration based on kv cache management}.
\newblock \bibinfo{journal}{\emph{arXiv preprint arXiv:2412.19442}} (\bibinfo{year}{2024}).
\newblock


\bibitem[Li et~al\mbox{.}(2024f)]%
        {liQuantityQualityBoosting2024}
\bibfield{author}{\bibinfo{person}{Ming Li}, \bibinfo{person}{Yong Zhang}, \bibinfo{person}{Zhitao Li}, \bibinfo{person}{Jiuhai Chen}, \bibinfo{person}{Lichang Chen}, \bibinfo{person}{Ning Cheng}, \bibinfo{person}{Jianzong Wang}, \bibinfo{person}{Tianyi Zhou}, {and} \bibinfo{person}{Jing Xiao}.} \bibinfo{year}{2024}\natexlab{f}.
\newblock \showarticletitle{From {{Quantity}} to {{Quality}}: {{Boosting LLM Performance}} with {{Self-Guided Data Selection}} for {{Instruction Tuning}}}. In \bibinfo{booktitle}{\emph{Proceedings of the 2024 {{Conference}} of the {{North American Chapter}} of the {{Association}} for {{Computational Linguistics}}: {{Human Language Technologies}} ({{Volume}} 1: {{Long Papers}}), {{NAACL}} 2024, {{Mexico City}}, {{Mexico}}, {{June}} 16-21, 2024}}, \bibfield{editor}{\bibinfo{person}{Kevin Duh}, \bibinfo{person}{Helena {G{\'o}mez-Adorno}}, {and} \bibinfo{person}{Steven Bethard}} (Eds.). \bibinfo{publisher}{Association for Computational Linguistics}, \bibinfo{pages}{7602--7635}.
\newblock
\showeprint[arxiv]{2308.12032}~[cs]
\href{https://doi.org/10.18653/V1/2024.NAACL-LONG.421}{doi:\nolinkurl{10.18653/V1/2024.NAACL-LONG.421}}


\bibitem[Li and Liang(2021)]%
        {li2021prefix}
\bibfield{author}{\bibinfo{person}{Xiang~Lisa Li} {and} \bibinfo{person}{Percy Liang}.} \bibinfo{year}{2021}\natexlab{}.
\newblock \showarticletitle{Prefix-tuning: Optimizing continuous prompts for generation}.
\newblock \bibinfo{journal}{\emph{arXiv preprint arXiv:2101.00190}} (\bibinfo{year}{2021}).
\newblock


\bibitem[Li et~al\mbox{.}(2023)]%
        {li2023textbooks}
\bibfield{author}{\bibinfo{person}{Yuanzhi Li}, \bibinfo{person}{Sébastien Bubeck}, \bibinfo{person}{Ronen Eldan}, \bibinfo{person}{Allie~Del Giorno}, \bibinfo{person}{Suriya Gunasekar}, {and} \bibinfo{person}{Yin~Tat Lee}.} \bibinfo{year}{2023}\natexlab{}.
\newblock \bibinfo{title}{Textbooks Are All You Need II: phi-1.5 technical report}.
\newblock
\showeprint[arxiv]{2309.05463}~[cs.CL]


\bibitem[Li et~al\mbox{.}(2022)]%
        {liCamelManagingData2022}
\bibfield{author}{\bibinfo{person}{Yiming Li}, \bibinfo{person}{Yanyan Shen}, {and} \bibinfo{person}{Lei Chen}.} \bibinfo{year}{2022}\natexlab{}.
\newblock \showarticletitle{Camel: {{Managing Data}} for {{Efficient Stream Learning}}}. In \bibinfo{booktitle}{\emph{Proceedings of the 2022 {{International Conference}} on {{Management}} of {{Data}}}} \emph{(\bibinfo{series}{{{SIGMOD}} '22})}. \bibinfo{publisher}{Association for Computing Machinery}, \bibinfo{address}{New York, NY, USA}, \bibinfo{pages}{1271--1285}.
\newblock
\showISBNx{978-1-4503-9249-5}
\href{https://doi.org/10.1145/3514221.3517836}{doi:\nolinkurl{10.1145/3514221.3517836}}


\bibitem[Li et~al\mbox{.}(2024e)]%
        {10.14778/3696435.3696440}
\bibfield{author}{\bibinfo{person}{Zhaodonghui Li}, \bibinfo{person}{Haitao Yuan}, \bibinfo{person}{Huiming Wang}, \bibinfo{person}{Gao Cong}, {and} \bibinfo{person}{Lidong Bing}.} \bibinfo{year}{2024}\natexlab{e}.
\newblock \showarticletitle{LLM-R2: A Large Language Model Enhanced Rule-Based Rewrite System for Boosting Query Efficiency}.
\newblock \bibinfo{journal}{\emph{Proc. VLDB Endow.}} \bibinfo{volume}{18}, \bibinfo{number}{1} (\bibinfo{date}{Sept.} \bibinfo{year}{2024}), \bibinfo{pages}{53–65}.
\newblock
\showISSN{2150-8097}
\href{https://doi.org/10.14778/3696435.3696440}{doi:\nolinkurl{10.14778/3696435.3696440}}


\bibitem[Liang et~al\mbox{.}(2024)]%
        {liangMultiLayerTransformersGradient2024}
\bibfield{author}{\bibinfo{person}{Yingyu Liang}, \bibinfo{person}{Zhizhou Sha}, \bibinfo{person}{Zhenmei Shi}, \bibinfo{person}{Zhao Song}, {and} \bibinfo{person}{Yufa Zhou}.} \bibinfo{year}{2024}\natexlab{}.
\newblock \bibinfo{title}{Multi-{{Layer Transformers Gradient Can}} Be {{Approximated}} in {{Almost Linear Time}}}.
\newblock
\showeprint[arxiv]{2408.13233}~[cs]
\href{https://doi.org/10.48550/arXiv.2408.13233}{doi:\nolinkurl{10.48550/arXiv.2408.13233}}


\bibitem[Lin(2004)]%
        {lin-2004-rouge}
\bibfield{author}{\bibinfo{person}{Chin-Yew Lin}.} \bibinfo{year}{2004}\natexlab{}.
\newblock \showarticletitle{{ROUGE}: A Package for Automatic Evaluation of Summaries}. In \bibinfo{booktitle}{\emph{Text Summarization Branches Out}}. \bibinfo{publisher}{Association for Computational Linguistics}, \bibinfo{address}{Barcelona, Spain}, \bibinfo{pages}{74--81}.
\newblock
\urldef\tempurl%
\url{https://aclanthology.org/W04-1013/}
\showURL{%
\tempurl}


\bibitem[Liu et~al\mbox{.}(2024)]%
        {DBLP:conf/icde/LiuDLLCZ24}
\bibfield{author}{\bibinfo{person}{Hanmo Liu}, \bibinfo{person}{Shimin Di}, \bibinfo{person}{Haoyang Li}, \bibinfo{person}{Shuangyin Li}, \bibinfo{person}{Lei Chen}, {and} \bibinfo{person}{Xiaofang Zhou}.} \bibinfo{year}{2024}\natexlab{}.
\newblock \showarticletitle{Effective Data Selection and Replay for Unsupervised Continual Learning}. In \bibinfo{booktitle}{\emph{40th {IEEE} International Conference on Data Engineering, {ICDE} 2024, Utrecht, The Netherlands, May 13-16, 2024}}. \bibinfo{publisher}{{IEEE}}, \bibinfo{pages}{1449--1463}.
\newblock
\href{https://doi.org/10.1109/ICDE60146.2024.00119}{doi:\nolinkurl{10.1109/ICDE60146.2024.00119}}


\bibitem[Lloyd(1982)]%
        {lloyd1982least}
\bibfield{author}{\bibinfo{person}{Stuart Lloyd}.} \bibinfo{year}{1982}\natexlab{}.
\newblock \showarticletitle{Least squares quantization in PCM}.
\newblock \bibinfo{journal}{\emph{IEEE transactions on information theory}} \bibinfo{volume}{28}, \bibinfo{number}{2} (\bibinfo{year}{1982}), \bibinfo{pages}{129--137}.
\newblock


\bibitem[Lou et~al\mbox{.}(2024)]%
        {10597732}
\bibfield{author}{\bibinfo{person}{Yuze Lou}, \bibinfo{person}{Chuan Lei}, \bibinfo{person}{Xiao Qin}, \bibinfo{person}{Zichen Wang}, \bibinfo{person}{Christos Faloutsos}, \bibinfo{person}{Rishita Anubhai}, {and} \bibinfo{person}{Huzefa Rangwala}.} \bibinfo{year}{2024}\natexlab{}.
\newblock \showarticletitle{DATALORE: Can a Large Language Model Find All Lost Scrolls in a Data Repository?}. In \bibinfo{booktitle}{\emph{2024 IEEE 40th International Conference on Data Engineering (ICDE)}}. \bibinfo{pages}{5170--5176}.
\newblock
\href{https://doi.org/10.1109/ICDE60146.2024.00388}{doi:\nolinkurl{10.1109/ICDE60146.2024.00388}}


\bibitem[Ma et~al\mbox{.}(2024)]%
        {10.1145/3677139}
\bibfield{author}{\bibinfo{person}{Lei Ma}, \bibinfo{person}{Lei Cao}, \bibinfo{person}{Peter~M. VanNostrand}, \bibinfo{person}{Dennis~M. Hofmann}, \bibinfo{person}{Yao Su}, {and} \bibinfo{person}{Elke~A. Rundensteiner}.} \bibinfo{year}{2024}\natexlab{}.
\newblock \showarticletitle{Pluto: Sample Selection for Robust Anomaly Detection on Polluted Log Data}.
\newblock \bibinfo{journal}{\emph{Proc. ACM Manag. Data}} \bibinfo{volume}{2}, \bibinfo{number}{4}, Article \bibinfo{articleno}{203} (\bibinfo{date}{Sept.} \bibinfo{year}{2024}), \bibinfo{numpages}{25}~pages.
\newblock
\href{https://doi.org/10.1145/3677139}{doi:\nolinkurl{10.1145/3677139}}


\bibitem[Marion et~al\mbox{.}(2023)]%
        {marionWhenLessMore2023}
\bibfield{author}{\bibinfo{person}{Max Marion}, \bibinfo{person}{Ahmet {\"U}st{\"u}n}, \bibinfo{person}{Luiza Pozzobon}, \bibinfo{person}{Alex Wang}, \bibinfo{person}{Marzieh Fadaee}, {and} \bibinfo{person}{Sara Hooker}.} \bibinfo{year}{2023}\natexlab{}.
\newblock \bibinfo{title}{When {{Less}} Is {{More}}: {{Investigating Data Pruning}} for {{Pretraining LLMs}} at {{Scale}}}.
\newblock
\showeprint[arxiv]{2309.04564}~[cs]
\href{https://doi.org/10.48550/arXiv.2309.04564}{doi:\nolinkurl{10.48550/arXiv.2309.04564}}


\bibitem[Mekala et~al\mbox{.}(2024)]%
        {mekalaSmallerLanguageModels2024}
\bibfield{author}{\bibinfo{person}{Dheeraj Mekala}, \bibinfo{person}{Alex Nguyen}, {and} \bibinfo{person}{Jingbo Shang}.} \bibinfo{year}{2024}\natexlab{}.
\newblock \bibinfo{title}{Smaller {{Language Models}} Are Capable of Selecting {{Instruction-Tuning Training Data}} for {{Larger Language Models}}}.
\newblock
\href{https://doi.org/10.48550/ARXIV.2402.10430}{doi:\nolinkurl{10.48550/ARXIV.2402.10430}}


\bibitem[Mirzasoleiman et~al\mbox{.}(2020)]%
        {mirzasoleimanCoresetsDataefficientTraining2020}
\bibfield{author}{\bibinfo{person}{Baharan Mirzasoleiman}, \bibinfo{person}{Jeff Bilmes}, {and} \bibinfo{person}{Jure Leskovec}.} \bibinfo{year}{2020}\natexlab{}.
\newblock \showarticletitle{Coresets for {{Data-efficient Training}} of {{Machine Learning Models}}}. In \bibinfo{booktitle}{\emph{Proceedings of the 37th {{International Conference}} on {{Machine Learning}}}}. \bibinfo{publisher}{PMLR}, \bibinfo{pages}{6950--6960}.
\newblock
\showeprint[arxiv]{1906.01827}~[cs, stat]


\bibitem[Narayan et~al\mbox{.}(2022)]%
        {10.14778/3574245.3574258}
\bibfield{author}{\bibinfo{person}{Avanika Narayan}, \bibinfo{person}{Ines Chami}, \bibinfo{person}{Laurel Orr}, {and} \bibinfo{person}{Christopher R\'{e}}.} \bibinfo{year}{2022}\natexlab{}.
\newblock \showarticletitle{Can Foundation Models Wrangle Your Data?}
\newblock \bibinfo{journal}{\emph{Proc. VLDB Endow.}} \bibinfo{volume}{16}, \bibinfo{number}{4} (\bibinfo{date}{Dec.} \bibinfo{year}{2022}), \bibinfo{pages}{738–746}.
\newblock
\showISSN{2150-8097}
\href{https://doi.org/10.14778/3574245.3574258}{doi:\nolinkurl{10.14778/3574245.3574258}}


\bibitem[Nguyen et~al\mbox{.}(2024)]%
        {nguyenMemoryefficientTrainingLLMs2024}
\bibfield{author}{\bibinfo{person}{Dang Nguyen}, \bibinfo{person}{Wenhan Yang}, \bibinfo{person}{Rathul Anand}, \bibinfo{person}{Yu Yang}, {and} \bibinfo{person}{Baharan Mirzasoleiman}.} \bibinfo{year}{2024}\natexlab{}.
\newblock \bibinfo{title}{Memory-Efficient {{Training}} of {{LLMs}} with {{Larger Mini-batches}}}.
\newblock
\showeprint[arxiv]{2407.19580}~[cs]
\href{https://doi.org/10.48550/arXiv.2407.19580}{doi:\nolinkurl{10.48550/arXiv.2407.19580}}


\bibitem[OpenAI(2023)]%
        {openai2023gpt4}
\bibfield{author}{\bibinfo{person}{OpenAI}.} \bibinfo{year}{2023}\natexlab{}.
\newblock \bibinfo{title}{GPT-4 Technical Report}.
\newblock
\showeprint[arxiv]{2303.08774}~[cs.CL]


\bibitem[Paul et~al\mbox{.}(2021)]%
        {paulDeepLearningData2021}
\bibfield{author}{\bibinfo{person}{Mansheej Paul}, \bibinfo{person}{Surya Ganguli}, {and} \bibinfo{person}{Gintare~Karolina Dziugaite}.} \bibinfo{year}{2021}\natexlab{}.
\newblock \showarticletitle{Deep {{Learning}} on a {{Data Diet}}: {{Finding Important Examples Early}} in {{Training}}}. In \bibinfo{booktitle}{\emph{Advances in {{Neural Information Processing Systems}}}}, Vol.~\bibinfo{volume}{34}. \bibinfo{publisher}{Curran Associates, Inc.}, \bibinfo{pages}{20596--20607}.
\newblock


\bibitem[Ren et~al\mbox{.}(2024)]%
        {ren2024purple}
\bibfield{author}{\bibinfo{person}{Tonghui Ren}, \bibinfo{person}{Yuankai Fan}, \bibinfo{person}{Zhenying He}, \bibinfo{person}{Ren Huang}, \bibinfo{person}{Jiaqi Dai}, \bibinfo{person}{Can Huang}, \bibinfo{person}{Yinan Jing}, \bibinfo{person}{Kai Zhang}, \bibinfo{person}{Yifan Yang}, {and} \bibinfo{person}{X~Sean Wang}.} \bibinfo{year}{2024}\natexlab{}.
\newblock \showarticletitle{Purple: Making a large language model a better sql writer}. In \bibinfo{booktitle}{\emph{2024 IEEE 40th International Conference on Data Engineering (ICDE)}}. IEEE, \bibinfo{pages}{15--28}.
\newblock


\bibitem[Sorscher et~al\mbox{.}(2022)]%
        {sorscherNeuralScalingLaws2022}
\bibfield{author}{\bibinfo{person}{Ben Sorscher}, \bibinfo{person}{Robert Geirhos}, \bibinfo{person}{Shashank Shekhar}, \bibinfo{person}{Surya Ganguli}, {and} \bibinfo{person}{Ari Morcos}.} \bibinfo{year}{2022}\natexlab{}.
\newblock \showarticletitle{Beyond Neural Scaling Laws: Beating Power Law Scaling via Data Pruning}.
\newblock \bibinfo{journal}{\emph{Advances in Neural Information Processing Systems}}  \bibinfo{volume}{35} (\bibinfo{date}{Dec.} \bibinfo{year}{2022}), \bibinfo{pages}{19523--19536}.
\newblock


\bibitem[Team(2024)]%
        {qwen2.5}
\bibfield{author}{\bibinfo{person}{Qwen Team}.} \bibinfo{year}{2024}\natexlab{}.
\newblock \bibinfo{title}{Qwen2.5: A Party of Foundation Models}.
\newblock
\urldef\tempurl%
\url{https://qwenlm.github.io/blog/qwen2.5/}
\showURL{%
\tempurl}


\bibitem[Thrush et~al\mbox{.}(2024)]%
        {thrushImprovingPretrainingData2024}
\bibfield{author}{\bibinfo{person}{Tristan Thrush}, \bibinfo{person}{Christopher Potts}, {and} \bibinfo{person}{Tatsunori Hashimoto}.} \bibinfo{year}{2024}\natexlab{}.
\newblock \bibinfo{title}{Improving {{Pretraining Data Using Perplexity Correlations}}}.
\newblock
\showeprint[arxiv]{2409.05816}~[cs, stat]
\href{https://doi.org/10.48550/arXiv.2409.05816}{doi:\nolinkurl{10.48550/arXiv.2409.05816}}


\bibitem[Tirumala et~al\mbox{.}(2023)]%
        {tirumalaD4ImprovingLLM2023}
\bibfield{author}{\bibinfo{person}{Kushal Tirumala}, \bibinfo{person}{Daniel Simig}, \bibinfo{person}{Armen Aghajanyan}, {and} \bibinfo{person}{Ari Morcos}.} \bibinfo{year}{2023}\natexlab{}.
\newblock \showarticletitle{D4: {{Improving LLM Pretraining}} via {{Document De-Duplication}} and {{Diversification}}}.
\newblock \bibinfo{journal}{\emph{Advances in Neural Information Processing Systems}}  \bibinfo{volume}{36} (\bibinfo{date}{Dec.} \bibinfo{year}{2023}), \bibinfo{pages}{53983--53995}.
\newblock


\bibitem[Tran et~al\mbox{.}(2024)]%
        {10.1093/jamia/ocae122}
\bibfield{author}{\bibinfo{person}{Hieu Tran}, \bibinfo{person}{Zhichao Yang}, \bibinfo{person}{Zonghai Yao}, {and} \bibinfo{person}{Hong Yu}.} \bibinfo{year}{2024}\natexlab{}.
\newblock \showarticletitle{{{BioInstruct}}: Instruction Tuning of Large Language Models for Biomedical Natural Language Processing}.
\newblock \bibinfo{journal}{\emph{Journal of the American Medical Informatics Association}} \bibinfo{volume}{31}, \bibinfo{number}{9} (\bibinfo{date}{June} \bibinfo{year}{2024}), \bibinfo{pages}{1821--1832}.
\newblock
\showISSN{1527-974X}
\showeprint{https://academic.oup.com/jamia/article-pdf/31/9/1821/58868340/ocae122.pdf}
\href{https://doi.org/10.1093/jamia/ocae122}{doi:\nolinkurl{10.1093/jamia/ocae122}}


\bibitem[Vaswani(2017)]%
        {vaswani2017attention}
\bibfield{author}{\bibinfo{person}{A Vaswani}.} \bibinfo{year}{2017}\natexlab{}.
\newblock \showarticletitle{Attention is all you need}.
\newblock \bibinfo{journal}{\emph{Advances in Neural Information Processing Systems}} (\bibinfo{year}{2017}).
\newblock


\bibitem[Wang et~al\mbox{.}(2022)]%
        {10.14778/3561261.3561267}
\bibfield{author}{\bibinfo{person}{Jiayi Wang}, \bibinfo{person}{Chengliang Chai}, \bibinfo{person}{Nan Tang}, \bibinfo{person}{Jiabin Liu}, {and} \bibinfo{person}{Guoliang Li}.} \bibinfo{year}{2022}\natexlab{}.
\newblock \showarticletitle{Coresets over multiple tables for feature-rich and data-efficient machine learning}.
\newblock \bibinfo{journal}{\emph{Proc. VLDB Endow.}} \bibinfo{volume}{16}, \bibinfo{number}{1} (\bibinfo{date}{Sept.} \bibinfo{year}{2022}), \bibinfo{pages}{64–76}.
\newblock
\showISSN{2150-8097}
\href{https://doi.org/10.14778/3561261.3561267}{doi:\nolinkurl{10.14778/3561261.3561267}}


\bibitem[Wang et~al\mbox{.}(2025)]%
        {wang2025datawhispererefficientdata}
\bibfield{author}{\bibinfo{person}{Shaobo Wang}, \bibinfo{person}{Xiangqi Jin}, \bibinfo{person}{Ziming Wang}, \bibinfo{person}{Jize Wang}, \bibinfo{person}{Jiajun Zhang}, \bibinfo{person}{Kaixin Li}, \bibinfo{person}{Zichen Wen}, \bibinfo{person}{Zhong Li}, \bibinfo{person}{Conghui He}, \bibinfo{person}{Xuming Hu}, {and} \bibinfo{person}{Linfeng Zhang}.} \bibinfo{year}{2025}\natexlab{}.
\newblock \bibinfo{title}{Data Whisperer: Efficient Data Selection for Task-Specific LLM Fine-Tuning via Few-Shot In-Context Learning}.
\newblock
\showeprint[arxiv]{2505.12212}~[cs.CL]
\urldef\tempurl%
\url{https://arxiv.org/abs/2505.12212}
\showURL{%
\tempurl}


\bibitem[Wei et~al\mbox{.}(2022)]%
        {wei2022chain}
\bibfield{author}{\bibinfo{person}{Jason Wei}, \bibinfo{person}{Xuezhi Wang}, \bibinfo{person}{Dale Schuurmans}, \bibinfo{person}{Maarten Bosma}, \bibinfo{person}{Fei Xia}, \bibinfo{person}{Ed Chi}, \bibinfo{person}{Quoc~V Le}, \bibinfo{person}{Denny Zhou}, {et~al\mbox{.}}} \bibinfo{year}{2022}\natexlab{}.
\newblock \showarticletitle{Chain-of-thought prompting elicits reasoning in large language models}.
\newblock \bibinfo{journal}{\emph{Advances in neural information processing systems}}  \bibinfo{volume}{35} (\bibinfo{year}{2022}), \bibinfo{pages}{24824--24837}.
\newblock


\bibitem[Xia et~al\mbox{.}(2024)]%
        {xiaLESSSelectingInfluential2024}
\bibfield{author}{\bibinfo{person}{Mengzhou Xia}, \bibinfo{person}{Sadhika Malladi}, \bibinfo{person}{Suchin Gururangan}, \bibinfo{person}{Sanjeev Arora}, {and} \bibinfo{person}{Danqi Chen}.} \bibinfo{year}{2024}\natexlab{}.
\newblock \showarticletitle{{{LESS}}: {{Selecting Influential Data}} for {{Targeted Instruction Tuning}}}. In \bibinfo{booktitle}{\emph{Forty-First {{International Conference}} on {{Machine Learning}}}}.
\newblock


\bibitem[Xinyang~Zhao(2024)]%
        {zhao2024llmdbdemo}
\bibfield{author}{\bibinfo{person}{Guoliang~Li Xinyang~Zhao, Xuanhe~Zhou}.} \bibinfo{year}{2024}\natexlab{}.
\newblock \bibinfo{title}{Chat2Data: An Interactive Data Analysis System with RAG, Vector Databases and LLMs}.
\newblock


\bibitem[Yang et~al\mbox{.}(2023)]%
        {yangSustainableLearningCoresets2023}
\bibfield{author}{\bibinfo{person}{Yu Yang}, \bibinfo{person}{Hao Kang}, {and} \bibinfo{person}{Baharan Mirzasoleiman}.} \bibinfo{year}{2023}\natexlab{}.
\newblock \showarticletitle{Towards {{Sustainable Learning}}: {{Coresets}} for {{Data-efficient Deep Learning}}}. In \bibinfo{booktitle}{\emph{Proceedings of the 40th {{International Conference}} on {{Machine Learning}}}}. \bibinfo{publisher}{PMLR}, \bibinfo{pages}{39314--39330}.
\newblock
\showISSN{2640-3498}
\urldef\tempurl%
\url{https://proceedings.mlr.press/v202/yang23g.html}
\showURL{%
\tempurl}


\bibitem[Yang et~al\mbox{.}(2024)]%
        {yangSmallToLargeS2LScalable2024}
\bibfield{author}{\bibinfo{person}{Yu Yang}, \bibinfo{person}{Siddhartha Mishra}, \bibinfo{person}{Jeffrey~N Chiang}, {and} \bibinfo{person}{Baharan Mirzasoleiman}.} \bibinfo{year}{2024}\natexlab{}.
\newblock \bibinfo{title}{{{SmallToLarge}} ({{S2L}}): {{Scalable Data Selection}} for {{Fine-tuning Large Language Models}} by {{Summarizing Training Trajectories}} of {{Small Models}}}.
\newblock
\href{https://doi.org/10.48550/ARXIV.2403.07384}{doi:\nolinkurl{10.48550/ARXIV.2403.07384}}


\bibitem[Yin and Rush(2024)]%
        {yin2024compute}
\bibfield{author}{\bibinfo{person}{Junjie~Oscar Yin} {and} \bibinfo{person}{Alexander~M Rush}.} \bibinfo{year}{2024}\natexlab{}.
\newblock \showarticletitle{Compute-constrained data selection}.
\newblock \bibinfo{journal}{\emph{arXiv preprint arXiv:2410.16208}} (\bibinfo{year}{2024}).
\newblock


\bibitem[Yu et~al\mbox{.}(2024a)]%
        {yuDiversifyConquerDiversityCentric2024}
\bibfield{author}{\bibinfo{person}{Simon Yu}, \bibinfo{person}{Liangyu Chen}, \bibinfo{person}{Sara Ahmadian}, {and} \bibinfo{person}{Marzieh Fadaee}.} \bibinfo{year}{2024}\natexlab{a}.
\newblock \bibinfo{title}{Diversify and {{Conquer}}: {{Diversity-Centric Data Selection}} with {{Iterative Refinement}}}.
\newblock
\showeprint[arxiv]{2409.11378}~[cs]
\href{https://doi.org/10.48550/arXiv.2409.11378}{doi:\nolinkurl{10.48550/arXiv.2409.11378}}


\bibitem[Yu et~al\mbox{.}(2024b)]%
        {yuMATESModelAwareData2024}
\bibfield{author}{\bibinfo{person}{Zichun Yu}, \bibinfo{person}{Spandan Das}, {and} \bibinfo{person}{Chenyan Xiong}.} \bibinfo{year}{2024}\natexlab{b}.
\newblock \bibinfo{title}{{{MATES}}: {{Model-Aware Data Selection}} for {{Efficient Pretraining}} with {{Data Influence Models}}}.
\newblock
\href{https://doi.org/10.48550/ARXIV.2406.06046}{doi:\nolinkurl{10.48550/ARXIV.2406.06046}}


\bibitem[Yue et~al\mbox{.}(2024)]%
        {yue2024mammoth}
\bibfield{author}{\bibinfo{person}{Xiang Yue}, \bibinfo{person}{Xingwei Qu}, \bibinfo{person}{Ge Zhang}, \bibinfo{person}{Yao Fu}, \bibinfo{person}{Wenhao Huang}, \bibinfo{person}{Huan Sun}, \bibinfo{person}{Yu Su}, {and} \bibinfo{person}{Wenhu Chen}.} \bibinfo{year}{2024}\natexlab{}.
\newblock \showarticletitle{{MA}mmo{TH}: Building Math Generalist Models through Hybrid Instruction Tuning}. In \bibinfo{booktitle}{\emph{The Twelfth International Conference on Learning Representations}}.
\newblock
\urldef\tempurl%
\url{https://openreview.net/forum?id=yLClGs770I}
\showURL{%
\tempurl}


\bibitem[Zhang et~al\mbox{.}(2024b)]%
        {zhangHarnessingDiversityImportant2024}
\bibfield{author}{\bibinfo{person}{Chi Zhang}, \bibinfo{person}{Huaping Zhong}, \bibinfo{person}{Kuan Zhang}, \bibinfo{person}{Chengliang Chai}, \bibinfo{person}{Rui Wang}, \bibinfo{person}{Xinlin Zhuang}, \bibinfo{person}{Tianyi Bai}, \bibinfo{person}{Jiantao Qiu}, \bibinfo{person}{Lei Cao}, \bibinfo{person}{Ju Fan}, \bibinfo{person}{Ye Yuan}, \bibinfo{person}{Guoren Wang}, {and} \bibinfo{person}{Conghui He}.} \bibinfo{year}{2024}\natexlab{b}.
\newblock \bibinfo{title}{Harnessing {{Diversity}} for {{Important Data Selection}} in {{Pretraining Large Language Models}}}.
\newblock
\showeprint[arxiv]{2409.16986}~[cs]
\href{https://doi.org/10.48550/arXiv.2409.16986}{doi:\nolinkurl{10.48550/arXiv.2409.16986}}


\bibitem[Zhang et~al\mbox{.}(2024a)]%
        {zhangTAGCOSTaskagnosticGradient2024}
\bibfield{author}{\bibinfo{person}{Jipeng Zhang}, \bibinfo{person}{Yaxuan Qin}, \bibinfo{person}{Renjie Pi}, \bibinfo{person}{Weizhong Zhang}, \bibinfo{person}{Rui Pan}, {and} \bibinfo{person}{Tong Zhang}.} \bibinfo{year}{2024}\natexlab{a}.
\newblock \bibinfo{title}{{{TAGCOS}}: {{Task-agnostic Gradient Clustered Coreset Selection}} for {{Instruction Tuning Data}}}.
\newblock
\showeprint[arxiv]{2407.15235}~[cs]
\href{https://doi.org/10.48550/arXiv.2407.15235}{doi:\nolinkurl{10.48550/arXiv.2407.15235}}


\bibitem[Zhao et~al\mbox{.}(2024)]%
        {zhaoLongMoreAlignment2024}
\bibfield{author}{\bibinfo{person}{Hao Zhao}, \bibinfo{person}{Maksym Andriushchenko}, \bibinfo{person}{Francesco Croce}, {and} \bibinfo{person}{Nicolas Flammarion}.} \bibinfo{year}{2024}\natexlab{}.
\newblock \bibinfo{title}{Long {{Is More}} for {{Alignment}}: {{A Simple}} but {{Tough-to-Beat Baseline}} for {{Instruction Fine-Tuning}}}.
\newblock
\showeprint[arxiv]{2402.04833}~[cs]
\href{https://doi.org/10.48550/arXiv.2402.04833}{doi:\nolinkurl{10.48550/arXiv.2402.04833}}


\bibitem[Zheng et~al\mbox{.}(2023)]%
        {zhengCoveragecentricCoresetSelection2023}
\bibfield{author}{\bibinfo{person}{Haizhong Zheng}, \bibinfo{person}{Rui Liu}, \bibinfo{person}{Fan Lai}, {and} \bibinfo{person}{Atul Prakash}.} \bibinfo{year}{2023}\natexlab{}.
\newblock \bibinfo{title}{Coverage-Centric {{Coreset Selection}} for {{High Pruning Rates}}}.
\newblock
\showeprint[arxiv]{2210.15809}~[cs]
\href{https://doi.org/10.48550/arXiv.2210.15809}{doi:\nolinkurl{10.48550/arXiv.2210.15809}}


\bibitem[Zhou et~al\mbox{.}(2023)]%
        {zhouLoBaSSGaugingLearnability2023}
\bibfield{author}{\bibinfo{person}{Haotian Zhou}, \bibinfo{person}{Tingkai Liu}, \bibinfo{person}{Qianli Ma}, \bibinfo{person}{Jianbo Yuan}, \bibinfo{person}{Pengfei Liu}, \bibinfo{person}{Yang You}, {and} \bibinfo{person}{Hongxia Yang}.} \bibinfo{year}{2023}\natexlab{}.
\newblock \bibinfo{title}{{{LoBaSS}}: {{Gauging Learnability}} in {{Supervised Fine-tuning Data}}}.
\newblock
\showeprint[arxiv]{2310.13008}~[cs]
\href{https://doi.org/10.48550/arXiv.2310.13008}{doi:\nolinkurl{10.48550/arXiv.2310.13008}}


\bibitem[Zhou et~al\mbox{.}(2025)]%
        {zhou2025cracksql}
\bibfield{author}{\bibinfo{person}{Wei Zhou}, \bibinfo{person}{Yuyang Gao}, \bibinfo{person}{Xuanhe Zhou}, {and} \bibinfo{person}{Guoliang Li}.} \bibinfo{year}{2025}\natexlab{}.
\newblock \showarticletitle{{Cracking SQL Barriers:} {An} LLM-based Dialect Transaltion System}.
\newblock \bibinfo{journal}{\emph{Proc. {ACM} Manag. Data}} \bibinfo{volume}{3}, \bibinfo{number}{3 (SIGMOD)} (\bibinfo{year}{2025}).
\newblock


\bibitem[Zhou et~al\mbox{.}(2024)]%
        {10.14778/3675034.3675043}
\bibfield{author}{\bibinfo{person}{Xuanhe Zhou}, \bibinfo{person}{Guoliang Li}, \bibinfo{person}{Zhaoyan Sun}, \bibinfo{person}{Zhiyuan Liu}, \bibinfo{person}{Weize Chen}, \bibinfo{person}{Jianming Wu}, \bibinfo{person}{Jiesi Liu}, \bibinfo{person}{Ruohang Feng}, {and} \bibinfo{person}{Guoyang Zeng}.} \bibinfo{year}{2024}\natexlab{}.
\newblock \showarticletitle{D-Bot: Database Diagnosis System using Large Language Models}.
\newblock \bibinfo{journal}{\emph{Proc. VLDB Endow.}} \bibinfo{volume}{17}, \bibinfo{number}{10} (\bibinfo{date}{June} \bibinfo{year}{2024}), \bibinfo{pages}{2514–2527}.
\newblock
\showISSN{2150-8097}
\href{https://doi.org/10.14778/3675034.3675043}{doi:\nolinkurl{10.14778/3675034.3675043}}


\end{thebibliography}

\section{Appendix}\label{appendix}

\subsection{Proof of Theorem~\ref{thm:select}}\label{eq:proof4-1}
\begin{proof}
  The gradient with respect to all parameters is a concatenation of gradients with respect to each parameter's weights. Therefore, by the triangle inequality:
  \begin{align*}
    & \|\nabla_{\boldsymbol{\theta}} \mathcal{L}(X_i, \boldsymbol{\theta}) - \nabla_{\boldsymbol{\theta}} \mathcal{L}(X_j, \boldsymbol{\theta})\| \\
    \leq & \sum_{\mathbf{W}}\|\nabla_{\mathbf{W}} \mathcal{L}(X_i, \boldsymbol{\theta}) - \nabla_{\mathbf{W}} \mathcal{L}(X_j, \boldsymbol{\theta})\|
  \end{align*}
  where $\mathbf{W} \in\{\mathbf{W}_1, \mathbf{W}_2, \mathbf{W}_Q, \mathbf{W}_K, \mathbf{W}_V \mathbf{W}_O\}$ We analyze all parameters $\mathbf{W}_1$, $\mathbf{W}_2$, $\mathbf{W}_Q$, $\mathbf{W}_K$, $\mathbf{W}_V$, and $\mathbf{W}_O$, and combine them into the final form. Since the layer norm and activation function are generally Lipschitz continuous, we omit them from the equation.

  Following~\cite{liangMultiLayerTransformersGradient2024}, we define:
  \begin{align*}
    G(\mathbf{X}) & = \nabla_{\mathbf{Z}}\mathcal{L}(X, \boldsymbol{\theta}) \in \mathbb{R}^{T\times d} \\
    q(\mathbf{X}) &= G(\mathbf{X}) \cdot (\mathbf{X}\mathbf{W}_V\mathbf{W}_O)^T \in \mathbb{R}^{T\times T} \\
    p_1(\mathbf{X})&=h(\mathbf{X})\odot q(\mathbf{X}) \in \mathbb{R}^{T\times T}
  \end{align*}
  \noindent where $T$ is the sequential length, $d$ is the hidden dimension, and $\odot$ is the Hadamard product. In the following analysis, we mainly use the triangle inequality and the Cauchy-Schwarz inequality, and we state them here to avoid repeated descriptions. For the gradient match associated with $\mathbf{W}_1$, we have:
  \begin{align*}
& \left\|\nabla_{\mathbf{W}_1}\mathcal{L}(X_i,\boldsymbol{\theta})) - \nabla_{\mathbf{W}_1}\mathcal{L}(X_j,\boldsymbol{\theta}))\right\| \nonumber\\
& \leq \left\|\mathbf{W}_O^T\mathbf{W}_V^T\mathbf{X}_i^T\mathbf{Z}_i^T\nabla_{\hat{\mathbf{O}_i}}\mathcal{L}(X_i,\boldsymbol{\theta}))\mathbf{W}_2^T \right. \nonumber \\
& - \left.\mathbf{W}_O^T\mathbf{W}_V^T\mathbf{X}_j^T\mathbf{Z}_j^T\nabla_{\hat{\mathbf{O}_i}}\mathcal{L}(X_j,\boldsymbol{\theta}))\mathbf{W}_2^T\right\| \nonumber\\
& \leq c_1^1\left\|\nabla_{\hat{\mathbf{O}_i}}\mathcal{L}(X_i,\boldsymbol{\theta})) - \nabla_{\hat{\mathbf{O}_i}}\mathcal{L}(X_j,\boldsymbol{\theta}))\right\| \nonumber\\
& + c_3^1\left\|\mathbf{X}_i-\mathbf{X}_j\right\|  + c_2^1 \left\|\mathbf{Z}_i-\mathbf{Z}_j\right\| \\
& \text{where}\  c_1^1=\left\|\mathbf{Z}_i\mathbf{X}_i\right\|\left\|\mathbf{W}_V\mathbf{W}_O\right\|\left\|\mathbf{W}_2\right\|, \nonumber \\
& c^1_3=\left\|\mathbf{W}_V\mathbf{W}_O\right\|\left\|\nabla_{\hat{\mathbf{O}_i}}\mathcal{L}(X_j,\boldsymbol{\theta}))\right\|\left\|\mathbf{Z}_i\right\|\left\|\mathbf{W}_2\right\|, \nonumber \\
& c_2^1=\left\|\mathbf{W}_V\mathbf{W}_O\right\|\left\|\nabla_{\hat{\mathbf{O}_i}}\mathcal{L}(X_j,\boldsymbol{\theta}))\right\|\left\|\mathbf{X}_j\right\|\left\|\mathbf{W}_2\right\| \nonumber
\end{align*}

For the gradient match associated with $\mathbf{W}_2$, we have:

\begin{align*}
& \left\|\nabla_{\mathbf{W}_2}\mathcal{L}(X_i,\boldsymbol{\theta})) - \nabla_{\mathbf{W}_2}\mathcal{L}(X_j,\boldsymbol{\theta}))\right\| \nonumber\\
&
\leq \left\|\mathbf{W}_1^T\mathbf{W}_O^T\mathbf{W}_V^T\mathbf{X}_i^T\mathbf{Z}_i^T\nabla_{\hat{\mathbf{O}_i}}\mathcal{L}(X_i,\boldsymbol{\theta}))\right. \nonumber \\
& - \left. \mathbf{W}_1^T\mathbf{W}_O^T\mathbf{W}_V^T\mathbf{X}_j^T\mathbf{Z}_j^T\nabla_{\hat{\mathbf{O}_i}}\mathcal{L}(X_j,\boldsymbol{\theta}))\mathbf{W}_2^T\right\|
\nonumber\\
& \leq c^2_1\left\|\nabla_{\hat{\mathbf{O}_i}}\mathcal{L}(X_i,\boldsymbol{\theta})) - \nabla_{\hat{\mathbf{O}_i}}\mathcal{L}(X_j,\boldsymbol{\theta}))\right\| \nonumber\\
& + c^2_3 \left\|\mathbf{X}_i-\mathbf{X}_j\right\|+ c^2_2\left\|\mathbf{Z}_i-\mathbf{Z}_j\right\| \\
& \text{where}\
c^2_1=\left\|\mathbf{Z}_i\mathbf{X}_i\right\|\left\|\mathbf{W}_V\mathbf{W}_O\mathbf{W}_1\right\|\nonumber \\
& c^2_3=\left\|\mathbf{W}_V\mathbf{W}_O\mathbf{W}_1\right\|\left\|\nabla_{\hat{\mathbf{O}_i}}\mathcal{L}(X_j,\boldsymbol{\theta}))\right\|\left\|\mathbf{Z}_i\right\|\nonumber\\
&c^2_2=\left\|\mathbf{W}_V\mathbf{W}_O\mathbf{W}_1\right\|\left\|\nabla_{\hat{\mathbf{O}_i}}\mathcal{L}(X_j,\boldsymbol{\theta}))\right\|\left\|\mathbf{X}_j\right\| \nonumber
\end{align*}

For the gradient match associated with $\mathbf{W}_O$, we have:
\begin{align*}
& \left\|\nabla_{\mathbf{W}_O}\mathcal{L}(X_i,\boldsymbol{\theta})) - \nabla_{\mathbf{W}_O}\mathcal{L}(X_j,\boldsymbol{\theta}))\right\| \nonumber\\
& \leq \left\|\mathbf{W}_V^T\mathbf{X}_i^T\mathbf{Z}_i\nabla_{\hat{\mathbf{O}_i}}\mathcal{L}(X_i,\boldsymbol{\theta}))\mathbf{W}_2^T\mathbf{W}_1^T \right. \nonumber \\
& - \left. \mathbf{W}_V^T\mathbf{X}_j^T\mathbf{Z}_j\nabla_{\hat{\mathbf{O}_i}}\mathcal{L}(X_j,\boldsymbol{\theta}))\mathbf{W}_2^T\mathbf{W}_1^T\right\| \nonumber\\
& \leq c^O_1\left\|\nabla_{\hat{\mathbf{O}_i}}\mathcal{L}(X_i,\boldsymbol{\theta})) - \nabla_{\hat{\mathbf{O}_i}}\mathcal{L}(X_j,\boldsymbol{\theta}))\right\| \nonumber\\
& + c^O_3\left\|\mathbf{X}_i-\mathbf{X}_j\right\|\nonumber\\
& + c_2^O \left\|\mathbf{Z}_i-\mathbf{Z}_j\right\| \\
& \text{where}\ c^O_1=\left\|\mathbf{W}_1\mathbf{W}_2\right\|\left\|\mathbf{X}_i^T\mathbf{Z}_i\right\|\left\|\mathbf{W}_V\right\| \nonumber \\
& c^O_3=\left\|\mathbf{W}_1\mathbf{W}_2\right\|\left\|\mathbf{W}_V\right\|\left\|\nabla_{\hat{\mathbf{O}_i}}\mathcal{L}(X_j,\boldsymbol{\theta}))\right\|\left\|\mathbf{Z}_i\right\| \nonumber \\
& c_2^O=\left\|\mathbf{W}_1\mathbf{W}_2\right\|\left\|\mathbf{W}_V\right\|\left\|\nabla_{\hat{\mathbf{O}_i}}\mathcal{L}(X_j,\boldsymbol{\theta}))\right\|\left\|\mathbf{X}_j\right\| \nonumber
\end{align*}

For the gradient match associated with $\mathbf{W}_V$, we have:
\begin{align*}
& \left\|\nabla_{\mathbf{W}_V}\mathcal{L}(X_i,\boldsymbol{\theta})) - \nabla_{\mathbf{W}_V}\mathcal{L}(X_j,\boldsymbol{\theta}))\right\| \nonumber\\
& \leq \left\|\mathbf{X}_i^T\mathbf{Z}_i\nabla_{\hat{\mathbf{O}_i}}\mathcal{L}(X_i,\boldsymbol{\theta}))\mathbf{W}_2^T\mathbf{W}_1^T\mathbf{W}_O^T \right. \nonumber \\
& - \left. \mathbf{X}_j^T\mathbf{Z}_j\nabla_{\hat{\mathbf{O}_i}}\mathcal{L}(X_j,\boldsymbol{\theta}))\mathbf{W}_2^T\mathbf{W}_1^T\mathbf{W}_O^T\right\| \nonumber\\
& \leq c^V_1\left\|\nabla_{\hat{\mathbf{O}_i}}\mathcal{L}(X_i,\boldsymbol{\theta})) - \nabla_{\hat{\mathbf{O}_i}}\mathcal{L}(X_j,\boldsymbol{\theta}))\right\| \nonumber\\
& + c^V_3\left\|\mathbf{X}_i-\mathbf{X}_j\right\| + c^V_2\left\|\mathbf{Z}_i-\mathbf{Z}_j\right\| \\
& \text{where} \ c^V_1=\left\|\mathbf{Z}_i\mathbf{X}_i^T\right\|\left\|\mathbf{W}_O\mathbf{W}_1\mathbf{W}_2\right\|\nonumber \\
& c^V_3=\left\|\mathbf{W}_O\mathbf{W}_1\mathbf{W}_2\right\|\left\|\nabla_{\hat{\mathbf{O}_i}}\mathcal{L}(X_j,\boldsymbol{\theta}))\right\|\left\|\mathbf{Z}_i\right\| \nonumber \\
& c^V_2=\left\|\mathbf{W}_O\mathbf{W}_1\mathbf{W}_2\right\|\left\|\nabla_{\hat{\mathbf{O}_i}}\mathcal{L}(X_j,\boldsymbol{\theta}))\right\|\left\|\mathbf{X}_j\right\| \nonumber
\end{align*}

For gradient match associated with $\mathbf{W}_Q$ and $\mathbf{W}_K$, we first define:
\begin{align*}
& c_3=\left\|\mathbf{Z}_i\odot\nabla_{\hat{\mathbf{O}_i}}\mathcal{L}(X_i,\boldsymbol{\theta}))\mathbf{X}_i^T\right\| \\
& \left\|\mathbf{W}_1\mathbf{W}_2\right\|\left\|\mathbf{W}_V\mathbf{W}_O\right\|\left(\left\|\mathbf{X}_i\right\| + \left\|\mathbf{X}_j\right\|\right)
\\
& c_4=\left(\left\|\mathbf{W}_1\mathbf{W}_2\right\|\left\|\mathbf{W}_V\mathbf{W}_O\right\|\left\|\mathbf{X}_j\right\|^2\right)  \\
& c_5=\left(\left\|\text{diag}(p_1(\mathbf{X}_i)\vmathbb{1}_n)\mathbf{Z}_i\right\|\left(\left\|\mathbf{X}_i\right\| + \left\|\mathbf{X}_j\right\|\right)\right)  \\
& c_6=\left(\left\|\text{diag}(p_1(\mathbf{X}_i)\vmathbb{1}_n)\mathbf{Z}_i\right\|\right)  \\
& c_7=\left(\left\|\mathbf{X}_j\right\|^2\left\|h(\mathbf{X})_j\right\|\right)
\end{align*}
and we further define:
\begin{align*}
& c_3^Q=(c_3+c_5)\left\|\mathbf{W}_K\right\|, c_4^Q=c_4\left\|\mathbf{W}_K\right\| \\
& c_2^Q=c_6\left\|\mathbf{W}_K\right\|, c_5^Q=c_7\left\|\mathbf{W}_K\right\| \\
& c_3^K=(c_3+c_5)\left\|\mathbf{W}_Q\right\|, c_4^K=c_4\left\|\mathbf{W}_Q\right\| \\
& c_2^K=c_6\left\|\mathbf{W}_Q\right\|, c_5^K=c_7\left\|\mathbf{W}_Q\right\|
\end{align*}

So for $\mathbf{W}_Q$ we have:
\begin{align*}
& \left\|\nabla_{\mathbf{W}_Q}\mathcal{L}(X_i,\boldsymbol{\theta})) - \nabla_{\mathbf{W}_Q}\mathcal{L}(X_j,\boldsymbol{\theta}))\right\| \nonumber\\
& \leq c_3^Q \left\|\mathbf{X}_i - \mathbf{X}_j\right\| \nonumber\\
& +
c_4^Q \left\|\mathbf{Z}_i\odot\nabla_{\hat{\mathbf{O}_i}}\mathcal{L}(X_i,\boldsymbol{\theta}))\mathbf{X}_i^T - \mathbf{Z}_j\odot\nabla_{\hat{\mathbf{O}_i}}\mathcal{L}(X_j,\boldsymbol{\theta}))\mathbf{X}_j^T\right\| \nonumber\\
& + c_2^Q \left\|\mathbf{Z}_i - \mathbf{Z}_j\right\| \nonumber\\
& + c_5^Q \left\|\text{diag}(p_1(\mathbf{X}_i)\vmathbb{1}_n) - \text{diag}(p_1(\mathbf{X}_j)\vmathbb{1}_n)\right\|
\end{align*}

Similar for $\mathbf{W}_K$ we have:
\begin{align*}
& \left\|\nabla_{\mathbf{W}_K}\mathcal{L}(X_i,\boldsymbol{\theta})) - \nabla_{\mathbf{W}_K}\mathcal{L}(X_j,\boldsymbol{\theta}))\right\|\nonumber \\
& \leq c_3^K \left\|\mathbf{X}_i - \mathbf{X}_j\right\| \nonumber \\
& +
c_4^K \left\|\mathbf{Z}_i\odot\nabla_{\hat{\mathbf{O}_i}}\mathcal{L}(X_i,\boldsymbol{\theta}))\mathbf{X}_i^T - \mathbf{Z}_j\odot\nabla_{\hat{\mathbf{O}_i}}\mathcal{L}(X_j,\boldsymbol{\theta}))\mathbf{X}_j^T\right\|\nonumber \\
& + c_2^K \left\|\mathbf{Z}_i - \mathbf{Z}_j\right\| \nonumber\\
& + c_5^K \left\|\text{diag}(p_1(\mathbf{X}_i)\vmathbb{1}_n) - \text{diag}(p_1(\mathbf{X}_j)\vmathbb{1}_n)\right\|
\end{align*}

Summarize the above, we have:
\begin{align}
& \left\|\nabla_{\boldsymbol{\theta}}\mathcal{L}(X_i,\boldsymbol{\theta})) - \nabla_{\boldsymbol{\theta}}\mathcal{L}(X_j,\boldsymbol{\theta}))\right\| \nonumber\\
& \leq \sum_{k\in \{Q,K,V,O,1,2\}} \left\|\nabla_{\mathbf{W}_k}\mathcal{L}(X_i,\boldsymbol{\theta})) - \nabla_{\mathbf{W}_k}\mathcal{L}(X_j,\boldsymbol{\theta}))\right\| \nonumber\\
& \leq (c_1^1+c_1^2+c_1^O+c_1^V) \cdot \nonumber \\
& \left\|\nabla_{\hat{\mathbf{O}_i}}\mathcal{L}(X_i,\boldsymbol{\theta})) - \nabla_{\hat{\mathbf{O}_i}}\mathcal{L}(X_j,\boldsymbol{\theta}))\right\| \label{eq:1}\\
& + (c_2^1+c_2^2+c_2^O+c_2^V) \left\|\mathbf{Z}_i-\mathbf{Z}_j\right\| \label{eq:2}\\
& + (c_3^1+c_3^2+c_3^O+c_3^V+c_3^Q+c_3^K) \left\|\mathbf{X}_i - \mathbf{X}_j\right\| \label{eq:3}\\
& + (c_4^Q+c_4^K) \cdot \nonumber \\
& \left\|\mathbf{Z}_i\odot\nabla_{\hat{\mathbf{O}_i}}\mathcal{L}(X_i,\boldsymbol{\theta}))\mathbf{X}_i^T - \mathbf{Z}_j\odot\nabla_{\hat{\mathbf{O}_i}}\mathcal{L}(X_j,\boldsymbol{\theta}))\mathbf{X}_j^T\right\| \label{eq:4}\\
& + (c_5^Q+c_5^K) \left\|\text{diag}(p_1(\mathbf{X}_i)\vmathbb{1}_n) - \text{diag}(p_1(\mathbf{X}_j)\vmathbb{1}_n)\right\| \label{eq:5}
\end{align}

With local Lipschitz continuity of the attention mechanism \cite{castin2024how}, we have for Equation~\eqref{eq:3}
\begin{align*}
& \left\|\mathbf{Z}_i - \mathbf{Z}_j\right\| \leq \alpha\left\|\mathbf{X}_i - \mathbf{X}_j\right\|
\end{align*}

Also, we have that for Equation~\eqref{eq:4}
\begin{align*}
& \left\|\mathbf{Z}_i\odot\nabla_{\hat{\mathbf{O}_i}}\mathcal{L}(X_i,\boldsymbol{\theta}))\mathbf{X}_i^T - \mathbf{Z}_j\odot\nabla_{\hat{\mathbf{O}_i}}\mathcal{L}(X_j,\boldsymbol{\theta}))\mathbf{X}_j^T\right\| \\
& \leq max_{l\in\{i,j\}, p, q<n}\left\|\mathbf{Z}_l^{p,q}\right\| \cdot \\
& \left\|\nabla_{\hat{\mathbf{O}_i}}\mathcal{L}(X_i,\boldsymbol{\theta}))\mathbf{X}_i^T - \nabla_{\hat{\mathbf{O}_i}}\mathcal{L}(X_j,\boldsymbol{\theta}))\mathbf{X}_j^T\right\| \\
& \leq \beta \left\|\nabla_{\hat{\mathbf{O}_i}}\mathcal{L}(X_i,\boldsymbol{\theta}))\mathbf{X}_i^T - \nabla_{\hat{\mathbf{O}_i}}\mathcal{L}(X_j,\boldsymbol{\theta}))\mathbf{X}_j^T\right\| \\
& \leq \beta\left\|\mathbf{X}_i\right\|\left\|\nabla_{\hat{\mathbf{O}_i}}\mathcal{L}(X_i,\boldsymbol{\theta}) - \nabla_{\hat{\mathbf{O}_i}}\mathcal{L}(X_j,\boldsymbol{\theta})\right\| \\
& + \beta\left\|\nabla_{\hat{\mathbf{O}_i}}\mathcal{L}(X_j,\boldsymbol{\theta})\right\|\left\|\mathbf{X}_i - \mathbf{X}_j\right\|
\end{align*}

And for the Equation~\eqref{eq:5}, similarly, we can relax it to a combination of $\left\|\nabla_{\hat{\mathbf{O}_i}}\mathcal{L}(X_i,\boldsymbol{\theta})) - \nabla_{\hat{\mathbf{O}_i}}\mathcal{L}(X_j,\boldsymbol{\theta}))\right\|$ and $\left\|\mathbf{X}_i - \mathbf{X}_j\right\|$.

For simpler notation and analysis, we directly denote the remaining constant terms in Equation~\eqref{eq:4} and Equation~\eqref{eq:5} associated with $\left\|\nabla_{\hat{\mathbf{O}_i}}\mathcal{L}(X_i,\boldsymbol{\theta}) - \nabla_{\hat{\mathbf{O}_i}}\mathcal{L}(X_j,\boldsymbol{\theta})\right\|$ as $C_1$, and $\left\|\mathbf{X}_i - \mathbf{X}_j\right\|$ as $C_2$.

Then, define $c_1=c_3^1+c_3^2+c_3^O+c_3^V+c_3^Q+c_3^K+C_1$ and $c_2=c_1^1+c_1^2+c_1^O+c_1^V + \alpha (c_2^1+c_2^2+c_2^O+c_2^V) + C_2$, we have
\begin{align*}
&\|\nabla_{\boldsymbol{\theta}} \mathcal{L}(X_i, \boldsymbol{\theta}) - \nabla_{\boldsymbol{\theta}} \mathcal{L}(X_j, \boldsymbol{\theta})\| \\
\leq& c_1\|\mathbf{X}_i - \mathbf{X}_j\| + c_2\|\nabla_{\hat{\mathbf{O}_i}}\mathcal{L}(X_i,\boldsymbol{\theta})-\nabla_{\hat{\mathbf{O}_j}}\mathcal{L}(X_j,\boldsymbol{\theta})\| \\
\leq & c_1\|\mathbf{X}_i - \mathbf{X}_j\| + c_2\|\nabla_{\hat{\mathbf{O}_i}}\mathcal{L}(X_i,\boldsymbol{\theta})\| + c_2\|\nabla_{\hat{\mathbf{O}_j}}\mathcal{L}(X_j,\boldsymbol{\theta})\| \\
\end{align*}

Define $c_3=\| \sum_{X_i \in \mathcal{D}_{train} }{ \nabla_{\mathbf{\hat{O}}_i} \mathcal{L}(X_i,\boldsymbol{\theta})}\|$ then we have
\begin{align*}
&  \sum_{X_i\in\mathcal{D}_{train}}\min_{X_j\in\mathcal{D}_{core}}\|\nabla_{\boldsymbol{\theta}} \mathcal{L}(X_i, \boldsymbol{\theta}) - \nabla_{\boldsymbol{\theta}} \mathcal{L}(X_j, \boldsymbol{\theta})\|  \nonumber\\
\leq &
\sum_{X_i\in \mathcal{D}_{train}} \min_{X_j\in \mathcal{D}_{core}} c_1\left\|\mathbf{X}_i - \mathbf{X}_j \right\|  \nonumber \\
& +  c_2\|
\nabla_{\mathbf{\hat{O}}_j}\mathcal{L}(X_j,\boldsymbol{\theta})\| +c_3
\end{align*}
Thus, Theorem~\ref{thm:select} is proved.
\end{proof}

\subsection{Proof of Theorem~\ref{thm:nphard}}

\begin{proof}
We show that the computation of the representation score can be reduced from the Maximum Coverage Problem~\cite{hochbaum1997approximating}, which is known to be NP-hard. Given a universe $U$ of elements, a collection of subsets $\{S\}_j^n$, and a budget k, MCP is going to pick a fixed number of sets such that the cardinality of the union of chosen sets $|\cup S_j|$ is maximized. Consider each sample $X_i \in \mathcal{D}_{train}$ as an item that needs to be covered, and each candidate $X_j\in \mathcal{D}_{core}$ as a set that covers the samples in $\mathcal{D}_{train}$. We define the cosine similarity as an indicator function: $\cos(X_i, X_j)=1$ if $X_i$ is covered by $X_j$ and 0 otherwise. For simplicity, we set the importance score $\hat{I}(X_i)$ to 0 for all $X_i$. Then $RS(\mathcal{D}_{core})=\lambda R (\mathcal{D}_{core})$ is equivalent to the total number of covered items. If $RS(\mathcal{D}_{core})$ can be solved exactly in polynomial time, then MCP can also be solved in polynomial time. Since MCP is known to be NP-hard, our selection objective $RS(\mathcal{D}_{core})$ is also NP-hard.
\end{proof}

\subsection{Proof of Theorem~\ref{thm:greedy}}\label{proof4-5}
\begin{proof}
We first prove the following lemma.
\begin{lemma}\label{lemma:is}
The $\Delta RS(X_i|\mathcal{D}_{core}^{k})$ is monotone increasing and submodular.
\end{lemma}

\begin{proof}
We prove that the $\Delta RS(X_i|\mathcal{D}_{core}^{k})$ is monotone and submodular as follows.
\begin{itemize}[leftmargin=10pt]
\item {\textbf{Monotone increasing}: }
Given any $\mathcal{D}_{core}^{k}$ and data point $X_i$, we can have $RS(\mathcal{D}_{core}^{k}\cup \{X_i\})-RS(\mathcal{D}_{core}^{k}) \geq 0$, since the representation score and the importance score are both non-negative for a new $X_i$ for the selected coreset $\mathcal{D}_{core}^{k}$.
Thus, $\Delta RS(X_i|\mathcal{D}_{core}^{k})$  is
monotone increasing.
\item {\textbf{Submodularity}: }
Given  $\tilde{\mathcal{D}}_{core}^{\prime k} \subset \mathcal{D}_{core}^{\prime k}$, $X_i \in \mathcal{D}_{train}$, and $X_i \notin \tilde{\mathcal{D}_{core}^{\prime k}}, \mathcal{D}_{core}^{\prime k}$,
we can obtain:
\begin{align}
&\Delta  RS(X_i|\mathcal{D}_{core}^{\prime k})= \lambda (R(\mathcal{D}_{core}^{\prime k}\cup\{X_i\})  \nonumber \\
&-R(\mathcal{D}_{core}^{\prime k})) + (1-\lambda)\hat{I}(\{X_i\}) \nonumber\\
&\Delta  RS(X_i|\tilde{\mathcal{D}}_{core}^{\prime k})=  \lambda (R(\tilde{\mathcal{D}}_{core}^{\prime k}\cup\{X_i\}) -R(\tilde{\mathcal{D}}_{core}^{\prime k})) \nonumber\\
& + (1-\lambda)\hat{I}(\{X_i\}) \nonumber
\end{align}
Since the marginal score satisfies the following inequality:
\begin{align}
\Delta  RS(X_i|\mathcal{D}_{core}^{k}) - \Delta  RS(X_i|\tilde{\mathcal{D}}_{core}^{\prime k}) \le 0
\end{align}
we can obtain the inequality as:

\begin{small}
\begin{align}
RS(\tilde{\mathcal{D}}_{core}^{\prime k} \cup {X_i}) -  RS(\tilde{\mathcal{D}}_{core}^{\prime k}) \ge  \nonumber  RS(\mathcal{D}_{core}^{\prime k} \cup {X_i}) -  RS(\mathcal{D}_{core}^{\prime k}) \nonumber
\end{align}
\end{small}

It demonstrates that $\Delta RS(X_i|\mathcal{D}_{core}^{k})$ is submodular.
\end{itemize}
\end{proof}
Lemma~\ref{lemma:is} indicates that $\Delta RS(X_i|\mathcal{D}_{core}^{k})$ is monotone increasing and submodular.

Suppose $RS(\mathcal{D}_{core}^{opt})$ denotes the optimal value of the objective for the coreset selection problem within budget $k_S$, we can derive:
\begin{equation} \label{eq:upperbound}
\Delta RS(X_i|\mathcal{D}_{core}^{k})\geq\frac{1}{k_S}\sum_{X_i\in \mathcal{D}_{core}^{opt}\backslash \mathcal{D}_{core}^{k}}\Delta RS(X_i|\mathcal{D}_{core}^{k})
\end{equation}

The Equation~\eqref{eq:upperbound} provides the upper bound for $RS(\mathcal{D}_{core}^{opt})$. By the inductive hypothesis,
\begin{equation} \label{eq:approximation}
RS(\mathcal{D}_{core}^{k}) \geq (1-(\frac{1}{k_S})^{k_S})RS(\mathcal{D}_{core}^{opt})
\end{equation}
The fraction $(\frac{1}{k_S})^{k_S})$ approaches $\frac{1}{e}$ as $k_S$ grows. Thus, the approximation ratio for the coreset selection algorithm is $1- \frac{1}{e}$.
\end{proof}

\subsection{Full Version and Proof of Theorem~\ref{thm:error-bound-short}}

\begin{theorem}[Bounded Error of the Selective Update]\label{thm:error-bound}
Consider the update round as defined in Equation~\eqref{eq:update}.
Let $I^\star(X_i)$ be the true importance score of $X_i$ at the current model parameter $\theta_t$. Define $\mathcal{N}_i^\star=\{j\in\mathcal{N}(X_i), \text{aff}(i,j)>\beta\}$ to be the set of neighbors of $X_i$ that pass the threshold in Equation~\eqref{eq:aff}, $r_i=\max_{j\in \mathcal{N}^\star_i}|\|\bar{\mathbf{H}_i}-\bar{\mathbf{H_j}}\|$ as the maximum distance between neighbors and $X_i$ in $k$-NN. We also define $\varepsilon$ as the LSH approximation error according to~\cite{10.1145/276698.276876}. The score error of $X_i$ is then defined as $e_i = I(X_i) - I^\star(X_i)$.
Then for every $X_i$, the following two properties hold:

\noindent (1) \textbf{Change Stability}: The per-round change of any node is bounded by:
\begin{align*}
& \big| I^{new}(X_i) - I(X_i) \big| \nonumber \leq \frac{1}{2}\max\{0,\max_{X_j \in \mathcal{N}_i^*}   \big|I^*(X_j) - I(X_i)\big|\}
\end{align*}

\noindent (2) \textbf{Local Approximation Error bound}: Let the current estimation error of node $X_i$ before the update be   $e_i^{old} = I(X_i) - I^\star(X_i)$ and after the update be $e_i^{new} = I^{new}(X_i) - I^\star(X_i)$. Then we have:
\begin{equation*}
\big| e_i^{new} \big| \leq \frac{1}{2}\big| e_i^{old} \big| + \frac{1}{2} (L_hr_i+\varepsilon),
\end{equation*}
\end{theorem}

\begin{proof}
(1) is directly obtained from the update rule in Equation~\eqref{eq:update}.
For (2), starting from Equation~\eqref{eq:update}, if we subtract $I^\star(X_i)$ from both sides, we obtain $e_i^{\mathrm{new}}=\frac{1}{2} e_i^{\mathrm{old}}+\frac{1}{2}\sum_{X_j\in\mathcal{N}_i^\star}w_{ij}\,\big(I^\star(X_j)-I^\star(X_i)\big)$. By local Lipschitz continuity
we have $|I^\star(X_j)-I^\star(X_i)|\le L_h\|\bar{\mathbf{H}}_j-\bar{\mathbf{H}}_i\|\le L_h r_i$ for every similar neighbor $X_j\in\mathcal{N}_i^\star$ when exact neighbors are used. When neighborhoods are retrieved via LSH, the use of approximate candidates perturbs the convex combination by at most $\varepsilon_{\mathrm{lsh}}$ in magnitude. Thus, we have $\big|\sum_{j}w_{ij}(I^\star(X_j)-I^\star(X_i))\big|\le L_h r_i+\varepsilon_{\mathrm{lsh}}$. Combining these inequalities with the triangle inequality, we obtain the stated bound.
\end{proof}

\end{document}